\documentclass[11pt]{article}
\usepackage{amssymb, latexsym, mathtools, amsthm}
\begingroup
    \makeatletter
    \@for\theoremstyle:=definition,remark,plain\do{%
        \expandafter\g@addto@macro\csname th@\theoremstyle\endcsname{%
            \addtolength\thm@preskip\parskip
            }%
        }
\endgroup
\usepackage{pgf,tikz}
\usepackage{pdfsync}
\usepackage{txfonts}
\usepackage[T1]{fontenc}
\usepackage{booktabs}
 \usepackage{graphicx}
\definecolor{dnrbl}{rgb}{0,0,0.3}
\definecolor{dnrgr}{rgb}{0,0.3,0}
\definecolor{dnrre}{rgb}{0.5,0,0}
\usepackage[colorlinks=true, citecolor=dnrgr, linkcolor=dnrre, urlcolor=dnrre] {hyperref}
\usepackage{xcolor}
\usepackage{parskip}
\usepackage{tabularx, colortbl}
\usepackage{geometry}
 \usepackage[scriptsize, up]{caption}
\usepackage{pgf,tikz}
\usetikzlibrary{decorations, decorations.pathmorphing}
\usetikzlibrary {shapes}
\newcommand*{\ditto}{------\raisebox{-0.5ex}{\textquotedbl}------\ \ }

\theoremstyle{plain}
\newtheorem{thm}{Theorem}[section]
\newtheorem{prop}[thm]{Proposition}
\newtheorem{lem}[thm]{Lemma}
\newtheorem{coro}[thm]{Corollary}
\newtheorem{defi}[thm]{Definition}
\numberwithin{equation}{subsection}
\makeatletter

\let\c@table\c@figure
\makeatother

\newcommand{\Nat}{\mathbb{N}}

\newcommand{\restr}{\upharpoonright}  
\DeclarePairedDelimiter{\ceil}{\lceil}{\rceil}
\DeclarePairedDelimiter{\floor}{\lfloor}{\rfloor}

\newcommand{\bigo}[1]{\mathop{\bf O}\/\left({#1}\right)}
\newcommand{\smo}[1]{\mathop{\bf o}\/\left({#1}\right)}
\newcommand{\ZZ}{\mathbf{Z}}
\newcommand{\DD}{\mathbf{D}}
\newcommand{\GG}{\mathbf{G}}

\newcommand{\YY}{\mathbf{Y}}
\newcommand{\CC}{\mathbf{C}}

\newcommand{\LL}{\mathbf{L}}

\newcommand{\JJ}{\mathbf{J}}
\newcommand{\pp}{\mathbf{p}}

\newcommand{\bias}[1]{\mathtt{B}(#1)}
\newcommand{\mix}{\textsc{mix}}
\newcommand{\unhap}{\mathtt{U}}

\newcommand{\blockb}{\mathtt{k}_{\beta}}
\newcommand{\Punhap}{\mathbb{P}_{\textrm{unhap}}}
\newcommand{\Pstab}{\mathbb{P}_{\textrm{stab}}}

\DeclareRobustCommand{\expe}[2][{\mbox{
$\mathbb{E}$}}]{\ensuremath {#1}\left[ {#2} \right]}
\DeclareRobustCommand{\expec}[3][{\mbox{
$\mathbb{E}$}}]{\ensuremath {#1}\left[ {#2} \ \big|\  {#3} \right]}
\DeclareRobustCommand{\expeco}[4][{\mbox{
$\mathbb{E}$}}]{\ensuremath {#1}_{#4}\left[{#2} \ \big|\  {#3} \right]}
\DeclareRobustCommand{\proba}[2][{\mbox{
$\mathbb{P}$}}]{\ensuremath {#1} [ {#2} ]}
\DeclareRobustCommand{\probac}[3][{\mbox{
$\mathbb{P}$}}]{\ensuremath {#1}[ {#2} \ |\  {#3} ]}

\renewenvironment{abstract}
 { \normalsize
  \list{}{
    \setlength{\leftmargin}{.0cm}%
    \setlength{\rightmargin}{\leftmargin}%
    }%
  \item {\bf \abstractname.} \relax}
 {\endlist}
\title{Minority
population in the one-dimensional 
Schelling  model of segregation
\thanks{Authors are listed alphabetically. 
Barmpalias was supported by the 
1000 Talents Program for Young Scholars from the Chinese Government,
and the Chinese Academy of Sciences (CAS) President's International 
Fellowship Initiative No. 2010Y2GB03.
Additional support was received by
the CAS and the Institute of Software of the CAS.
Partial support was also received from a Marsden grant of New Zealand 
and the China Basic Research Program (973) grant No.~2014CB340302.
Lewis-Pye was supported by a Royal Society University 
Research Fellowship. We would like to thank Pan Peng and Zhang Wei 
from the Institute of Software of the CAS, for helpful discussions.}}

\author{George Barmpalias  \and Richard Elwes \and	Andy Lewis-Pye}
\date{\today}
\begin{document}

\maketitle

\begin{abstract}
The Schelling model of segregation looks to explain the way in which a population of agents or particles of two types may come to organise itself into large homogeneous clusters, and can be seen as a variant of the Ising model in which the system is subjected to rapid cooling.  
While the model has been very extensively studied, the unperturbed (noiseless) version has largely resisted rigorous analysis, with most results in the literature pertaining to versions of the model in which noise is introduced into the dynamics so as to make it amenable to standard techniques from statistical mechanics or stochastic evolutionary game theory. 

We rigorously analyse the one-dimensional version of the model in 
which one of the two types is in the minority, and establish various forms of threshold behaviour. 
 Our  results are in sharp contrast with the case when the distribution of the two types
is uniform (i.e.\ each agent has equal chance of being of each type in the initial configuration), which was studied in
\cite{brandt:an, BELschel13}. 
\end{abstract}
\vspace*{\fill}
\noindent{\bf George Barmpalias}\\[0.5em]
\noindent
State Key Lab of Computer Science, 
Institute of Software, Chinese Academy of Sciences, Beijing, China.
School of Mathematics, Statistics and Operations Research,
Victoria University of Wellington, New Zealand.\\[0.2em] 
\textit{E-mail:} \texttt{\textcolor{dnrgr}{barmpalias@gmail.com}}.
\textit{Web:} \texttt{\textcolor{dnrre}{http://barmpalias.net}}\par
\addvspace{\medskipamount}
\noindent{\bf Richard Elwes}\\[0.5em]
\noindent School of Mathematics,
University of Leeds, LS2 9JT Leeds, United Kingdom.\\[0.2em]
\textit{E-mail:} \texttt{\textcolor{dnrgr}{r.h.elwes@leeds.ac.uk.}}
\textit{Web:} \texttt{\textcolor{dnrre}{http://richardelwes.co.uk}}\par
\addvspace{\medskipamount}
\noindent{\bf Andy Lewis-Pye}\\[0.5em]  
\noindent Department of Mathematics,
Columbia House, London School of Economics, 
Houghton Street, London, WC2A 2AE, United Kingdom.\\[0.2em]
\textit{E-mail:} \texttt{\textcolor{dnrgr}{A.Lewis7@lse.ac.uk.}}
\textit{Web:} \texttt{\textcolor{dnrre}{http://aemlewis.co.uk}} 
\vfill 
\thispagestyle{empty}
\clearpage
\section{Introduction}
The economist Thomas Schelling introduced his model of segregation in 
 \cite{TS1} (developed later in \cite{TS71a, TS71b}), with the explicit intention of 
 explaining the phenomenon of
racial segregation in large cities. 
Perhaps the earliest 
agent-based model studied by economists, since then it has become 
an archetype of agent-based modelling, prominently
featuring in libraries of modelling software tools such as NetLogo \cite{NetLogo}
and often being the subject of 
experimental analysis and simulations in the modeling and AI communities 
\cite{MF, howtop, FMpref, Gilbert14052002, Martijn2007, yilmaz2009agent, heppenstall2011agent, 
deSmith:2007:GAC:1557282, epstein1996growing}.
Many versions of the model have been analysed theoretically, from
a number of different viewpoints and disciplines: statistical mechanics
\cite{DM, RevModPhys.81.591} and \cite[Section 3.1]{Bertin}, 
evolutionary game theory \cite{HY, JZ1, JZ2, JZ3}
the social sciences \cite{CF, Clarkdem, SandSDo}, and more recently
computer science and AI \cite{CACP07, DBLP:conf/soda/BhaktaMR14, brandt:an, BELschel13}. 
It was observed in \cite{brandt:an}, however,  that despite the vast amount of work that has been
done on the Schelling model in the last 40 years, rigorous mathematical analyses in the previous  literature generally 
concern altered versions of the model, in which noise is introduced in the dynamics, i.e.\ where
one allows that agents may make non-rational decisions that are
detrimental to their welfare with small probability.
The introduction of such `perturbations' may be justifiable from a
`bounded rationality' standpoint. 

The model (which will be formally defined shortly) concerns a population of agents arranged geographically, each being of one of two types. Each agent has a certain neighbourhood around them that they are concerned with, and also an intolerance parameter $\tau\in [0,1]$ which we shall assume here to be the same for all agents. An agent's behaviour is dictated by the proportion of the agents in their neighbourhood which are of its own type. So long as  this proportion is $\geq \tau$ the agent may be considered `happy' and will not move.  Starting with a random configuration,  one then considers a discrete time dynamical process.  At each stage unhappy agents  may be  given the opportunity to move, swapping positions with another agent,  so as to increase the proportion of their own type within their neighbourhood. Now one might  justify  a perturbed version of these dynamics,  in which  agents will occasionally move in such a way as to decrease their utility (i.e.\  the  proportion of their own type within their neighbourhood) by arguing, for example, that it is reasonable to suppose that only incomplete information about the make-up of each neighbourhood is available to the agents. 
It is a fact, however,  that 
\begin{itemize}
\item[(a)] the methods used for the analysis of
the perturbed models do not apply to the unperturbed model; 
\item[(b)] the segregation that occurs in the perturbed models is often very
different than in the unperturbed model.
\end{itemize}
In the unperturbed models the underlying 
Markov chain does not have the regularities that are found in the perturbed case
(e.g.\ the Markov process is irreversible).  The presence of a large variety of absorbing states means that entirely different and more combinatorial methods are now required. Beyond the basic aim of a rigorous analysis for these unperturbed models, which have been so extensively studied via simulations, further motivation is provided by the fact that the Schelling model is part of a large family of models, arising in a broad variety of contexts---spin glass models, Hopfield nets, cascading phenomena as studied by those in the networks community---all of which look to understand the discrete time dynamics of competing populations on underlying network structures of one kind or another,  and for many of which the unperturbed dynamics are of significant interest.  The hope is that techniques developed in analysing  unperturbed Schelling segregation may pave the way for similar analyses in these variants of the model.  

The first rigorous analysis of an 
unperturbed Schelling model was described 
by  Brandt, Immorlica, Kamath, and Kleinberg
in \cite{brandt:an}. 
In this work it was also demonstrated that the eventual state of the 
process differs significantly from the stochastically stable states of the perturbed models.
This study focused on the one-dimensional Schelling model 
and provided an asymptotic analysis, in the sense that the results hold 
with arbitrarily high probability
for all sufficiently
large neighbourhoods and population.
More significantly, however, it dealt only with the symmetric case where
intolerance parameter $\tau=0.5$  (i.e.\ an agent is happy when at least 50\% of the agents
in its neighbourhood are of its own type). 
 In  \cite{BELschel13}
a much more general analysis of the unperturbed one-dimensional Schelling model
for $\tau\in [0,1]$ was provided. In fact it was shown there that various forms of surprising threshold behaviour exist. 
A significant symmetry assumption underlying 
the results in \cite{brandt:an, BELschel13}
is that
the populations of the two types of agents are assumed to be uniform 
(i.e.\ each agent has equal chance of being of each type in the initial configuration). 
Indeed, there is no rigorous study of the unperturbed spacial proximity model with swapping agents
for the rather realistic  case where the distribution of the two types of agents is skewed.
In fact, the question as to what type of segregation occurs with a skewed population distribution
was raised by  Brandt, Immorlica, Kamath, and Kleinberg 
in \cite[Section 4]{brandt:an} as well as in popular expositions
of the Schelling model like \cite{Hayes}.

The purpose of the present work is to give an answer to this question.
We show that complete segregation is the likely outcome if and only if the intolerance
parameter is larger than $0.5$. Moreover in the case that the minority type is at most 
25\%, there is a dichotomy between complete segregation and almost complete
absence of segregation.

\begin{table}
\caption{Parameters of the Schelling model and the main result.}\label{ta:paraandmainres}
\colorbox{black!10}{\arrayrulecolor{green!50!black}
  \begin{tabular}{lcc}
{\bf\small  Parameter}&  
{\bf\small  Symbol} &{\bf\small  Range} \\[1ex]
\toprule
{\small Population} \hspace{0.0cm}   &	 {\small $n$} \hspace{0.0cm} &{\small $\mathbb{N}$}\\[1ex]
{\small Neighbourhood radius}    \hspace{0.0cm} & {\small $w$} \hspace{0.0cm} &{\small $[0,n]$}\\[1ex]
 {\small Tolerance threshold}	 \hspace{0.0cm}  &   {\small $\tau$} \hspace{0.0cm} &{\small $[0,1]$}\\[1ex]
{\small Expected/Actual minority proportion}  \hspace{0.0cm}  &  {\small $\rho$/$\rho_{\ast}$}
\hspace{0.0cm}  & {\small $[0, 1]$}\\[1ex]
\end{tabular}}
\quad
\colorbox{black!10}{
 \begin{tabular}{rclcl}
\multicolumn{3}{c}{\bf\small  Process parameters}  & & {\bf\small  Segregation}  \\[1ex]
\toprule
{\small $\tau< \lambda_0$ \hspace{0.1cm}}&{\small\&\hspace{0.2cm}}&
{\small $\rho< \lambda_0$} & \hspace{0.0cm}  & {\small Negligible}\\[1ex]
{\small $\tau\leq \kappa_0$ \hspace{0.1cm}}&{\small\&\hspace{0.2cm}}&
{\small $\rho< 0.5$} & \hspace{0.0cm}  & {\small Negligible}\\[1ex]
{\small $\tau\leq 0.5$}\hspace{0.1cm} &{\small\&\hspace{0.2cm}}& 
{\small $\rho\leq 0.25$} &	\hspace{0.0cm}  & {\small Negligible}\\[1ex]
{\small $\tau> 0.5$}\hspace{0.1cm} &{\small\&\hspace{0.2cm}}&
{\small $\rho \leq 0.5$} &	\hspace{0.0cm}  & {\small Complete}
\end{tabular}} 
\centering
\end{table}

\subsection{Definition of the model}\label{def:modseg}
Schelling's model of residential segregation belongs to a large family of agent-based models, where
a system of 
competitive agents perform actions in order to increase their personal welfare, while possibly
decreasing the welfare of
other individuals.
This phenomenon roughly corresponds to the so-called {\em spontaneous order 
approach}\footnote{This contrasts the  {\em mechanism design approach} 
which studies the exogenous (a priori) design 
of regulations in order to achieve desired properties in a system of interacting agents.}
in economics literature, which studies the emergence of norms from the endogenous agreements among rational individuals.

The Schelling model that we study is a direct generalisation of that in
\cite{brandt:an} and also that studied by the authors in  \cite{BELschel13}.  
The one-dimensional  model with parameters
$n, w, \tau, \rho$ (as listed in 
Table \ref{ta:paraandmainres}) is defined as follows.
We consider $n$ 
individuals which occupy an equal number of {\em sites} $0,\dots, n-1$ 
(ordered clockwise) on a circle.  Each of the individuals belongs to one of the
two types $\alpha$ and $\beta$. 
The  type assignment of individuals is independent and identically distributed (i.i.d.), with each individual having probability $\rho$
of being type $\beta$. Without loss of generality  
we always assume that $\rho\leq 0.5$, i.e.\ that the individuals of
type $\beta$ are the expected minority (so long as $\rho\neq 0.5$).
This random type assignment takes place at stage 0 of the process, and defines the \emph{initial state}. At the end of stage 0,  we let
$\rho_{\ast}$ be the {\em actual} proportion of the individuals that are of type
$\beta$.

Unless stated otherwise, addition and subtraction on indices for sites are performed modulo $n$. Given two sites $u, v$ in any configuration of the individuals on the circle, 
the interval $[u, v]$ consists of the individuals that occupy sites
between $u$ and $(u+v)\mod n$ (inclusive). For example, if  $0\leq v<u<n$ then 
we let $[u,v]$ denote the set of nodes  $[u,n-1] \cup [0,v]$ 
(while $[v,u]$ is, of course, understood in the standard way).
 When we talk about a particular configuration, we identify each individual
with the site it occupies, referring to both entities as a {\em node}.
The {\em neighbourhood} of node $u$ consists 
of the interval $[(u-w),\ \ (u+w)]$
where $w$ is a parameter of the model
that we call the (neighbourhood) {\em radius}. 
The {\em tolerance threshold}  $\tau\in (0,1)$ is another parameter of the model 
that reflects how tolerant a node is to nodes of different type in its
neighbourhood.
We say that a node is {\em happy} if the proportion of the nodes 
 in its neighbourhood which are of its own type is at least $\tau$.

 \begin{figure} 
\includegraphics[scale=0.26]{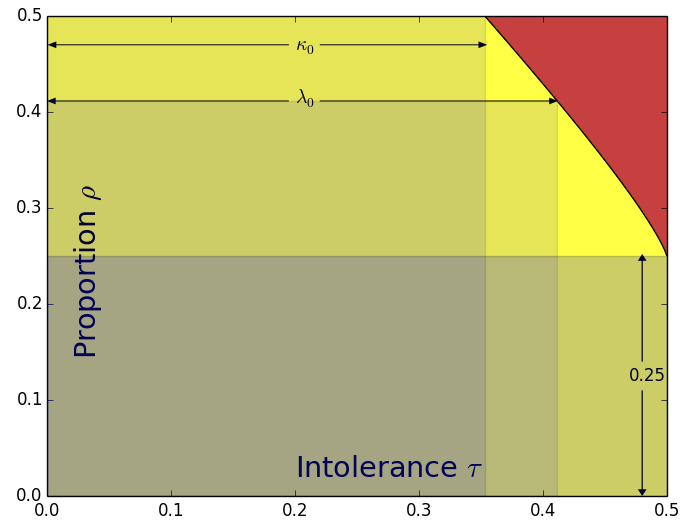}\hspace{0.5cm}
\includegraphics[scale=0.24]{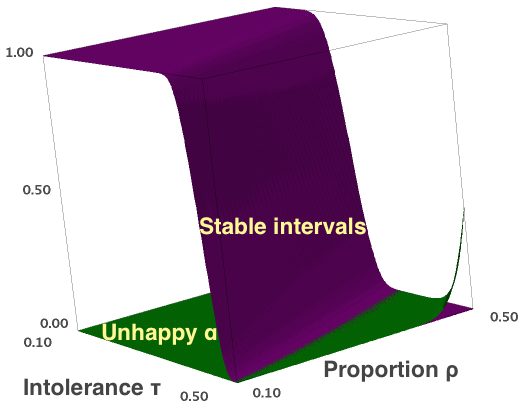}\hspace{0.5cm}
\includegraphics[scale=0.32]{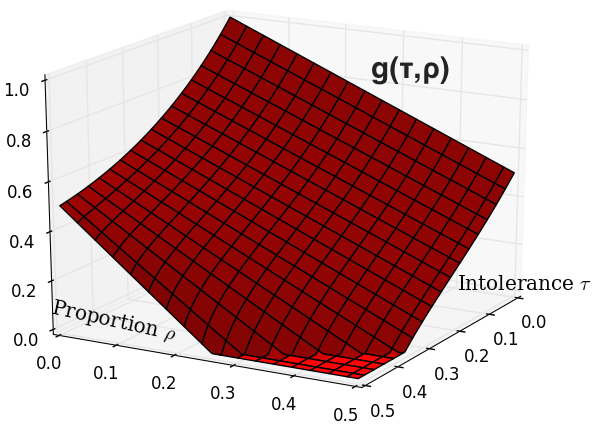}
\centering
 \caption{Threshold behaviour when $\tau,\rho$ are in $[0, 0.5]$. The two-dimensional axes  refer to $\tau$ and $\rho$.
 In the first figure, the process is static except for the $(\tau,\rho)$ in the small area at the top right corner. The second figure
 is a plot of $\Pstab$ and $\Punhap$ (for w=100) as functions of $(\tau,\rho)$. The third figure is a plot of 
 $g(\tau,\rho)$ for w=100.}\label{fig:red_3D_tr}
\end{figure}  

Given the initial type assignment (colouring) of the nodes, 
the {\em Schelling process} then evolves 
dynamically in stages as follows. 
At each stage $s>0$ we pick uniformly at random a pair of 
unhappy nodes of different type, and we swap them provided that
in both cases the number of nodes of the 
same type in the new neighbourhood is at least that in the original neighbourhood.
If at some stage there are no further legal swaps the process terminates. 
If at some stage all nodes of the same type are grouped into a single block, 
we say that at that stage we have {\em complete segregation}.

This completes the definition of the Schelling process with parameters $n,w,\tau$ and $\rho$, 
which we denote by the tuple $(n, w, \tau, \rho)$.   
The process can be seen as a Markov chain with $2^n$ states
corresponding to the configurations that we get by varying the
type of each node between $\alpha$ and $\beta$. 
A state is called {\em dormant} if either all $\alpha$-nodes are happy, or
all $\beta$-nodes are happy.
We shall be interested in the case that $w$ is large, 
and that $n$ is large compared to $w$.  
In this context it will turn out that the absorbing states of the Schelling process are
exactly the dormant states and, in fact, the only recurrence classes of the Schelling
process are the dormant states and 
complete segregation. 
Note that the number of nodes
of type $\alpha$ and of type $\beta$ 
does not change between transitions, 
once the initial state has been chosen.

\subsection{Our results}\label{se:ourresults}
Given the Schelling process $(n, w, \tau, \rho)$ we wish to determine with high
probability the
type of equilibrium that will eventually occur in the system.
Moreover, we are interested in asymptotic results, i.e.\ statements that
hold with arbitrarily high probability for all sufficiently large $w$ and all sufficiently large $n$ compared to $w$.
We denote this quantification on $w, n$ 
by `$0\ll w$ and $w\ll n$' respectively (and
write `$0\ll w\ll n$' for the combined statement). 
The following definition encapsulates the type of 
asymptotic statements about the Schelling process
$(n,w,\tau, \rho)$ that we are interested in establishing.

\begin{defi}[Properties with high probability and static processes]\label{de:prophighprob}\label{defi:staticproc}
Suppose that $R$ is a property which may or may not be satisfied by 
any given run of the Schelling process
$(n, w, \tau, \rho)$, and $T$ is a property of 
the parameters $\tau, \rho$. By the sentence
``if $T(\tau, \rho)$, then with high probability $R(n, w, \tau, \rho)$''
we mean that, provided that $\tau, \rho$ satisfy $T$, 
for every $\epsilon>0$ and
all $w\gg 1/\epsilon$, $n\gg w$ the process
$(n, w, \tau, \rho)$ satisfies $R$ with probability at least $1-\epsilon$.
We say that the process $(n,w,\tau,\rho)$ is static if, given $\epsilon>0$,
with high probability the number of
nodes that ever change their type 
in the entire duration of the process is $\leq\epsilon \cdot n$.
\end{defi}

By  \cite{brandt:an,BELschel13},
 the asymptotic behaviour of the process $(n,w,\tau,\rho)$
 is known for $\rho=0.5$ (except perhaps on the threshold
 $\tau=\kappa_0\approx 0.353$). The present work is dedicated to the case
 where one type of node is the minority, i.e.\ when $\rho<0.5$. 
We show that with probability 1 the process will either reach complete segregation
 or reach a dormant state. Complete segregation is, strictly speaking, 
 a a recurrence class of the process, consisting of the rotations of the two 
 blocks, one consisting of all the $\alpha$-nodes and the other consisting of all the $\beta$-nodes.
 Hence, modulo symmetries, we may regard complete segregation as
 an absorbing state. Dormant states are a different kind of absorbing state, as the process
 actually stops when it hits a dormant state.
We show that  when $\tau>0.5$ the highly probable outcome is complete segregation.
Moreover, in many cases when $\tau\leq 0.5$ the outcome is negligible segregation
(i.e.\ the process is static).
Let $\kappa_0\approx 0.353$ and $\lambda_0\approx  0.4115$
be 
the unique solutions of 
$\left( 0.5-x \right)^{0.5-x} = \left( 1-x \right)^{1-x}$
and
$2\tau \cdot \left(0.5-\tau\right)^{1-2\tau}= (1-\tau)^{2(1-\tau)}$
respectively in $[0,0.5]$.
 
\begin{thm}[Main result]\label{th:complsegm}
If $\tau>0.5$, $\rho<0.5$ and $\tau+\rho\neq 1$, then with high probability the
Schelling process $(n,w,\tau,\rho)$ 
reaches complete segregation.  
The process $(n, w, \tau, \rho)$ is static (with high probability) if
\[
\parbox{12cm}{\textup{
[$\tau\leq \lambda_0$ \& $\rho\leq \lambda_0$]\hspace{0.5cm}or\hspace{0.5cm}
[$\tau\leq \kappa_0$ \& $\rho<0.5$]\hspace{0.5cm}or\hspace{0.5cm}
[$\tau\leq 0.5$ \& $\rho\leq 0.25$]}
}
\]
or, more generally, if 
$2\rho \cdot (1-2\kappa_0)+\tau +\kappa_0<1$.
\end{thm} 

\begin{table}
\caption{Metrics of welfare and
critical stages in the unbalanced happiness 
process, as stopping times for certain conditions.}\label{ta:validprglobbouatd}
\colorbox{black!10}{\arrayrulecolor{green!50!black}  \begin{tabular}{lcl}
    \multicolumn{1}{l}{\bf\small Metric} &
        \multicolumn{1}{c}{\bf\small Symbol}& 
    \multicolumn{1}{c}{\bf\small Dynamics} \\[1ex]
\toprule
{\small Social welfare}
&  {\small $\mathtt{V}$} &
{\small Positive (strictly  if $\tau\leq 0.5$)} \\[1ex]
{\small Mixing index} 
&  {\small $\mix$} &
 {\small Negative  (strictly  if $\tau\leq 0.5$)} \\[1ex]
  {\small No.\ of unhappy nodes} 
&  {\small $\unhap$} &
 {\small Approximately negative if $\tau\geq 0.5$} \\[1ex]
\ditto {\small $\alpha$-nodes}
&  {\small $\unhap_{\alpha}$}   
& {\small Ambiguous} 
\end{tabular}}
\quad
\colorbox{black!10}{\arrayrulecolor{green!50!black} \begin{tabular}{ll}
   \multicolumn{1}{c}{\bf\small Stage} &
    \multicolumn{1}{l}{\bf\small  Stopping time for}      \\[1ex]
 \toprule
{\small $T_g$} &{\small $\GG_s> \tau\rho\cdot n/(4w)$ }\\[1ex]
{\small $T_y$} & {\small $\YY_s\leq \GG_s$, $s< T_g$} \\[1ex]
{\small $T_{\textrm{mix}}$}  &{\small $\mix >n(w+1)\tau\rho_{\ast}$}\\[1ex]
{\small $T_{\textrm{stop}}$}  &{\small $\unhap_{\alpha}>0$}
\end{tabular} } 
\centering
\end{table}

The values of $(\tau,\rho)$ for which we show that the process is static, correspond to the yellow
area of the first diagram (or, equivalently, the collapsed part of the surface of the third diagram) 
of Figure \ref{fig:red_3D_tr}. 
The case when $\rho\leq 0.25$ presents a remarkable contrast as $\tau$ crosses the
boundary of $0.5$. In this case, when $\tau$ exceeds  the
threshold $0.5$, the process changes from static to the other extreme of 
complete segregation.

\begin{table}
\caption{Two cases for the process $(n, w, \tau, \rho)$ and the corresponding expectations of
the number of initially happy nodes.}\label{ta:validptwocasesnadcorr}
\colorbox{black!10}{\arrayrulecolor{green!50!black} 
  \begin{tabular}{llll}
    \multicolumn{1}{l}{\bf\small Case} &
        \multicolumn{1}{l}{\bf\small Condition}& 
        \multicolumn{1}{l}{\bf\small  \hspace{0.6cm}Happy $\alpha$}& 
    \multicolumn{1}{l}{\bf\small  \hspace{0.6cm}Happy  $\beta$} \\[0.5ex]
\toprule
{\small Balanced happiness\hspace{0.6cm}}
&  {\small $\tau +\rho>1$, $\tau>0.5$}  
&\hspace{0.6cm}{\small $n\cdot e^{-\Theta(w)}$} 
&\hspace{0.6cm}{\small $n\cdot  e^{-\Theta(w)}$} \\[1ex]
{\small Unbalanced happiness\hspace{0.6cm}} 
&  {\small $\tau +\rho<1$, $\tau>0.5$} 
&\hspace{0.6cm}{\small $n\cdot \left(1-e^{-\Theta(w)}\right)$}
& \hspace{0.6cm} {\small $n\cdot  e^{-\Theta(w)}$}  \\[0.1ex]
\end{tabular}}
\centering
\end{table}

\begin{coro}[Phase transition on $0.5$]\label{coro:drphasetans}
If $\rho\leq 0.25$, 
then with high probability the process $(n, w, \tau, \rho)$ 
\begin{itemize}
\item converges to complete segregation if $\tau>0.5$;
\item is static, if $\tau\leq 0.5$.
\end{itemize}
Moreover with high probability it 
reaches its final state in time
$\smo{n}$, if $\tau\leq 0.5$
and time $\Omega(n)$, if $\tau> 0.5$.
\end{coro}

We display these results in the second item of Table \ref{ta:paraandmainres}.
In Sections \ref{se:analyticsche}--\ref{se:casetaubhalf} we present the argument that proves these results.
This argument uses a number of smaller results which are stated without proof, and are the building blocks of
the proof of Theorem \ref{th:complsegm}. It is our intention that the reader gets a fairly good understanding
of our analysis in this part of the paper, without the burden of having to verify some of the more technical
parts of the proof. Section \ref{se:prelim} is an appendix with detailed proofs of all the facts that were used in
Sections \ref{se:analyticsche}--\ref{se:casetaubhalf}, and completes the proofs of Theorem \ref{th:complsegm} and
Corollary \ref{coro:drphasetans}.

Our proof of Theorem 
is nonuniform, and the analysis is roughly divided in the two cases displayed in Table \ref{ta:validptwocasesnadcorr}:
{\em balanced and unbalanced happiness}. Here {\em happiness} refers to the numbers of initially happy nodes of the
two types, and determines the dynamics that drives the process to an equilibrium. Of the two cases, 
{\em unbalanced happiness} is the most challenging to deal with, and the dynamics is driven by small number of unhappy
$\alpha$-nodes against the large number of unhappy $\beta$-nodes, which in fact is preserved throughout a significant part of the process.

\section{Metrics and reaching complete segregation}\label{se:analyticsche}
One of the most challenging problems in the analysis of the segregation process
is the large number of absorbing states. In order to understand which transitions are possible,
we use certain metrics that 
describe the current state.

\begin{figure}
\includegraphics[scale=0.32]{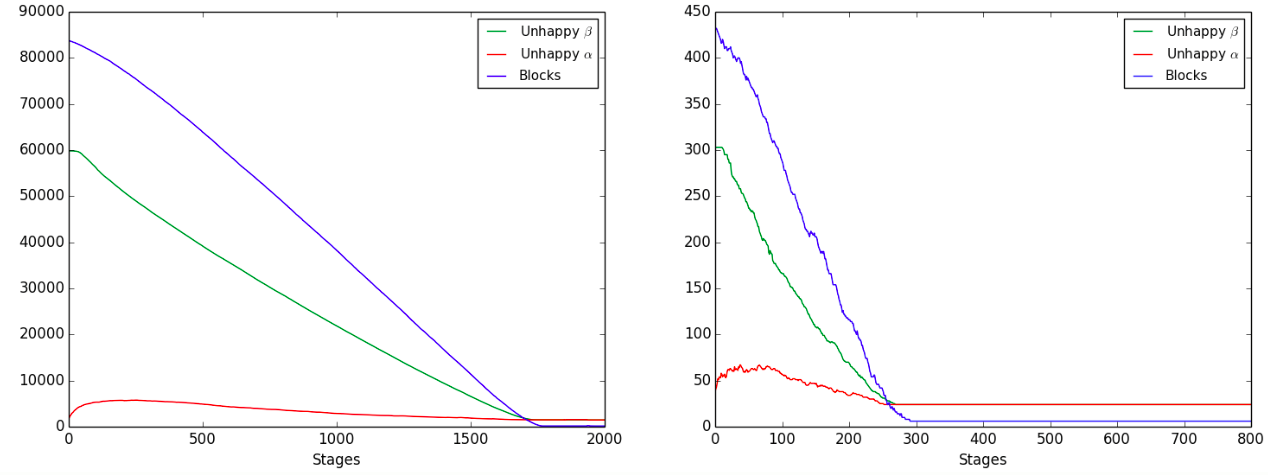}
\centering
 \caption{The first plot is
from the process $(200000, 50, 0.6, 0.3)$
and the second one 
from the process $(1000, 20, 0.6, 0.3)$.
These simulations illustrate that the number $\alpha$-nodes in the infected area remains bounded, until
the number of $\beta$-nodes outside the infected area becomes small. The second figure also illustrates
the fact that the number of unhappy nodes fluctuates locally.}\label{fig:inf_area}
\end{figure}

\subsection{Welfare, mixing, and expectations}
We define
global metrics that reflect the welfare of the entire population. An obvious choice
is the number of happy nodes at a given state. It is not hard to devise 
transitions of the process which reduce the total number of happy nodes 
(see the second plot of Figure \ref{fig:inf_area}).
However it is possible to show that if $\tau>0.5$ the
total number of happy nodes is {\em approximately non-decreasing} (in the sense that it is $\Theta(g)$ for some nondecreasing function $g$ 
on the stages, where the underlying constant depends only on $w$).
Let the {\em utility} of a node (at a certain state) be the number of nodes of the same type in its neighborhood.
A better behaved global metric of welfare 
of a state is the sum of the
utilities of the nodes in the state.
We call this parameter the {\em social welfare} 
of the state and denote it by $\mathtt{V}$. 
 A consequence of the transition rule and the definition of utility is
 that the {\em social welfare} 
 does not decrease along the stages of the process.
Furthermore, 
if $\tau\leq 0.5$, every transition of the process strictly increases the social welfare.
Let the {\em mixing index} of a node be the number of nodes in its neighbourhood
that are of different type.
The {\em mixing index} $\mix$ of a state is the sum of the mixing indices
of the $\alpha$-nodes in that state. 
The mixing index of a state is also equal to the
sum of the mixing indices
of the $\beta$-nodes in that state.
The relationship between the two metrics is
\[
\mathtt{V}= (2w+1)\cdot n - 2\cdot \mix.
\]
Hence the mixing index is non-increasing along the transitions.
Note that a single swap cannot decrease the mixing index by more than $4w$. On the other hand, 
by linearity of expectation we can calculate that
 \begin{equation*}
\parbox{12.5cm}{the expectation of the mixing index in the initial state of $(n,w,\tau,\rho)$ is
$2nw\rho(1-\rho)$.}
\end{equation*}
The mixing index of 
complete segregation (in nontrivial cases) is
$w(w+1)$.
Since $\rho\leq 1/2$, this means that (with high probability) the process can reach complete segregation only
after $(n\rho-(w+1))/4 >n\rho/5$ stages, i.e.\ $\Omega(n)$ stages. 
On the other hand, a case analysis shows that if $\tau\leq 0.5$, each step in the process
decreases the mixing index by at least 4.
This means that if $\tau\leq 0.5$ and the process is static, then it reaches its final state within 
$\smo{n}$ stages. 
This happens because each time a swap occurs, the mixing index decreases by at least 4
(so its not possible that the same few nodes swap more than $\smo{n}$ times).
We have shown that the second clause of Corollary
\ref{coro:drphasetans} (concerning the time to the final state) follows from the first clause.

As another measure of mixing, we may consider the number
$\mathtt{k}_{\beta}$ of maximal $\beta$-blocks in the state. These are the contiguous $\beta$-blocks
that are maximal, in the sense that they cannot be extended to a larger contiguous $\beta$-block.
Let $\unhap$ be the number of unhappy nodes in a state.
It is not hard to show that if $\tau>0.5$  then $\mix=\Theta(\unhap)=\Theta(\mathtt{k}_{\beta})$ and in particular
\begin{equation}\label{eq:finalumixkrel}
\mix \leq w\cdot (w+1) \cdot \mathtt{k}_{\beta} \leq w\cdot (w+1) \cdot \unhap < \mix\cdot 2w/(1-\tau). 
\end{equation}
This means that 
the number of unhappy nodes at a certain state reflects
the progress of the process towards segregation.
More precisely,
the metrics $\mix$, $\blockb$, $\unhap$ are mutually proportional when
$\tau>0.5$, where the analogy coefficient depends on $w$ (see Figure \ref{fig:inf_area}). 
In Table \ref{ta:validprglobbouatd} we display these global metrics of welfare,
along with their dynamics. A function (on the stages of the process) 
has positive dynamics if it is non-decreasing
and approximately positive dynamics if it
is $\Theta(g)$ for some nondecreasing function $g$, where the multiplicative constant does not depend on $n$.
Similar definitions apply for `negative'.
The first clause of Theorem \ref{th:complsegm} (the case when $\tau>0.5$)
is the hardest to prove. 
It turns out that in this case  we can
deduce a non-trivial lower bound on the mixing index of dormant states.
\begin{lem}[Mixing in dormant states]\label{le:mixindorsta}
Consider the process $(n, w, \tau, \rho)$ with $\tau>0.5$. The  
mixing index in a dormant state is more than $n(w+1)\tau\rho_{\ast}$,
as long as $w>1/(2\tau-1)$.
\end{lem}

The case $\tau>0.5$ is further divided in two cases, which reflect the
proportions of happy nodes in the initial state. We display these in 
Table \ref{ta:validptwocasesnadcorr}, along with the corresponding 
expectations for the numbers of happy nodes of each type.
Lemma \ref{le:mixindorsta} is crucial  
for the proof of the first clause of Theorem \ref{th:complsegm}
(in particular the  case $\tau+\rho<1$).

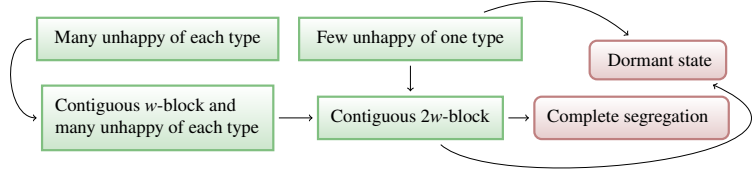
\begin{figure}
\scalebox{0.7}{
\begin{tikzpicture}[
main/.style={rectangle,  minimum size=6mm, rounded corners, very thick,
draw=red!50!black!50, top color=white, bottom color=red!50!black!20, font=\small},
Nnode/.style={rectangle,  minimum size=6mm, very thick,
draw=green!50!black!50, top color=white, bottom color=green!50!black!20, font=\small },
tnode/.style={rectangle,  minimum size=6mm, rounded corners, very thick,
draw=red!50!black!50, top color=white, bottom color=red!50!black!20, 
 font=\small},
 ttnode/.style={rectangle,  minimum size=6mm, rounded corners, very thick,
draw=yellow!50!black!50, top color=white, bottom color=yellow!50!black!20, 
 font=\small}]
\node (cr) [Nnode,  outer sep=3pt, inner sep=7pt] at (11.3,0.5) {\textrm{Contiguous $2w$-block }};
\node (ivr) [main, outer sep=3pt, inner sep=7pt] at (16,1.6) {\parbox{2.1cm}{\hspace{0.1cm}Dormant state\hspace{0.1cm}}};
\node (tivr) [main, outer sep=3pt, inner sep=7pt] at (15.5,0.5) {\parbox{3.2cm}{Complete segregation}};
\node (fvr) [Nnode, outer sep=3pt, inner sep=7pt] at (11.3,2) {\textrm{\ Few unhappy of one type\ \ }};
\node (kr) [Nnode, outer sep=3pt, inner sep=7pt] at (6.5,2) {\textrm{\ Many unhappy of each type\ \ }};
\node (bi) [Nnode, outer sep=3pt, inner sep=7pt] at (6.5,0.5) {\parbox{3.9cm}{Contiguous $w$-block and\\ 
many unhappy of each type}};
\draw [->] (cr) --  (tivr); 
\draw [->] (fvr) to [in=140, out = 20]  (ivr); 
\draw[->]  (cr) .. controls (13,-1) and (20, -0.2) .. (ivr); 
\draw [->] (bi) -- (cr) ; 
\draw [->] (fvr) -- (cr); 
\draw [->] (kr) to [in=180, out = -180] (bi); 
\end{tikzpicture}}
\centering
\caption{The path to a dormant state or complete segregation when $\tau>0.5$.}
\label{fig:skewstratcomd}
\end{figure}


\subsection{Accessibility of dormant states and complete segregation}\label{subse:accdorcose}
We show  the case of Theorem \ref{th:complsegm} where $\tau>0.5$ and
$\tau+\rho>1$. This argument consists of two parts. First, we show that in this case
with high probability the initial state is such that
every state with the same number of $\alpha$-nodes has unhappy nodes of both types (i.e.\ it is not dormant).
Hence under these conditions, no accessible state is dormant.
The second part consists of showing that from every state there is a sequence of transitions to either a dormant state 
or complete segregation. Moreover the latter fact holds in general, for any values of $\tau,\rho$, so it can be reused for the case when
$\tau+\rho<1$, in Section \ref{se:recomsehaca}. This latter case is more challenging, 
as it can be seen that there are permutations of the initial state which are dormant.
\begin{lem}[Existence of unhappy nodes]\label{coro:exiunhangen}
Suppose that $\tau>0.5$, $\rho_{\ast}<\tau$  and $w$ is sufficiently large.
Then for every $c\in\Nat$ and
all sufficiently large $n$, every state of the process
$(n, w, \tau, \rho)$ has more than $c$ unhappy $\beta$-nodes.
If in addition $\tau+\rho_{\ast}>1$, every state also has 
more than $c$ unhappy $\alpha$-nodes.
\end{lem}

Given $\rho$, by the law of large numbers with high probability 
(tending to 1, as $n$ tends to infinity)
$\rho_{\ast}$ will be arbitrarily close to $\rho$. 
Hence we may deduce the absence of dormant states
(with high probability) in the case that $\tau+\rho>1$.
\begin{coro}[Absence of dormant states when $\tau>0.5$ and $\tau+\rho>1$]\label{coro:absdormst.g}
If $\rho\leq 0.5<\tau$ and $\tau+\rho>1$  then
with high probability none of the accessible states of
the process $(n, w, \tau, \rho)$
is dormant.
\end{coro}

It remains to show the accessibility of either a dormant state or
complete segregation, from any state of the process.
An inductive argument can be used in order to prove this fact.

\begin{lem}[Complete segregation or dormant state]\label{coro:inevcompsegd}
From any state of the process $(n, w, \tau, \rho)$ with $0\ll w\ll n$
there exists a series of transitions to complete segregation or to a dormant state.
\end{lem}

Here is a sketch of the proof.
If $\tau\leq 0.5$ the mixing index is strictly decreasing through the transitions, so it is immediate
that the process will reach a dormant state (indeed, 0 is a lower bound for the mixing index).
For the case where $\tau>0.5$ (which we assume for the duration of this discussion) we can argue
inductively, in four steps. 
An interval of nodes of the same type is called a {\em contiguous block}. 
First we show that
from a stage with few unhappy nodes of one type (here $5w^4$ is a convenient upper bound of
what we mean by `few', which is by no means optimal) there is a series of transitions which lead to either
a state with a contiguous block of length $2w$ or a dormant state.
Second, 
from a state with a contiguous block of length $\geq 2w$
there is a series of transitions to complete segregation or to a dormant state.
Third,
any state which has at least $w^4$ unhappy nodes of each type,
there is a series of transitions to a state with a contiguous 
block of length at least $w$.
Finally 
from a state that has a contiguous block  of length $\geq w$
and at least $4w$ unhappy nodes of opposite type from the block,  there is a series of
transitions to a state with a contiguous block  of length $\geq 2w$.
The combination of these four statements constitutes a strategy for arriving to a dormant state or
a state of complete segregation, from any given state.
We illustrate this strategy in Figure \ref{fig:skewstratcomd}, where
two arrows leaving a node indicate that at least one of these routes are possible.

 \begin{figure} 
\centering\includegraphics[scale=0.36]{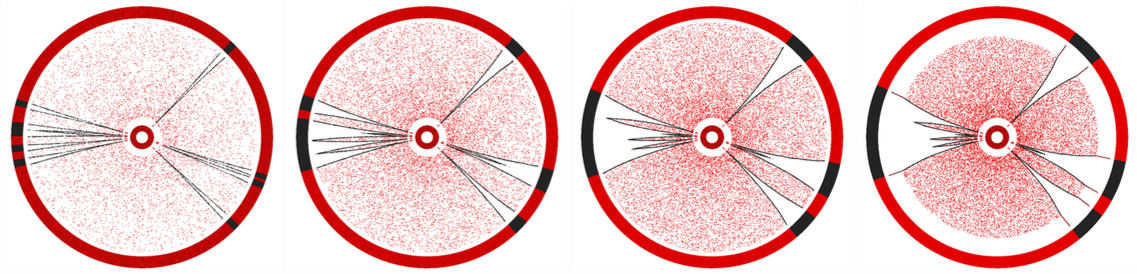}
\caption{The evolution of the infected area when $\tau+\rho<1$.}\label{fig:cicle_slides}
\end{figure}

\section{Reaching complete segregation \texorpdfstring{when $\tau>0.5$ and  $\tau+\rho<1$}{in the hard case}}\label{se:recomsehaca}
This case of Theorem \ref{th:complsegm} is challenging because we need to show that the process avoids accessible dormant states, until it reaches a {\em safe state}
i.e.\ a state from which no dormant state is accessible.
The reason for this avoidance is (in contrast with the case $\tau+\rho>1$ of Section \ref{subse:accdorcose}) 
the dynamics of the process with the given parameters.
The methodology we use is based on a martingale argument, which involves a great deal of the analytical tools 
(e.g.\ the metrics of social welfare) and their properties that were developed in the previous sections.
Having shown that dormant states are avoided until the process reaches a safe state, Lemma \ref{coro:inevcompsegd} gives
Theorem \ref{th:complsegm} (for the case where $\tau>0.5$ and  $\tau+\rho<1$). An overview of this argument is given in Figure \ref{fig:dysimpleskads}.

\subsection{The persistence of large contiguous \texorpdfstring{$\beta$}{beta}-blocks}\label{se:persist}
According to our plan, we wish to establish the existence of unhappy nodes of both types until a safe state is reached.
By Lemma \ref{coro:exiunhangen}, we do not have to worry about the existence of unhappy $\beta$-nodes. One device that guaranties the  existence of
unhappy $\alpha$-nodes is a contiguous block of $\beta$-nodes, of length at least $w$. 
Such a block exists in the initial random state (with high probability). One way to argue for its preservation in subsequent stages is to consider the 
ratio of the unhappy nodes of the two types. 
Even more relevant is the ratio between the number of unhappy $\alpha$-nodes, 
and the number of $\beta$-nodes which are not just unhappy, but actually sufficiently unhappy that they can swap with any unhappy $\alpha$-node.

\begin{defi}[Very unhappy $\beta$-nodes]\label{de:veryunhapy}
Given a stage of the process, a node of type $\beta$ is very unhappy if
there are at least $(2w+1)\tau$ nodes of type $\alpha$ in its neighbourhood.
The number of very unhappy $\beta$-nodes 
is denoted by $\unhap_{\beta}^{\ast}$.
\end{defi}

In the case that we study ($\tau>0.5$ and  $\tau+\rho<1$) 
initially, the number of very unhappy $\beta$-nodes is $\Omega(n)$ while the
number of unhappy $\alpha$-nodes is $\smo{n}$.
The following lemma says that as long as this imbalance is preserved, it is very likely that a
sufficiently long contiguous block of $\beta$-nodes is preserved.
\begin{lem}[Persistent \texorpdfstring{$\beta$}{beta}-block]
\label{le:persistb}
Consider the process
$(n,w,\tau,\rho)$ with $\tau>0.5$  
and let $s_{\ast}$ be the least stage where
the ratio between the very unhappy $\beta$-nodes and the unhappy $\alpha$-nodes
becomes less than $4w^2$ (putting $s_{\ast}=\infty$ if no such stage exists). Then with 
high probability there is a $\beta$-block of length $\geq 2w$ at all
stages $<s_{\ast}$ of the process. 
\end{lem}

Since a $\beta$-block of length at least $w$ is a
guarantee for unhappy $\alpha$-nodes, we get the following corollary.

\begin{coro}[Conditional existence of unhappy $\alpha$-nodes]\label{le:percosistb}
Under the hypotheses of Lemma \ref{le:persistb}, 
with high probability 
there are unhappy $\alpha$-nodes at 
all stages $<s_{\ast}$ of the process.
\end{coro}

It remains to construct an elaborate martingale argument in order to show that 
the imbalance between $\unhap_{\alpha}$ and $\unhap^{\ast}_{\beta}$
persists for a sufficiently long time (until the process reaches a safe state). 

\subsection{Infected area view of the Schelling process}\label{subse:infareaview}
In the case of unbalanced happiness (i.e.\ when $\tau>0.5$, $\tau+\rho<1$, see Table \ref{ta:validptwocasesnadcorr})
the unhappy $\alpha$-nodes are initially very rare, so the interesting activity (namely $\alpha$-to-$\beta$ swaps) occurs
in small intervals of the entire population (at least in the early stages). 
These intervals contain the unhappy $\alpha$-nodes, and gradually expand, while outside these intervals
all $\beta$-nodes are very unhappy. Figure \ref{fig:biasplot}  
shows the development of this process, where the height of the nodes (perpendicular lines) is proportional to
the number of $\alpha$-nodes in their neighborhood and the horizontal black line denotes the 
threshold where an $\alpha$-node  becomes unhappy.
Hence nodes with high proportion of $\alpha$-nodes in their neighbourhood
will be higher than the nodes with low proportion of $\alpha$-nodes in their neighbourhood. 
The three horizontal bars are snapshots of the process, and show
cascades forming, originating from the initially unhappy $\alpha$-nodes. 
Figure \ref{fig:cicle_slides} 
shows the same process, with the current state in the outer circle, and with swaps represented by a
dot at a distance from the center which is proportional to the stage where the swap occurred.
These cascades that spread the unhappy $\alpha$-nodes are due to the following domino effect. An unhappy 
$\alpha$-node moves out of a neighbourhood, thus reducing the number of $\alpha$-nodes in that interval. This
in turn often makes another $\alpha$-node in the interval unhappy, which can move out at a latter stage, thus
causing another $\alpha$-node nearby to be unhappy, and so on.
The expanding intervals are the {\em infected segments} which start their life as {\em incubators}.
For the sake of simplicity, we omit the formal definitions of these notions, which can be found in the appendix.
Roughly speaking,  incubators are a small intervals that surround  the unhappy $\alpha$-nodes in the initial state.
Moreover they are defined in such a way that, every $\beta$-node that is outside the incubators is {\em very unhappy}
in the initial state.
During the process, as we discussed above, these expand into larger {\em infected segments}, so that at each stage
every unhappy $\alpha$-node is inside an infected segment.
The union of all infected segments is called the 
{\em infected area}. At any stage, every $\beta$-node outside the infected area is {\em very unhappy}
and every $\alpha$-node outside the infected area is happy.
It is not hard to show that if $\tau+\rho<1$, the probability that a node belongs to an incubator is
$e^{-\Theta(w)}$. Hence with high probability  
 the number of incubators as well as the number of nodes belonging to incubators
 of the process $(n,w,\tau,\rho)$ is $n e^{-\Theta(w)}$.

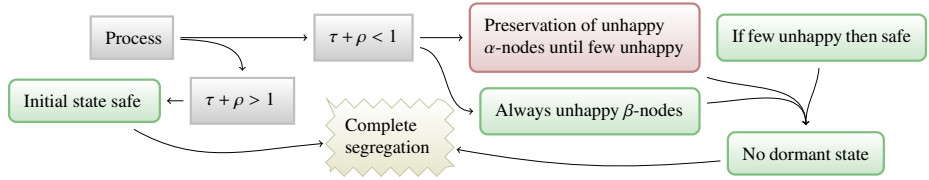
\begin{figure}
\scalebox{0.7}{
\begin{tikzpicture}[
Nnode/.style={rectangle, rounded corners, minimum height=4mm, very thick, 
draw=green!50!black!50, top color=white, bottom color=green!50!black!20, font=\small },
tnode/.style={rectangle, rounded corners, very thick,
draw=red!50!black!50, top color=white, bottom color=red!50!black!20, font=\small},
pnode/.style={rectangle, very thick,
draw=white!50!black!50, top color=white, bottom color=white!50!black!40, font=\small}]
 \node (Npreserv) [Nnode,  outer sep=4pt, inner sep=8pt] at (7.8,2) {\parbox{3.4cm}{
If few unhappy then safe}};
 \node (Nbreakth) [Nnode,  outer sep=4pt, inner sep=8pt] at (7.5,-0.2) {\parbox{2.4cm}{
 No dormant state}};
 \node (Nunbala) [tnode,  outer sep=4pt, inner sep=8pt] at (3.3,2) {\parbox{3.8cm}{Preservation
 of unhappy\\ $\alpha$-nodes until few unhappy}};
 \node (Nhard) [pnode,  outer sep=4pt, inner sep=8pt] at (-0.9,2) {\parbox{1.3cm}{$\tau+\rho<1$}};
 \node (Ninitsafe) [Nnode,  outer sep=4pt, inner sep=8pt] at (-6.2,0.8) {\parbox{2.3cm}{Initial state
 safe}};
 \node (Nalwunhab) [Nnode,  outer sep=4pt, inner sep=8pt] at (3.4,0.6) {\parbox{3.6cm}{Always 
unhappy $\beta$-nodes}};
 \node (Nprocess) [pnode,  outer sep=4pt, inner sep=8pt] at (-5.3,2) {\parbox{1cm}{Process}};
 \node (Neasy) [pnode,  outer sep=4pt, inner sep=8pt] at (-3.2,0.8) {\parbox{1.4cm}{$\tau+\rho>1$}};
\node (Ncomplete) [draw=yellow!50!black!50, font=\small, top color=white, bottom color=yellow!50!black!20, decorate, decoration=zigzag, 
outer sep=4pt, inner sep=8pt] at (-0.4,0.1) {\parbox{1.7cm}{Complete segregation}};
\draw [->, in = 90, out=-2] (Nprocess) to (Neasy);
\draw [->] (Neasy) to (Ninitsafe);
\draw [->] (Nprocess) to (Nhard);
\draw [->, in = 180, out=-10] (Nhard) to (Nalwunhab);
\draw [->] (Nhard) to (Nunbala);
\draw [->, in = 90, out=4] (Nalwunhab) to (Nbreakth);
\draw [->, in = 90, out=-95] (Npreserv) to (Nbreakth);
\draw [->, in = 90, out=-18] (Nunbala) to (Nbreakth);
\draw [->,  in = -10, out=185] (Nbreakth) to (Ncomplete);
\draw [->,  in = 185, out=-30] (Ninitsafe) to (Ncomplete);
\end{tikzpicture}}
\centering
\caption{The logic of the proof that if $\tau>0.5$, with high probability
the process reaches complete segregation.}\label{fig:dysimpleskads}
\end{figure}

It turns out that the number of unhappy $\beta$-nodes in an interval of nodes, is conveniently bounded
in terms of the number of $\alpha$-nodes in the interval.
This means that  if the number of $\alpha$-nodes in the infected area remains $\smo{n}$, then
the number of unhappy $\beta$-nodes in the infected area also remains $\smo{n}$.
In order to give a clear sketch of the argument depicted in Figure \ref{fig:dysimpleskads}
(for the current case when $\tau>0.5$ and $\tau+\rho<1$)
let us define the global variables in Table \ref{ta:pranvarinfarfa2}
(for the current discussion we will not be concerned with $\DD_s$ or its definition).
Note that Since $\mathtt{U}_s\leq \GG_s + \YY_s + \ZZ_s$.
A combinatorial argument can be used in order to show that $\YY_s\leq  \ZZ_s/(1-\tau)+2w\CC$.
Hence 
\begin{equation}\label{eq:uscsgszs}
\mathtt{U}_s\leq w\CC  + \GG_s + 2\ZZ_s /(1-\tau).
\end{equation}
By \eqref{eq:finalumixkrel}  we know that a stage where the number of unhappy nodes
is less than $n\tau\rho_{\ast}/w$ is a safe stage.
Hence we wish to show that (with high probability) 
the process will arrive at a stage where  each of the three
summands in \eqref{eq:uscsgszs} are at most
$n\tau\rho_{\ast}/(3w)$. We know that $\CC$ 
can be bounded appropriately. Our main argument will show
how to obtain a similar bound for $\ZZ_s$. 
Note that $\GG_s$ plays a different role,
since it is initially large and shrinks monotonically (as the infected area expands monotonically). 
In order to find a stage where $\GG_s$ becomes sufficiently small, it is instructive to consider what is a typical swap
in the process.
At the start of the process the infected area is a very small proportion of the
entire ring. The vast majority of unhappy $\beta$-nodes occur outside
the infected area, while all unhappy $\alpha$-nodes are inside the infected area.
It follows that with high probability a 
swap will involve an $\alpha$-node in the infected
area and a $\beta$-node outside the infected area. A {\em bogus swap} is a swap is one that is not of this kind.

\begin{defi}[Bogus swaps]
A swap
which involves a $\beta$-node 
currently inside the infected area is called bogus. 
Given an infected segment $I$, a bogus swap in $I$ is a
swap that moves an $\alpha$-node into $I$.
\end{defi}

\begin{table}
\caption{Random variables indicating the number of certain nodes in 
infected area at stage $s$ of the 
process.}\label{ta:pranvarinfarfa2}
\colorbox{black!10}{\arrayrulecolor{green!50!black} 
\begin{tabular}{cl}\toprule
{\bf\small $\mathbf{Z}_s$}  &   \textrm{\small $\alpha$-nodes in infected area}\\[1ex]
{\bf\small $\YY_s$}  &   \textrm{\small Unhappy $\beta$-nodes in infected area}\\[1ex]
{\bf\small $\DD_s$}  &   \textrm{\small Anomalous nodes in infected area}  \\[1ex]
\bottomrule 
\end{tabular} }
\quad
\colorbox{black!10}{\arrayrulecolor{green!50!black} 
\begin{tabular}{cl}\toprule
{\bf\small $\pp_s$}  &   \textrm{\small Probability of a bogus swap}\\[1ex]
{\bf\small $\GG_s$}  &   \textrm{\small $\beta$-nodes outside infected area.}  \\[1ex]
{\bf\small $\CC$}     & \textrm{\small Nodes inside the incubators}\\[1ex]  
\bottomrule
\end{tabular} } 
\centering
\end{table}

Note that any swap which is not bogus, reduces $\GG_s$ by at least 1. Hence if we show that the bogus swaps have
small probability throughout a significant part of the process, we can ensure that $\GG_s$ becomes sufficiently small.
In order to be more precise,
recall the stopping time $s_{\ast}$
from Lemma \ref{le:persistb}.
We introduce a few more stopping times, all of which 
will turn out to be earlier than $s_{\ast}$ (with high probability). These  
basically concern the satisfaction of conditions which will ensure that 
the mixing index is sufficiently low as to guarantee a safe state. 
By \eqref{eq:finalumixkrel} 
we have $\mix\leq \mathtt{U}\cdot w(w+1)$
and in order to ensure a safe state (by Lemma \ref{le:mixindorsta}) 
we want $\mix<n(w+1)\tau\rho_{\ast}$. 
So we want $\mathtt{U}<n\tau\rho_{\ast}/w$ at some stage of the process.
Let $T_{\textrm{mix}}$ be 
the first stage which satisfies this condition.
Similarly, consider the stopping times $T_g, T_{\textrm{stop}}$ of 
the second part of Table \ref{ta:validprglobbouatd} (for simplicity, we will not consider $T_y$ in the present discussion).
We use an elaborate martingale argument in order to show the following.

\begin{lem}[Bounding the $\alpha$-nodes in the infected area]\label{le:inthzbouwe}
If $\tau>0.5$ and $\tau+\rho<1$, 
with high probability we have $\ZZ_s=\smo{n}$ and $\pp_s=\smo{1}$ for all $s<T_g$.
\end{lem}

This lemma in combination with Corollary \ref{le:persistb} 
implies that $T_g\leq s_{\ast}\leq T_{\textrm{stop}}$.
Hence every stage up to $T_g$ involves a swap.
Then it follows from the second clause of Lemma \ref{le:inthzbouwe}  
that $T_g<n$ (since $\GG_s$ is reduced by at least 1 
at every non-bogus swap).
Hence by \eqref{eq:uscsgszs} we have established (with high probability) the existence of a stage $T_g<n$ such that

\begin{equation*}
\unhap_{T_g}\leq w\CC  + \GG_{T_g} + \frac{2\ZZ_{T_g}}{1-\tau} 
\leq \smo{n} + \frac{n\tau\rho_{\ast}}{4w}+ \smo{n} <\frac{n\tau\rho_{\ast}}{w}.
\end{equation*}

Hence  by \eqref{eq:finalumixkrel} we have $T_{\textrm{mix}}\leq T_g$,
which means that by stage $T_g$ a safe state has been reached.
Then by Corollary \ref{coro:inevcompsegd} the process will reach complete segregation, with probability
$1-\smo{1}$.

\begin{coro} [Safe state arrival]
Suppose that $\tau+\rho<1$. Then with high probability
the process $(n,w,\tau,\rho)$  reaches a safe state by stage $n$, and eventually complete segregation.
\end{coro}

This argument (with the full details given in Section \ref{se:prelim})
concludes the proof of Theorem \ref{th:complsegm} for the case $\tau>0.5$.
It remains to deal with the case $\tau\leq0.5$.

\begin{table}
\caption{Likelihood of various properties 
in the initial configuration under certain conditions, when $\rho\leq 0.5$ and $\tau \leq 0.5$}
\label{ta:vlikelyghd}
\colorbox{black!10}{\arrayrulecolor{green!50!black} 
 \begin{tabular}{ccccc}
 {\bf\small Property}&  & {\bf\small Probability} &  {\bf\small Distribution}& {\bf\small Likelihood}\\[1ex]
\toprule
{\small Stable $\alpha$-interval} &	\hspace{0.3cm} & {\small $\Pstab$} &{\small $Z_{\textrm{stable}}\sim B(w, 1-\rho)$}\hspace{0.5cm}  &
{\small high if $2\tau +\rho<1$, low if $2\tau +\rho >1$}\\[1ex]
{\small Unhappy $\alpha$-node}  & \hspace{0.3cm} & {\small $\Punhap$} &{\small $Z_{\textrm{unhap}}\sim B(2w, \rho)$}\hspace{0.5cm} &{\small always rare}\\[1ex]
\end{tabular}}
\centering
\end{table}

\section{The case 
when intolerance is at most 50\%}\label{se:casetaubhalf}
In this case the behaviour of the process $(n,w,\tau,\rho)$ is very different, since
the mixing index is strictly decreasing. This means that
the process is bound to arrive to a dormant state, with absolute certainty.
Note that if $\tau\leq 0.5$ then complete segregation is a dormant state, but it can be shown that the final state
is never complete segregation. 
We show
that in most typical cases for $\rho$,  the outcome is static when $\tau\leq 0.5$.
We assume that $\rho<0.5$ because the case
$\rho=0.5$ has already been analysed in \cite{brandt:an, BELschel13} 
and the 
case $\rho >0.5$ is symmetric. Hence on the hypothesis $\tau\leq 0.5$ we have
$\rho + \tau <1$ and by Table \ref{ta:validptwocasesnadcorr}
the unhappy $\alpha$-nodes are an arbitrarily small
proportion of the $\alpha$-nodes 
as $w\to\infty$.
In any case, since $\rho<0.5<1-\rho$
we have $\tau-\rho< 1-\tau -\rho$, 
so the probability that an $\alpha$-node is unhappy is much smaller
than the probability that a $\beta$-node is unhappy.
However what matters in the analysis for $\tau\leq 0.5$ is the relationship between
the likelihood of stable intervals and unhappy $\alpha$-nodes. This analysis is 
a reminiscent of the work in 
\cite{BELschel13}, but has some new features. 

\begin{figure} 
\centering\includegraphics[scale=0.3]{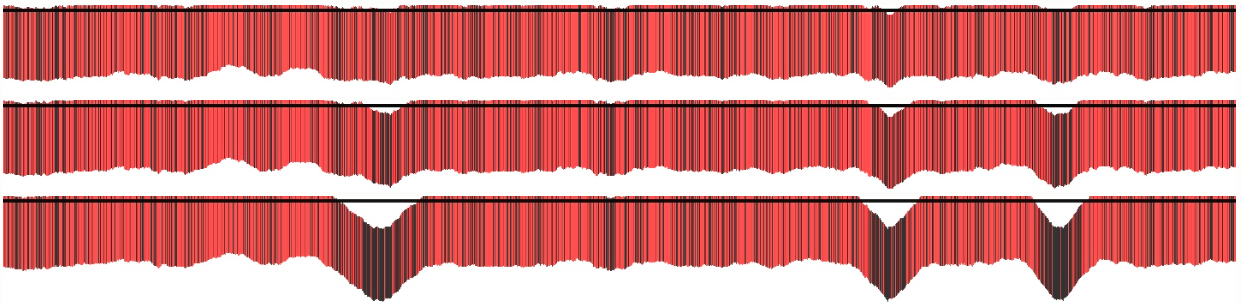}
\caption{The formation and dynamics of the infected area when $\tau+\rho<1$. }\label{fig:biasplot}
\end{figure}   

\begin{defi}[Stable intervals]
A stable interval is
an interval of nodes of length $w$ which contains at least $(2w+1)\tau$ nodes of one 
or the other type. 
An interval is $\alpha$-stable if it contains at least $(2w+1)\tau$ nodes
of type $\alpha$.
\end{defi}

The $\beta$-stable intervals are defined analogously.
Note that no $\alpha$-node which is inside 
an $\alpha$-stable interval can swap during the process.
The reason is that such $\alpha$-nodes are 
happy just because of the presence of the other $\alpha$-nodes
 in the same interval. Then a simple induction shows 
 that they will continue to be happy throughout the process,
 thereby remaining  immune to swaps and fixed in their initial positions.
A similar observation  applies to $\beta$-stable intervals. 
The existence of stable intervals is characteristic to the case $\tau\leq 0.5$.

The events we are interested are the occurrences of $\alpha$-stable intervals and
unhappy $\alpha$-nodes. 
The probabilities $\Pstab, \Punhap$ of these two rare events can be viewed as tails of certain
binomial distributions.
Consider the variables, probabilities and distributions of Table \ref{ta:vlikelyghd}.
It is not hard to see that
\[
\Pstab=\proba{Z_{\textrm{stab}}\geq (2w+1)\tau}
\hspace{1cm}
\textrm{and}
\hspace{1cm}
\Punhap\geq 
\proba{Z_{\textrm{unhap}}\geq 2w(1-\tau)}.
\]
We are interested in the event where the ratio 
$\Punhap/\Pstab$ becomes small, because of the following fact.
\begin{lem}[Static processes]
Suppose that $\tau, \rho$ are such that
 $\Punhap=
 \bigo{c^{-w}\cdot \Pstab}$
for some $c >1$. 
Then 
with high probability the
 process $(n,w,\tau, \rho)$ is static,
 and in fact there exists some
$c_{\ast}>1$ such that with high probability 
the process stops after at most $n\cdot c_{\ast}^{-w}$ 
many steps.
\end{lem}

\begin{table}
\caption{Threshold constants of interest and their derivation equations.}
\label{ta:constantthresint}
\colorbox{black!10}{\arrayrulecolor{green!50!black} 
  \begin{tabular}{cclc}
\multicolumn{1}{c}{\bf\small Threshold}  &  &\multicolumn{1}{c}{\bf\small Solution to equation} & 
{\bf\small Stability condition}\\[1ex]
 \cmidrule{1-4}
{\small  $\kappa_0\approx 0.35309$\hspace{0.1cm}}&&
{\small $(1-\tau)^{1-\tau}=(0.5-\tau)^{0.5-\tau}$} & 
{\small $\rho<0.5$ \hspace{0.2cm}and \hspace{0.2cm}$\tau\leq\kappa_0$}\\[1ex]
{\small $\lambda_0 \approx 0.41149$}\hspace{0.1cm} && 
{\small $(1-\tau)^{1-\tau}= \left(0.5-\tau\right)^{0.5-\tau} \cdot \sqrt{2\tau}$\hspace{0.4cm}} &	
{\small $\rho<\lambda_0$ \hspace{0.3cm}and\hspace{0.2cm} $\tau \leq \lambda_0$}\\[1ex]
\end{tabular}}
\centering
\end{table}

The intuition here is that, if the unhappy $\alpha$-nodes are much more rare than the
stable $\alpha$-stable intervals (i.e.\  if $\Punhap=\smo{\Pstab}$)  then it is very likely that unhappy
$\alpha$-nodes are enclosed in small intervals which are guarded by $\alpha$-stable intervlas. This means that
the familiar cascades that can be caused by the eviction of an unhappy $\alpha$-node are bound to be contained in small
areas of nodes. The very definition of stable intervals ensures that such cascades cannot pass through them. Hence the condition
 $\Punhap=\smo{\Pstab}$ guarantees that any $\alpha$-to-$\beta$ swaps are contained in small areas of nodes of total size $\smo{n}$.
 Due to the monotonicity of the mixing index, this means that there can only be at most 
 $\smo{n}$ swaps in this case.
 
The second item in Figure \ref{fig:red_3D_tr} shows the probabilities $\Pstab, \Punhap$ (for $w=100$)
with respect to $\tau,\rho$. We see that for points away from $(0.5, 0.5)$, the surface $\Punhap$ is above $\Pstab$, 
and there is a threshold curve beyond which the opposite relationship is established.
Using basic results about the tail of the binomial distribution, and Stirling's approximation
we can derive the following sufficient condition for $\Punhap =\smo{\Pstab}$:
\begin{equation}\label{eq:rhotaueqfostabgz}
\textrm{$g(\tau,\rho)>0$,\hspace{0.4cm} where \hspace{0.4cm}$g(\tau,\rho)=\frac{1}{2}\cdot 
\left(\frac{(1-\tau)^{1-\tau}}{(0.5-\tau)^{0.5-\tau}}
\right)^{2}-\rho.$}
\end{equation}
The third item of Figure \ref{fig:red_3D_tr}
is a representation of $g(\tau,\rho)$ in the space, up to where
it becomes negative, at which point we project it on the plane.
The values of $\tau,\rho$ that we are interested correspond to points
on the plane, outside the collapsed area.
This boundary (a curve) is more clear in the first 
item of Figure \ref{fig:red_3D_tr} which is the projection of the
surface to the plane, with different colours indicating the points
which make $g$ positive or negative.
This boundary can be simplified (with slight loss of generality) if we consider the line
that passes from the two points where the boundary curve intersects the lines $\tau=0.5$ and $\rho=0.5$.
Hence if $2\rho \cdot (1-2\kappa_0)+\tau +\kappa_0<1$, we are in the stable region, which shows a clause of Theorem \ref{th:complsegm}. 
Note that both of the partial derivatives of $g$ are negative 
when $\tau,\rho\in [0,0.5)$.
If we fix $\rho=0.5$ then the largest value of $\tau$ that keeps $g(\tau,\rho)\geq 0$
is the solution ($\kappa_0\approx 0.353092313$) of the first equation
of Table \ref{ta:constantthresint}. 
Hence we may conclude that
if  $\tau<\kappa_0$ 
and $\rho\in (0, 0.5]$ then 
$\Punhap =\bigo{c^{-w}\cdot\Pstab}$
for some $c>1$.
We can also look for the largest square that is contained 
in the large area of the first item of Figure \ref{fig:red_3D_tr} (where the process is static).
The edge of this square is given in Table \ref{ta:constantthresint}.
Hence if $\rho, \tau \in (0, \lambda_0)$ 
then $\Punhap =\bigo{c^{-w}\cdot\Pstab}$
for some $c>1$.

We have one last observation to make about the function $g$.
If we let do not restrict the values of $\tau\in (0,0.5)$ then we wish to
find the values of $\rho$ such that $g(\tau, \rho)$. According to the properties of
$g$ (in particular its negative derivative on $\rho$), 
these are all the positive numbers which are less than the limit 
(which is also an infimum)
\[
\lim_{\tau\to 0.5} \frac{1}{2} \cdot \left(
\frac{(1-\tau)^{1-\tau}}{(0.5-\tau)^{0.5-\tau}}
\right)^{2} = 0.25
\]

Hence 
we may conclude that
if $\rho\leq 0.25$
and $\tau \in (0, 0.5)$ 
then $\Punhap =\bigo{c^{-w}\cdot\Pstab}$
for some $c>1$.
This concludes the proof of the second clause 
of Theorem \ref{th:complsegm}.

%

\newpage

\section{Appendix}\label{se:prelim}
In this section we provide supplementary material to the main part of the paper. 
This includes mainly proofs of the claims we made towards the proof of our main theorem, but also
additional introductory material, figures, tables and mathematical background.
The structure of this supporting material follows the presentation of the main part of the paper.

\subsection{Schelling models}\label{ssub:varschel}
The definition of the Schelling model in Section \ref{def:modseg} is rather standard, close to
the spacial proximity model from  \cite{TS1, TS71a}  and identical to
the model studied in \cite{brandt:an, BELschel13}. 
Most significantly, it is an unperturbed Schelling model, where agents cannot
make moves that are detrimental to their welfare.
We have already remarked in the introduction
that various more realistic-looking rigorously analysed perturbed versions of the model 
in the literature (such as \cite{JZ1}) actually force `regularity' on the process, which
makes it fit an already existing methodology 
(such as Markov chains with a unique stationary distribution, 
or with properties that guarantee stochastically stable states).
Even if we commit to the absence of perturbations in the model, 
it is possible to add complications to the simple dynamics defined in
Section \ref{def:modseg}. For example, the agents
may take into account the distance they need to travel before they move.
However it is the simplicity of the original Schelling model,
contrasted by  the complexity of the analysis required to specify its behaviour
(as demonstrated in \cite{brandt:an, BELschel13}) 
that make this topic fundamental and interesting.

Under the above requirement for simplicity and proximity to the original model,
there remain a number of ways that the model can be altered or generalised.
For example, note that in the case that $\tau>0.5$ in the model of
Section  \ref{def:modseg}, two nodes may swap although the number of
same-type nodes in their neighbourhoods remain the same after the swap.
One may alternatively require that for such a swap, the corresponding
numbers of same-type nodes in the neighbourhoods increase
(note that such a modification would not make a difference if $\tau\leq 0.5$).
Our choice on this issue follows Brandt, Immorlica, Kamath, and Kleinberg in
\cite[Section 2]{brandt:an}.
One generalisation, considered in \cite{BELtipl13}, is to allow
different tolerance thresholds for the two types of
individuals. Another generalization, already present in \cite{TS1}, is to introduce a number
of vacancies, i.e.\ to allow the total number of individuals to 
be smaller than the number of sites. 
We could also alter the dynamics. Instead
of switching two chosen individuals at each stage, we could merely choose {\em one}
individual and change his type. Such an action may be 
interpreted as the departure of the individual to some external location and the arrival of an 
individual of the opposite type at the site that has just become available. 
Model with this dynamics are often said to have {\em switching agents}
(see \cite{BELtipl13}, where such a model was analysed) as opposed to the 
{\em swapping agents} of the model of Section \ref{def:modseg}. 

It is worth pointing out that the Schelling model with switching agents
is closely related to the spin-1 models used to analyse phase transitions in physics, and in
particular the Ising model. Indeed, in the Ising model (originally introduced in order to explain
ferromagnetism in the context of temperature) a system of atomic nuclei interact with an 
auxiliary `heat bath' which affects their spin. Such connections have been analysed by many authors
(see for example \cite{SS,DM,PW,GVN,GO}), where the dynamics is based on
the Boltzmann distribution on the set of possible configurations. A rough analogy between
the two models is that `energy'  corresponds to  some measure of the mixing of types
(see the definition of the mixing index for the Schelling model below) and `temperature' 
corresponds to the intolerance parameter $\tau$ (as least insofar phase transitions refer to
varying values of the temperature or $\tau$). On the other hand, 
the Schelling model with closed dynamics has a counterpart in the
Ising model with Kawasaki dynamics.

\subsection{Objectives of the analysis of the unperturbed model and related work}
\label{sub:knownressch}
\label{se:1Dmorphojob}
We use the notation of Section \ref{def:modseg}, so that 
the symbol $n$ always means the population
variable of the process, and $w$ always is the parameter of the process which
determines the length of the neighbourhood of nodes. Similarly, $\tau, \rho$
always refer to the parameters of the Schelling process.

In Section \ref{se:accesscomplseg} we show that,
with probability one, the process 
$(n, w, \tau, \rho)$ either reaches complete segregation or
it reaches a dormant state.
In the second case, we wish to 
determine the extent of segregation in the dormant state.
In view of the large number of states that the process may have
(most of them `random') a question arrises as to how to classify or even
talk precisely about different states that may be the outcome of the process.
Brandt, Immorlica, Kamath, and Kleinberg noticed in \cite{brandt:an} that, 
at least in the case $\tau=\rho=0.5$ that they considered, 
the extent of the segregation that occurs in the final state
depends crucially on $w$. In fact, they showed that
the dependence on $w$ is `polynomial'.
We may say that
a state is regarded as {\em polynomial segregation} if, with high probability
a randomly chosen node belongs to a contiguous block of size that is proportional to 
the value of a polynomial on $w$. A similar definition applies to 
{\em exponential segregation}. These two notions turn out to provide 
a very useful language for explaining the eventual outcome of the Schelling process.
A full characterization (extending the work of Brandt, Immorlica, Kamath, and Kleinberg 
\cite{brandt:an}) of the asymptotic behaviour of the process $(n,w, \tau, \rho)$
for $\rho=0.5$ and $\tau\in[0,1]$ was provided by the authors in \cite{BELschel13}
in terms of polynomial and exponential segregation, as well as static processes.
Intuitively, a random state is non-segregated, while polynomial 
and exponential segregation correspond to highly non-random states.
\begin{table}
\colorbox{black!10}{\arrayrulecolor{green!50!black}
\begin{tabular}{l|cccc}
\toprule
{\small\bf Intolerance}& {\small $\tau\in [0,\kappa_0)$} & {\small  $\tau\in (\kappa_0, 0.5)$} & 
{\small  $\tau=0.5$}&{\small  $\tau\in (0.5, 1]$} \\[1ex]
{\small\bf Segregation}&{\small Negligible}  &{\small Exponential} & {\small Polynomial}& {\small Complete} \\
\bottomrule
\end{tabular} }
\centering
\caption{Segregation regions in the case $\rho=0.5$.}
\label{ta:vsclassifpoex}
\end{table}
The characterization from \cite{BELschel13} is summarized in 
Table \ref{ta:vsclassifpoex}. 
It is rather striking that when intolerance is increased from, say, $0.4$ to $0.5$ 
the segregation is decreased. This phenomenon is akin to the many paradoxes
that stem from the missing link between local motives of agents and global behaviour of
a system (e.g.\ see Schelling's classic monograph \cite{TS2}, and in particular Chapter 4
which relates to his segregation models). 
Even more strikingly, the authors showed in
\cite{BELschel13} that the paradox occurs for all $\tau\in(\kappa_0, 0.5)$, i.e.\ as
$\tau$ approaches $0.5$ the segregation (in the final state) decreases.

This paradoxical phenomenon is also clear in many simulations of the model.
Figure \ref{fig:ciclmorphostatlides} shows typical runs of the processes
$(5\cdot 10^5, 3\cdot 10^3, \tau, 0.5)$ for $\tau\in\{0.485, 0.49,0.495, 0.5\}$.
The final state is depicted in the circle, where the nodes of one type are black and
the nodes of the other type are grey. We use the space between the centre of the
ring and the ring in order to record the actual process, as it evolves in time.
In particular, if a grey node switches its place with a black node, we put a black node
(the colour of the more recent node) between the location of the node and the centre of the ring,
at a distance from the centre which is proportional to the stage where the swap occurred.
Hence we may observe ``cascades' of swaps of nodes of the same type, which
are less severe  as $\tau$ approaches $0.5$. Such cascades are crucial in the
rigorous analysis of the model, 
both in
\cite{brandt:an} and in \cite{BELschel13}.
Figure \ref{fig:ciclmorphostatlides} shows  that as $\tau$ approaches $0.5$, the segregation
is decreased. This behaviour can be traced to the probability that a node is unhappy in the
initial configuration, and in fact, the threshold constant $\kappa_0$ is derived
by comparing related probabilities in \cite{BELschel13}.

In the case $\rho=0.5$
the two constants $\kappa_0$ and $0.5$ mark {\em phase transitions} in the limit state
of the process $(n, w, \tau, \rho)$, as $\tau$ takes values in $[0,1]$.
This brings us to another important objective
of the analysis of the Schelling process, which is the discovery of phase transitions with
respect to the parameters $\tau, \rho$.
Incidentally, we note that the discovery of phase transitions has been one of the
original motivations for the study of the one and two dimensional Ising model,
when one varies the temperature (see the end of Section \ref{ssub:varschel} for a brief
discussion of the analogy between the Ising and the Schelling models).
Finally we are also interested in the expected time that the process take to converge.

\begin{figure} 
\centering\includegraphics[scale=0.35]{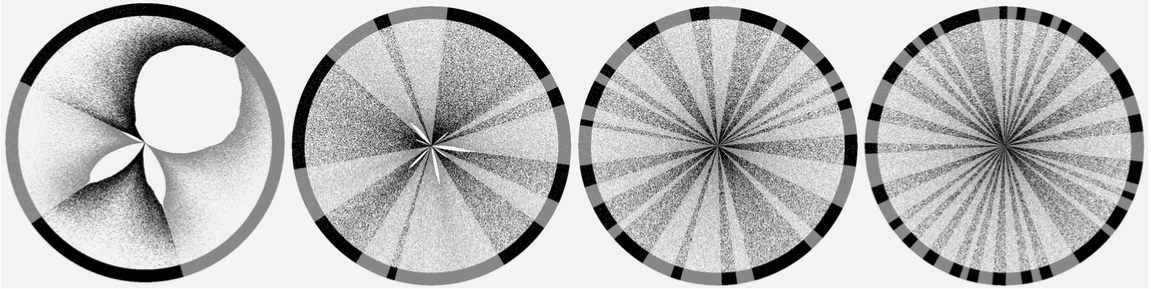}
\caption{500K population with $w=3000, \rho=0.5$ and $\tau=0.485, 0.49, 0.495, 0.5$. All made about 130K swaps.}
\label{fig:ciclmorphostatlides}
\end{figure}  

\subsection{Asymptotic notation}
 We use the asymptotic notation. 
Given two functions $f,g$ on the positive integers, (as is standard) we say that 
$f$ is $\bigo{g}$ if there exists a positive constant 
$c$ such that $f(t)\leq c\cdot g(t)$ for all $t$.
We say that  $g$ is $\Omega(f)$ if $f$ is $\bigo{g}$, and that $g$ is $\Theta(f)$ if both
$f$ is $\bigo{g}$ and $f$ is $\Omega(g)$.
We also use this notation, however,  in a more general sense: we say that $f$ is $g(\bigo{t})$
if there exists some $c>0$ such that $f\leq g(ct)$ for all $t$. For example, when we  say that
a function $f$ is $ne^{-\bigo{t}}$, this means that  there is $c>0$ such that $f(t)\leq ne^{-ct}$
for all $t$.  Or, if we say that $f$ is $n(1-e^{-\bigo{t}})$, this means that  there is $c>0$ such that $f(t)\leq n(1-e^{-ct})$
for all $t$.
Similarly, we use $\Theta$ in a more general sense. 
We say that $f$ is $g(\Theta(t))$ to mean that there exist constants $c_0$ and $c_1$ such that $ g(c_0\cdot t) \leq  f(t)  \leq g(c_1 \cdot t)$ for all $t$. 
We say that $f=\smo{g}$ if $\lim_t f(t)/g(t)=0$.
The (often hidden) variable underlying the asymptotic
notation in the various expressions will be $w$. 
In other words, for fixed values of $\rho$ and $\tau$, 
the choice of constants required in the asymptotic notation, will always depend only on $w$.  
We also combine the `high probability' terminology with the 
asymptotic notation in a manner which is worth clarifying. When we say, for example, that 
`with high probability 
 the number of initially unhappy $\alpha$-nodes in the process  
$(n,w,\tau,\rho)$ is 
$n\cdot (1-\rho)\cdot  e^{-\Theta(w)}$', this means that there exist constants 
$c_0$ and $c_1$ such that, with high probability,  the number of initially unhappy $\alpha$-nodes in the process  
$(n,w,\tau,\rho)$ lies  between 
$n\cdot (1-\rho)\cdot  e^{-c_0 \cdot w}$ and $n\cdot (1-\rho)\cdot  e^{-c_1 \cdot w}$.

\subsection{Overview of our analysis}
We use different methods for the cases $\tau\leq 0.5$ and $\tau>0.5$.
If $\tau\leq 0.5$, in order to derive conditions under which the process is static, we
analyse and compare the probabilities of initially unhappy nodes
and  {\em stable intervals}. This approach was introduced
by the authors in 
\cite{BELschel13}. If $\tau>0.5$ we consider the two cases $\tau +\rho<1$ and 
$\tau +\rho>1$ and argue (using distinct arguments) that in each of them complete
segregation is the high probability outcome. We elaborate on these arguments.

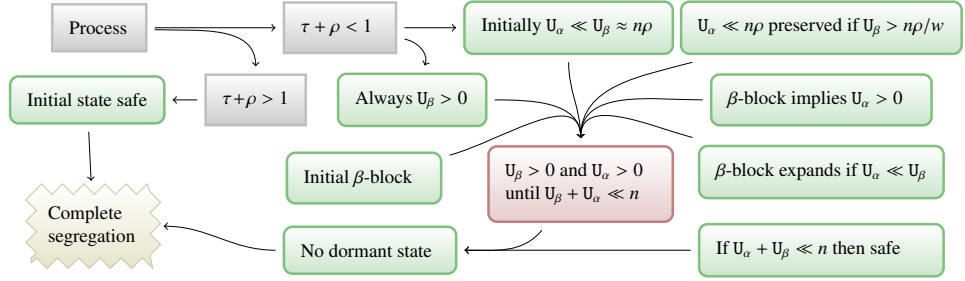
\begin{figure}
\scalebox{0.8}{
\begin{tikzpicture}[
Nnode/.style={rectangle, rounded corners, minimum height=4mm, very thick, 
draw=green!50!black!50, top color=white, bottom color=green!50!black!20, font=\footnotesize },
tnode/.style={rectangle, rounded corners, very thick,
draw=red!50!black!50, top color=white, bottom color=red!50!black!20, font=\footnotesize},
pnode/.style={rectangle, very thick,
draw=white!50!black!50, top color=white, bottom color=white!50!black!40, font=\footnotesize}]
 \node (Npreserv) [Nnode,  outer sep=4pt, inner sep=8pt] at (7.5,2) {\parbox{4.1cm}{
 $\unhap_{\alpha}\ll n\rho$ preserved if $\unhap_{\beta}> n\rho/w$}}; 
 \node (Nbetablgood) [Nnode,  outer sep=4pt, inner sep=8pt] at (7.5,0.8) {\parbox{3.1cm}{
 $\beta$-block implies $\unhap_{\alpha}>0$}};
 \node (Nbetablinitial) [Nnode,  outer sep=4pt, inner sep=8pt] at (-0.1,-0.5) {\parbox{2cm}{
Initial $\beta$-block}};
 \node (Nbetablexpands) [Nnode,  outer sep=4pt, inner sep=8pt] at (7.5,-0.4) {\parbox{3.5cm}{
 $\beta$-block expands if  $\unhap_{\alpha}\ll\unhap_{\beta}$}};
 \node (Nsafe) [Nnode,  outer sep=4pt, inner sep=8pt] at (7.5,-1.7) {\parbox{3.5cm}{
 If  $\unhap_{\alpha}+\unhap_{\beta}\ll n$ then safe}};
 \node (Nbreakth) [tnode,  outer sep=4pt, inner sep=8pt] at (3.5,-0.6) {\parbox{2.5cm}{
 $\unhap_{\beta}>0$ and $\unhap_{\alpha}>0$\\ until $\unhap_{\beta}+\unhap_{\alpha}\ll n$}};
 \node (Nunbala) [Nnode,  outer sep=4pt, inner sep=8pt] at (3.3,2) {\parbox{2.9cm}{Initially
 $\unhap_{\alpha}\ll \unhap_{\beta}\approx n\rho$}};
 \node (Nhard) [pnode,  outer sep=4pt, inner sep=8pt] at (-0.5,2) {\parbox{1.3cm}{$\tau+\rho<1$}};
 \node (Ninitsafe) [Nnode,  outer sep=4pt, inner sep=8pt] at (-4.7,0.8) {\parbox{2cm}{Initial state
 safe}};
 \node (Nalwunhab) [Nnode,  outer sep=4pt, inner sep=8pt] at (0.7,0.8) {\parbox{1.9cm}{Always 
$\unhap_{\beta}>0$}};
 \node (Nprocess) [pnode,  outer sep=4pt, inner sep=8pt] at (-4.5,2) {\parbox{1cm}{Process}};
 \node (Neasy) [pnode,  outer sep=4pt, inner sep=8pt] at (-1.9,0.8) {\parbox{1.1cm}{$\tau+\rho>1$}};
\node (Ncomplete) [draw=yellow!50!black!50, top color=white, bottom color=yellow!50!black!20, decorate, decoration=zigzag, 
outer sep=4pt, inner sep=8pt, font=\footnotesize] at (-4.6,-1.3) {\parbox{1.5cm}{Complete segregation}};
 \node (Nnodormant) [Nnode,  outer sep=4pt, inner sep=8pt] at (0,-1.7) {\parbox{2.3cm}{No dormant 
 state}};
\draw [->, in = 90, out=-2] (Nprocess) to (Neasy);
\draw [->] (Neasy) to (Ninitsafe);
\draw [->] (Nprocess) to (Nhard);
\draw [->, in = 70, out=-10] (Nhard) to (Nalwunhab);
\draw [->] (Nhard) to (Nunbala);
\draw [->, in = 90, out=0] (Nalwunhab) to (Nbreakth);
\draw [->, in = 90, out=15] (Nbetablinitial) to (Nbreakth);
\draw [->, in = 90, out=-90] (Nunbala) to (Nbreakth);
\draw [->, in = 90, out=195] (Npreserv) to (Nbreakth);
\draw [->, in = 90, out=180] (Nbetablgood) to (Nbreakth);
\draw [->, in = 90, out=165] (Nbetablexpands) to (Nbreakth);
\draw [->, in = 0, out=-130] (Nbreakth) to (Nnodormant);
\draw [->, in = 0, out=180] (Nsafe) to (Nnodormant);
\draw [->] (Ninitsafe) to (Ncomplete);
\draw [->, in = 0, out=180] (Nnodormant) to (Ncomplete);
\end{tikzpicture}}
\centering
\caption{The logic of the proof that if $\tau>0.5$, with high probability the process reaches complete segregation. Here `$\beta$-block'
refers to the persistent $\beta$-block of Section \ref{se:persist}.}
\label{fig:dyprogeneralskads}
\end{figure}

\subsubsection*{{\bf Case} \texorpdfstring {$\tau>0.5$}{Hard case}}
This case is divided to the cases $\tau+\rho>1$ and $\tau+\rho<1$, and the
structure of the analysis was depicted as a flowchart in Figure \ref{fig:dysimpleskads}.
Here we give a more detailed overview, which is illustrated in the more elaborate flowchart of Figure
\ref{fig:dyprogeneralskads}.
First, we show that asymptotically (on $w,n$), from any state
there is a series of transitions
that leads to either a dormant state, or complete segregation.
Hence, since there are only finitely many states, with probability one
the process will reach either a dormant state or complete segregation.
So in order to establish complete segregation as the eventual outcome, it suffices to show
that the process maintains unhappy nodes of each colour during all stages.

First, assume that $\tau+\rho>1$. In this case we can show that, assuming that the 
actual proportion of $\beta$-nodes is
sufficiently close to $\rho$ (which is very likely according to the law of large numbers),
every reachable state is not dormant. More precisely, we show that
given such numbers of $\alpha$ and $\beta$-nodes, every permutation of them on
the ring corresponds to a state which has both unhappy $\alpha$ and unhappy $\beta$-nodes.
Since the numbers of nodes of each type do not change during each transition, this argument
suffices for this case. Recall that states with the property that no series of transitions from them leads to
dormant states are called {\em safe}. So, in the case
$\tau+\rho>1$ we argue that (with high probability) the initial state is safe.

Second, we assume that $\tau+\rho<1$, which is a considerably harder case.
Under this hypothesis, in the initial configuration we have
$\smo{n}$ many unhappy $\alpha$-nodes and
$\Omega(n)$ many unhappy $\beta$-nodes.
As before, it suffices to show that (with high probability) the process never reaches a
dormant state. It is not hard to see that (with high probability) the initial state is not dormant.
However it is no longer clear if the initial state is safe. We show  
that given the expected numbers of nodes of the two types in the initial
state (or numbers sufficiently close to their expectations) any permutation of the nodes
on a ring corresponds to a state with at least one unhappy $\beta$-node.
Hence, with high probability, the process will never run-out 
of unhappy $\beta$-nodes and we only need
to argue about the preservation of unhappy $\alpha$-nodes.
Already it should be clear that this is an {\em asymmetric} case where the 
$\alpha$-nodes (the majority) and the $\beta$-nodes (the minority) play different roles.
When $\tau+\rho<1$ there are many permutations of the nodes 
(which correspond to states where all $\alpha$-nodes
are happy, i.e.\ dormant states. So the argument that was used in the case 
$\tau+\rho>1$ is no longer 
relevant for arguing for the preservation of unhappy $\alpha$-nodes in the process.
The argument we use instead is based on the asymmetry between the number of unhappy
$\beta$-nodes and the unhappy $\alpha$-nodes, which creates a dynamic that favours
the preservation of unhappy $\alpha$-nodes. More precisely, it favours the preservation of
$\beta$-blocks of length $>w$, which is a condition implying the existence of unhappy
$\alpha$-nodes (indeed, the $\alpha$-nodes 
neighbouring a $\beta$-block of length at least $w$
are unhappy). Hence if we show that the expected number of unhappy $\alpha$-nodes remains
small during the stages of the process,
 then we have that we can expect the existence of unhappy $\alpha$-nodes 
 (and unhappy $\beta$-nodes) up to the point where the total number of unhappy nodes
 is small.

In addition we show that if the
total number of unhappy nodes in a state is sufficiently small, then
this state is safe, i.e.\ there is no series of transitions from it to a dormant state.
The argument is concluded by showing that it is very likely that by stage $n$ the process
will arrive at a state with appropriately low number of unhappy nodes, before it reaches 
a dormant stage.
Figure \ref{fig:inf_area} is a plot of the numbers of unhappy $\alpha$-nodes
and the unhappy $\beta$-node during the stages, taken from two typical simulations
(one with large and one with small population), when $\tau+\rho<1$.
The process we described is clearly visible: the number of unhappy $\alpha$-nodes
remains small, until the number of unhappy $\beta$-nodes becomes small.
Up to the later point, as we explained, the dynamics favours the preservation of unhappy
$\alpha$-nodes.

\subsubsection*{{\bf Case} \texorpdfstring {$\tau\leq 0.5$}{Static case}}
In this case we have $\tau+\rho<1$, and  this means that
 in the initial configuration 
  the $\alpha$-population is happy with a few exceptions, while
the $\beta$-population  is unhappy, with a few exceptions.
By the definition of the dynamics of the model $\alpha$-to-$\beta$
swaps can only occur in areas where there are unhappy $\alpha$-nodes.
Hence in this case
the $\alpha$-to-$\beta$ swaps will be concentrated in a very few selected areas in the ring,
at least in the first stages of the process.
This concentration of 
$\alpha$-to-$\beta$ swaps creates cascades of $\alpha$-node evictions which can be clearly
seen in simulations such us the one displayed 
in Figure  \ref{fig:cicle_slides}.\footnote{Here the current configuration is the outer circle, while the initial random
state is the inner small circle. Whenever a swap occurs at some stage, a dot is placed at a distance from the center which
is proportional to that stage, at the same angle where the involved node lies. The color of the dot corresponds to the type that the node changed to under the particular swap.} If we could argue that such cascades are restricted to
small areas around the initially unhappy $\alpha$-nodes, then it is not hard to argue that
the process reaches a dormant state rather quickly, having affected only a very 
small number of nodes. The way we do this is through {\em stable intervals}, 
a device that was also used in \cite{BELschel13}.
Roughly speaking, these are intervals that do not allow the spread of unhappy $\alpha$-nodes through them.

If $\rho$ is very small, or if $\tau$ is very small, then stable intervals occur with high probability. 
On the other hand, if $\rho,\tau$ get sufficiently large,
the probability of a stable interval tends to 0 as $w\to\infty$.
This contrasts with prevalence of unhappy $\alpha$-nodes. When $\tau, \rho$ are small,
the probability of (the occurrence of) an unhappy $\alpha$-node is small, while it gets large
when $\tau, \rho$ increase. Figure \ref{fig:proba_unhap_stabl} shows
the actual probabilities (as calculated in Section \ref{se:casetaubhalf}) as functions
of $\tau,\rho$ for the specific value of $w=100$ (the shape of the plots does not 
change significantly for
different values of $w$). 
The interesting case is the range for $\tau,\rho$ where both probabilities tend to $0$ as
$w\to\infty$, i.e.\ both events become rare. Somewhere on the horizontal $\tau$-$\rho$ plane
there is a line marking the intersection of the two surfaces.
This is where the probability of a stable interval becomes less than the probability of an unhappy
$\alpha$-node. Moreover, as $w\to\infty$ the ratio of the two probabilities tends to infinity or zero,
depending whether $\tau,\rho$ sit on one side of the plain (with respect to the intersection line)
or the other. 
The crux of the argument in Section \ref{se:casetaubhalf}  is that for many values of 
$\tau,\rho$ stable intervals are much more common than unhappy $\alpha$-nodes in the
initial configuration. This allows us to argue that, in this case, 
the process has to reach a dormant state after 
$\smo{n}$ many swaps. 

 \begin{figure}  
 \centering
\includegraphics[scale=0.3]{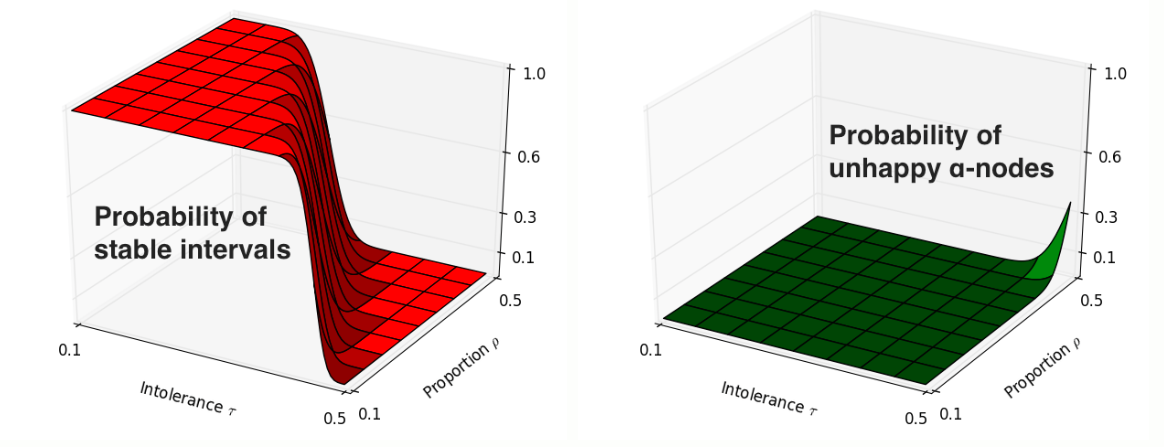}
\caption{The probabilities of a stable interval and an 
unhappy $\alpha$-node, as functions of $\tau,\rho\leq 0.5$ when $w=100$.}
\label{fig:proba_unhap_stabl}
\end{figure}

\subsection{Properties of welfare metrics}\label{subse:welfare}
The social welfare $\texttt{V}$ of the state can easily be seen to be
 non-decreasing along the transitions of the process.
Let us establish the relationship with the mixing index.
Given a certain state of the process and a node $u$, 
we let $u^{\alpha}$
denote the number of $\alpha$ nodes
that are located in the neighbourhood of $u$ at this state. Similarly, we let
$u^{\beta}$ denote the number of $\beta$-nodes
that are located in the neighbourhood of $u$.
Furthermore, we denote by $(\alpha_j)$ and $(\beta_i)$ and  the finite sequences of 
$\alpha$ and $\beta$  nodes respectively in the state. 
Hence $\alpha_j^{\beta}$ denotes the number of $\beta$-nodes
that are located in the neighbourhood of $\alpha_j$ while
$\beta_j^{\alpha}$ denotes the number of $\alpha$-nodes
that are located in the neighbourhood of $\beta_j$. 
Given a state, let $n_{\alpha}, n_{\beta}$ be the number of 
$\alpha$ and $\beta$-nodes respectively. Then
\begin{equation}\label{eq:gjrisrig}
\sum_{j<n_{\beta}} \beta_j^{\alpha} = \sum_{i<n_{\alpha}} \alpha_i^{\beta}.
\end{equation}
In order to prove this equality,
consider the state of $\alpha$ and $\beta$ types in
the state and start by removing all $\beta$ from their positions. Then, adding
the $\beta$ types one-by-one back to their original 
positions we can see each placement
incurs the same increase to the two sums.
Hence by induction, the two sums are equal.

We call the number in \eqref{eq:gjrisrig} the {\em mixing index} of the state,
because it can be used as a metric of how mixed 
(i.e.\ not segregated) the population of $\alpha$ and $\beta$ types
is at the given state. 
Indeed, suppose that the state has at least
$2w+1$ nodes of each type.
In the state of complete segregation  
the sums in \eqref{eq:gjrisrig} take the value $2\cdot (1+\cdots +w)$, which is $w(w+1)$. 
This can be shown to be the minimum mixing index (in a state which has 
at least
$2w+1$ nodes of each type).
At the other extreme, 
if the two types are uniformly mixed (in the sense that
every interval $I$ has approximately $\rho_{\ast} \cdot |I|$ green nodes) then 
the sums in
\eqref{eq:gjrisrig} take approximately the value 
$n\cdot 2w\cdot \rho_{\ast}(1-\rho_{\ast})$, 
which can be shown to be the maximum
possible mixing index.
We also have
\begin{equation}\label{eq:twosumeqs}
\sum_{i<n_{\alpha}} \alpha_i^{\alpha} + 
\sum_{i<n_{\alpha}} \alpha_i^{\beta} = 
(2w+1)\cdot n_{\alpha} \hspace{0.5cm}\textrm{ and }\hspace{0.5cm}
\sum_{j<n_{\beta}} \beta_j^{\beta} + \sum_{j<n_{\beta}} \beta_j^{\alpha} 
= (2w+1)\cdot n_{\beta}.
\end{equation}
From \eqref{eq:gjrisrig}
and \eqref{eq:twosumeqs} we get 
$\mathtt{V}= (2w+1)\cdot n - 2\cdot\mix$.

\begin{lem}
If $\tau\leq 0.5$, each step in the process
decreases the mixing index by at least 4.
\end{lem}
\begin{proof}
Suppose that we swap an unhappy $\alpha$-node $u$
with an unhappy $\beta$-node $v$. 
Let $N_u, N_v$ be the 
neighbourhoods of $u,v$ respectively and let 
$I=N_u\cap N_v$. Here we view the nodes as stationary, so that a swap of nodes
means a swap of their types.
The mixing index of the nodes in 
$I$ will not change after the swap.
Since $\tau\leq 0.5$ the number $x$ of $\alpha$-nodes
in $N_u-I-\{u\}$ is smaller the the number $y$ of $\alpha$-nodes
in $N_v-I-\{u\}$.  After the swap the mixing index of each of the 
$\alpha$-nodes in $N_u-I-\{u\}$ will increase by one while the mixing index
of each of the $\beta$-nodes in the same set will decrease by one.
If $t=2w+1$ is the length of the neighbourhood and $i$ is the number of
$\alpha$-nodes in $I$ then the mixing index of $u$
before and after the swap is $t-x-i$ (the size of the neighbourhood minus the $\alpha$-nodes in the
neighbourhood)  and $x+i$  (the number of $\alpha$-nodes in $N_u-I$ plus the number 
of $\alpha$-nodes in $I\subseteq N_u$) respectively.
Hence the difference in the sum of the mixing indices of the nodes in
$N_u-I$ before and after the swap is the addition of
\begin{itemize}
\item[(a)] the difference in the mixing index of $u$
\item[(b)] the difference in the sum of the mixing indices of the nodes in $N_u-I-\{u\}$
\end{itemize}
where the differences refer to the stages before and after the swap.
For (a) we have $(x+i)-(t-x-i)$. 
For (b) there is an increase (by 1) of the mixing indices of each $\alpha$-node in $N_u-I-\{u\}$
since $u$ becomes a $\beta$-node. Moreover there is a decrease (by 1) of
the mixing index of the $\beta$-nodes (as $u$ ceased to be an $\alpha$-node). Hence
for (b) we have $x-(t-x-i)$. Overall,  the difference in the sum of the mixing indices of the nodes in
$N_u-I$ before and after the swap is
$x-(t-x-i)-(t-x-i)+(x+i)=4x-2t+3i$.
A similar argument shows that the difference in the sum of the mixing indices
of nodes in $N_v-I$ is $2t-3i-4y$. Hence overall (and since the nodes outside 
$N_u\cup N_v$ maintain the same mixing index before and after the swap)
the difference in the (total) mixing index is $4(x-y)$. Since $x<y$ this means that
a decrease by at least 4 occurs due to the swap.
\end{proof}

In our analysis, one of the basic facts used is that that dormant states have at least a reasonably high mixing index. If we can show that with high probability the process reaches a point where the mixing index is too low for dormant states to be accessible, then by Corollary \ref{coro:inevcompsegd1} we will have shown that with high probability complete segregation is the eventual outcome. Proposition \ref{prop:maxmixindor} below provides an appropriate bound for the mixing index of dormant states. First we prove a technical lemma, which will then be used in the proof of Proposition \ref{prop:maxmixindor}.

\begin{lem}\label{le:blockdorml}
Suppose that  $\tau>0.5$, $\rho_{\ast}<\tau$, and $0\ll w \ll n$. In a dormant state of the process $(n, w, \tau, \rho)$ 
every $\beta$-block has length at most
$2\ceil{(1-\tau)w}$ and every $\beta$-node is
$\ceil{(1-\tau)w}$-near to an $\alpha$-node.
\end{lem}
\begin{proof}
Since the second claim implies the first, it suffices to prove the second claim. By Lemma \ref{coro:exiunhangen1} we can assume that there are unhappy $\beta$-nodes
in the given state. 
For a contradiction, suppose that
some $\beta$-node is not
$\ceil{(1-\tau)w}$-near to any $\alpha$-node.
Consider the $\alpha$-node
which is adjacent to the block and to the right of it.
For large $w$,  $2\ceil{(1-\tau)w}+1 < w$, 
meaning that this  $\alpha$-node has at least  
$2\ceil{(1-\tau)w}+1$ nodes of type $\beta$
in its neighbourhood. 
Hence the $\alpha$-node has at most 
$2w-2\floor{(1-\tau)w}$ nodes of type $\alpha$ in its neighbourhood, 
which is less than  $(2w+1)\tau$.  The fact that this  $\alpha$ node is unhappy means that the state is not dormant.  
\end{proof}

\begin{prop}[Mixing in dormant states]\label{prop:maxmixindor}
Suppose that  $\tau>0.5$, $\rho_{\ast}<\tau$, and $0\ll w \ll n$. 
The mixing index in a dormant state of the process $(n, w, \tau, \rho)$ is more than $n(w+1)\tau\rho_{\ast}$. 
\end{prop}
\begin{proof}
Suppose that in a dormant state the mixing index is at most
$n(w+1)\tau\rho_{\ast}$.
Since there are $n\rho_{\ast}$ nodes of type $\beta$,  
there exists such a node $u$
with mixing index at most
$(w+1)\tau$. By Lemma \ref{le:blockdorml}
there exists an $\alpha$-node $v$ within $\ceil{(1-\tau)w}$ nodes 
to the left or to the right
of $u$. The number of $\alpha$-nodes in the 
neighbourhood of $\nu$ is therefore at most 
$(w+1)\tau +\ceil{(1-\tau)w}$. However this same number 
must be at least $(2w+1)\tau$ since 
$v$ is happy in a dormant state. This holding for arbitrarily  large $w$ would imply that $(1-\tau)\geq \tau$ which gives the required contradiction. 
 \end{proof}
 
We do not know if the bound provided by
Proposition \ref{prop:maxmixindor} is tight.
However it is 
sufficient for the proof of Theorem \ref{th:complsegm}, which only requires
a bound that is proportional to the population size $n$.

\subsection{Number of unhappy nodes and maximal blocks}\label{se:nounhapblo}
While a low mixing index suffices to establish the inaccessibility 
of dormant states, in fact it will often be more convenient to work 
directly with the number of unhappy nodes. 
The aim of this subsection is to allow us to do this, 
by establishing a fairly tight relationship between the number of unhappy nodes and the mixing index.  

As another measure of mixing, we may consider the number
$\mathtt{k}_{\beta}$ of maximal contiguous $\beta$-blocks in the state. 
Let $\beta_i$ be the $i$th  node of type $\beta$ 
and let $\beta_i^{\alpha}$ denote the number of  $\alpha$-nodes 
in the neighbourhood
around $\beta_i$. Let $[x,y]$ be a finite interval of integers such that 
$\{\beta_i : i\in [x,y] \}$ constitutes a block (i.e. there is no $\alpha$-node
between $\beta_x$ and $\beta_y$). 
If $x-y\geq w$ then $\beta^{\alpha}_x+\cdots+\beta^{\alpha}_y$ 
is bounded above by $2\cdot(1+\cdots+w)=w(w+1)$.
If $x-y< w$ the number $w(w+1)$ continues to be a bound for 
$\beta^{\alpha}_x+\cdots+\beta^{\alpha}_y$. Therefore

\begin{equation}\label{eq:mixindtoblock}
\sum_{i<n_\beta} \beta_i^{\alpha} \leq w(w+1) \cdot \mathtt{k}_{\beta},\hspace{0.4cm}\textrm{where $\mathtt{k}_{\beta}$ is the number of maximal $\beta$-blocks.}
\end{equation}

This inequality is a formal expression of the rather obvious 
fact that the fewer maximal $\beta$-blocks there are, the less mixed the two types are.
By the definition of happy nodes,
if $\tau>0.5$ and $w>(1-\tau)/(2\tau-1)$ then  no
two adjacent nodes of different types can both be happy.
This means that, as we move around the circle of nodes, 
every time we cross the border between 
a maximal $\beta$-block and a maximal $\alpha$-block we may count an 
additional unhappy node.
So, provided that $\tau>0.5$ and  $w$ is sufficiently large,
the number of maximal $\beta$-blocks 
is bounded above by the number of unhappy nodes in the state.
Then by \eqref{eq:mixindtoblock} we get

\begin{equation*}
\mix \leq w\cdot (w+1) \cdot \mathtt{k}_{\beta} \leq w\cdot (w+1) \cdot \unhap
\end{equation*}

Intuitively this inequality says that the only way to 
have a small number of unhappy nodes is a small mixing index, i.e.\ a large
degree of segregation. On the other hand we may bound the number of 
unhappy nodes in terms of the mixing index. By
\eqref{eq:gjrisrig} and the definition of unhappy nodes 
\begin{equation*}
\unhap_{\alpha}\cdot (1-\tau)(2w+1)\leq  \mix\hspace{0.5cm}
\textrm{and}\hspace{0.5cm}\unhap_{\beta}\cdot (1-\tau)(2w+1)\leq \mix
\end{equation*}
where $\unhap_{\alpha}, \unhap_{\beta}$ 
are the numbers of unhappy nodes of type $\alpha$ and $\beta$ respectively.
So
\[
\mix \leq w\cdot (w+1) \cdot \mathtt{k}_{\beta} \leq w\cdot (w+1) \cdot \unhap
\leq \mix\cdot \frac{2w(1+1/w)}{(1-\tau)(2+1/w)}< \mix\cdot \frac{2w}{1-\tau}
\]
and
\[
\frac{1}{w}\cdot \frac{\mix}{w+1} \leq \mathtt{k}_{\beta} \leq\unhap< \frac{2}{1-\tau}\cdot \frac{\mix}{w+1}
\]
which means that if $\tau>0.5$ (and $w$ is sufficiently large) 
then $\unhap =\Theta(\texttt{k}_{\beta})=\Theta(\mix)$.

\subsection{Background on probability}
We make use of the various concentration of measure 
inequalities for random variables and (super)martingales.
The simplest of these is Markov's inequality,
which says that if $X$ is a non-negative random variable with 
$\mathbb{E}(X)=\mu$ and $a>0$ then $P(X>a\mu)\leq 1/a$.
Recall Hoeffding's inequality for independent Bernoulli trials.

\begin{lem}[Tight Hoeffding for Bernoulli variables]\label{le:hoeffdingber}\label{le:tighthoeffd}
Let $Z_i$ be independent Bernoulli trials with expected value $p$, and let 
$S_k=\sum_{i<k} Z_i$. Then 
$\proba{S_k\leq k(p-\epsilon)}\leq e^{-2\epsilon^2 k}$
and 
$\proba{S_k\geq k(p+\epsilon)}\leq e^{-2\epsilon^2 k}$
for each $\epsilon>0$. If $p\leq 1/2$
then $\proba{S_k\geq k(p+\epsilon)}\geq 1/4\cdot e^{-2\epsilon^2 k/p}$
for each $\epsilon>0$ such that $\epsilon\leq 1-2p$.
\end{lem}

The second clause of this lemma (the tightness of the inequality)
follows from Slud's inequality \cite{sludche} (which
gives a lower bound of the binomial upper tail in terms of the upper tail of the
normal distribution) and
standard lower bounds for upper tail of the normal distribution (see \cite{mousavi} for more details).

Since there are complex 
dependences amongst the random variables of the Schelling process,
we often need to `approximate' certain processes with
canonical processes like simple random walks.
Here a random walk with respect to the
integer-valued random variables $(Z_i)$ 
is the stochastic process $R_k= r+\sum_{i<k} Z_k$, for
some $r\in\mathbb{N}$.
We say that $(R_i)$ {\em is ruined} 
at step $k$ if $k$ is the least number such that 
$R_k\leq 0$. The following simple fact is obtained via a standard 
coupling  argument.


\begin{lem}[Random walk simulation]\label{le:ranwalksim}
Let $t_0,t_1\in\Nat$, $X_i\in\{-t_0,0,t_1\}$ be (possibly dependent) random variables, let
$\hat{X}_i\in\{-t_0,0,t_1\}$ be independent
Bernoulli trials and let $Y_k= \sum_{i<k} X_k$, 
$\hat{Y}_k=\sum_{i<k} \hat{X}_k$ be the associated random walks. 
Provided that, no matter what occurs at stages prior to $i$, at stage $i$ we have 
$\proba{X_i=-t_0}\leq \mathbb{P}[\hat{X}_i=-t_0]$ and $\proba{X_i=t_1}\geq \mathbb{P}[\hat{X}_i=t_1]$, then 
for all $k,x\in\Nat$
the probability that 
$(Y_i+x)$ is ruined by step $k$ is bounded above by the probability that
$(\hat{Y}_i+x)$ is ruined by step $k$.
\end{lem}

The following fact about biased random walks is folklore.

\begin{lem}[Biased random walks]\label{le:ranwalkbiasked}
Let $t_0,t_1,r \in\Nat$, and let $X_i\in\{-t_0,0,t_1\}$ be (possibly dependent) random variables such that at stage $i$, 
no matter what has occurred at previous stages, 
$\probac{X_i=t_1}{X_i\neq 0}>t_0/(t_0+t_1)+\delta$ for some $\delta>0$. 
Let $Y_j= r+\sum_{i<j} X_i$, 
be the associated random walk. Then the
probability that 
$(Y_j)$ is ever ruined is bounded above by $e^{-2r\delta^2/t_0}/(1-e^{-2\delta^2})$.
\end{lem}
\begin{proof}
 Let
$Z_i\in \{-t_0, t_1\}$ be independent variables such that
$\proba{Z_i=-t_0}= t_1/(t_0+t_1)-\delta$. 
Let $G_i=r+\sum_{j<i} Z_j$ be the associated random walk.
Then 
$\proba{\hat{Y}_i=-t_0}\leq\proba{Z_i=-t_0}$, so by Lemma \ref{le:ranwalksim}
it suffices to show that the probability that $(G_j)$ is ruined is bounded 
above by $e^{-2r\delta^2/t_0}/(1-e^{-2\delta^2})$.

We may view $Z_i$ as independent Bernoulli trials, 
where $Z_s=t_1$ is viewed as success and
$Z_s=-t_0$ is viewed as failure. Let 
$p=\proba{Z_i=t_1}$, so $p = t_0/(t_0+t_1) +\delta$. 
If $k_s$ is the number of successes up to step $s$,
then $G_s = r+t_1 k_s-(s-k_s)t_0$ so ruin of the random walk $G_j$ at step $s$ implies
that $k_s\leq (t_0s-r)/(t_0+t_1)$. We may 
use Lemma \ref{le:hoeffdingber}
in order to bound the probability of this event.
If we let $\delta_s=p-(t_0-r/s)(t_0+t_1)$ note that $\delta_s>p-t_0/(t_0+t_1)=\delta$, so 
by  Lemma \ref{le:hoeffdingber},
 $e^{-2\delta^2s}$ is an upper bound for the probability that $(\hat{Y}_j)$ is
ruined at step $s$.  Next, note that  $(G_j)$ can only be ruined at
stages $>r/t_0$. Hence 
\[
\sum_{s\in[ r/t_0, m)} e^{-2\delta^2s}\leq 
\frac{e^{-2r\delta^2/t_0}}{1-e^{-2\delta^2}}
\]
is an upper bound of the probability that $(G_j)$ is ever ruined (before stage $m$), which
concludes the proof.
\end{proof}

Our analysis depends on various exponential bounds that we can obtained
on the expectations of 
certain parameters (e.g.\ the number of unhappy $\alpha$-nodes).
The following fact will be routinely used in order to express such bounds in 
a canonical form.
In the following statement 
the variables $Z_s$ concern stage $s$ of the Schelling process
$(n, w, \tau, \rho)$ and
the constants $q,q',p$ are independent of 
$n, w$. 
\begin{lem}[Expectation bounds]\label{le:boundstech}
Let $f$ be a polynomial, $p<1$ and $Z_s$ a random variables such 
that $\mathbb{E}(Z_s)<np$ for all $s$.
If $\mathbb{E}(Z_s)\leq n \cdot f(w)\cdot  e^{-wq}$ for some $q>0$ and all
all $s$ and all sufficiently large $w$ then
there exists $q'>0$ such that 
$\mathbb{E}(Z_s)\leq n\cdot e^{-wq'}$ for all $w,s$.
\end{lem}
\begin{proof}
Since $f$ is a polynomial, we can choose $q_0>0$ and $w_0$ such that
$n f(w) e^{-wq}< n e^{-wq_0}$ for all $w>w_0$. 
Hence $\mathbb{E}(Z_s)\leq n\cdot e^{-wq_0}$ for all $s$ and all $w>w_0$.
We may choose $q'<q_0$ such that $p<e^{-wq'}$ for all $w\leq w_0$.
Then by the assumption on $p$ we have that   
$\mathbb{E}(Z_s)\leq n e^{-wq'}$ for all $w$ and all $s$.
\end{proof}

The binomial distribution with
$t$ trials and success probability $p$ is denoted by $B(t,p)$, and 
$Z\sim B(t,p)$ means that  random variable $Z$ follows this distribution.
Stirling's formula 
asserts that $n! \approx n^{n + \frac{1}{2}} e^{-n}$,
i.e.\ that the limit of the ratio of the two expressions tends to 1 as
$n$ tends to infinity.

\begin{lem}[Stirling's approximation]\label{le:stirlinggen}
There exists a  
polynomial $y\mapsto p(y)$ such 
that for all $k\in\Nat$ and all $x\in\mathbb{R}\cap (0, k)$
\[
\textrm{there exists}\hspace{0.3cm}
q\in \left(\frac{1}{p(k)}, p(k)\right)\hspace{0.4cm}
\textrm{such that}\hspace{0.4cm}
\binom{k}{\ceil{x}}=q\cdot \binom{k}{\floor{x}}.
\]
\end{lem}
\begin{proof}
Let $z=x$ or 
$z=k-x$. Also let  $z'=\ceil{z}$ or $z'=\floor{z}$.
Then according to the definition of the binomial coefficient
it suffices to show that there exists a polynomial $y\mapsto r(y)$
such that 
\[
z'!=q\cdot z\hspace{0.4cm}
\textrm{for some}\hspace{0.4cm}
q\in \left(\frac{1}{r(k)}, r(k)\right).
\]
Note that
there exists $\delta\in (-1,1)$ such that $z'=z+\delta$. Then
\[
(z+\delta)^{z+\delta+\frac{1}{2}}=z^z\cdot (z+\delta)^{\delta +\frac{1}{2}}
\cdot \left( 1+\frac{\delta}{z}\right)^z.
\]
The second term on the right side of the equation 
is bounded by a polynomial in $k$ while the third term is in
$(e^{-1}, e)$. Hence there is a quadratic polynomial $y\mapsto r(y)$ such that
\[
(z+\delta)^{z+\delta+\frac{1}{2}}\in \left(z^z\cdot r(k)^{-1},\ z^z\cdot r(k)\right). 
\]
By Stirling's approximation it follows that there exists
a quadratic polynomial $y\mapsto p(y)$ such that for all $k$, $x\leq k$ and $z, z'$ as defined above
there exists $q\in (1/p(k), p(k))$ such that
$z! = q\cdot z'!$.
This fact, along with the definition of the binomial coefficient, implies
the required statement.
\end{proof}

In our analysis of the Schelling process for the case when $\tau\leq 0.5$
we will need to compare the tails of different binomial distributions.
There are a number of ways for doing this (including using approximations with the
normal distribution) but the simplest is the following
elementary fact from \cite[Theorem 1.1]{Bollobas_random}.

\begin{lem}[Tails of the binomial distribution]\label{lem:binom}
Suppose that $X_N \sim B(N,p)$, $p,k \in (0,1)$ and
for all sufficiently large $N$,
$( 1+  k(1-p)/p )\cdot h(N)> N \geq h(N)> p\cdot N > 0$,
where
$h: \Nat \to \Nat$. Then
\[
\proba{X_N = h(N)} \ \ \leq \ \ 
\proba{X_N \geq h(N)} \ \ \leq \ \ 
\left( \frac{1}{1-k} \right) \cdot \proba{X_N = h(N)}
\]
for all sufficiently large $N$.
In asymptotic notation we have 
$\proba{X_N \geq h(N)}= \Theta \left( \proba{X_N = h(N)} \right)$.
\end{lem}

The combination of this result with Stirling's approximation of the binomial
coefficients gives the required information about the 
asymptotic behaviour of the ratio of the two binomial probabilities
of interest (unhappy nodes and stable intervals).

\subsection{Martingales in the Schelling process}
A crucial part of our analysis is based on two supermartingales, one regarding the 
non-anomalous $\alpha$-nodes in the infected area, and one regarding the anomalous nodes.
The latter is somewhat sophisticated, in  the sense that it is not {\em adapted} to the stages of the process.
Nevertheless it is a supermartingale relative to a more general process, and this is sufficient for our analysis.
Due to this sophistication, we clarify how we regard the process $(n,w,\tau,\rho)$ in probabilistic terms,
and what we mean by a martingale.

The states of the system are all configurations of $n$ nodes that can have one or the other type.
A state $B$ is accessible from another state $A$ (thought as an arrow from $A$ to $B$) if an application of a legitimate swap on
$A$ gives $B$.
We view the random process as a combination of two parts. The first is the production of the initial state according to the given probability
distribution of the two types. The second is the stochastic process that starts from the initial state and  moves to the next
state, choosing uniformly randomly from all  the (finitely many) currently accessible states.
We denote the initial state by $F_0$ and the state at stage $s$ by $F_s$. 
The remaining discussion refers to the second part of the process, where $F_0$ is a constant.
The underlying probability space $\Omega$ is the
set of all infinite sequences of states, which start with $F_0$ and have the property that each term is a 
state which is accessible from its predecessor. We also add into $\Omega$ the finite sequences of states,
which start with $F_0$, each of their terms is accessible from its predecessor, and its last term is an absorbing state.
We view this as a tree, where the $i$th level of the tree
(prefixes of points in the space of length $i$) describes all possible outcomes of the process up to stage $i$.
This tree has dead-ends, namely the
absorbing states. The probability measure on $\Omega$ is the uniform one, namely the one induced by splitting the total measure
$1$ uniformly inductively starting from the route and considering all accessible paths.
Then each $F_i$ can be viewed as random variable on $\Omega$,  which takes any point in the space and outputs its $i$th term.

A number of other processes will be defined, relative to the process $(F_i)$ which contains all the information. Clearly 
$(F_i)$ is memoryless (has the Markov property) since the distribution of $F_i$ only depends on the value of $F_{i-1}$.
The secondary processes that we consider in our analysis (like $\ZZ_s$ or $\GG_s$) 
can be seen as recording only part of the information of the full process
$F_0,\dots, F_s$ up to stage $s$. 
In general, a process $X_s$ is {\em adapted to} (or defined in terms of) another process 
$J_s$ if there is a function such that
$f(J_s)=X_s$ for every point in $\Omega$. 
Recall that a {\em filtration} on $\Omega$ is an increasing sequence of $\sigma$-algebras on $\Omega$.
The reader who is used to working with filtrations (especially with respect to martingales) can equivalently view
a process $X_s$ adapted to another process $J_s$ as $X_s$ adapted to the {\em natural 
filtration} $(\mathcal{J}_s)$ of $(J_s)$: this is the filtration generated by the inverse images 
of the Borel sets of $\Omega$, with respect to the variables
$J_s$.
For example, the natural filtration of the full process $(F_s)$ is 
$(\mathcal{F}_s)$ where $\mathcal{F}_s$ is the $\sigma$-algebra generated
from the maximal branches of $\Omega$ restricted to  strings of length $s$ or less. 
Intuitively $\mathcal{F}_s$ can measure all events that can possibly
happen up to stage $s$.   

In order to show that a certain process is a martingale, we will have to adapt to another suitable process. 
Equivalently, we would have to adapt
it to a suitable filtration (which may be different than the standard filtration $(\mathcal{F}_s)$ that we described above).
This is the reason for introducing adapted processes: the simplest martingale notion corresponds to 
processes adapted to themselves, and is not sufficient for our proof. 
Recall that a process $H_s$ is a supermartingale relative to a Markov process $J_s$ if it is adapted to it and
$\expec{H_{s+1}}{J_s}\leq H_s$ for all $s$. This means that relative to the set of reals in $\Omega$
which have the particular value of $J_s$ (which is regarded as fixed) the expectation of $H_{s+1}$ is bounded by
$H_s$ (which is a function of $J_s$). This is the standard definition of conditional expectation in terms of processes.
In our analysis we occasionally need to consider $\expec{H_{s+1}}{J_s}$ conditional on a set of reals 
$A\subseteq\Omega$. We denote this by $\expeco{H_{s+1}}{J_s}{A}$.
A {\em stopping time} with respect to a process $(J_s)$) 
is a random variable $T$ such that the truth of the event $T=k$ (for any integer $k$) 
is a function of $J_i, i\leq k$. 
If $T$ is a stopping time for $(J_s)$ and $(H_s)$ is a supermartingale with respect to 
$(J_s)$, then the {\em stopped process}
$H_{s\wedge T}$ (which proceeds as $H_s$ up to stage $T$, and then it is constantly 
equal to $H_T$) is also a supermartingale (with respect to $(J_s)$). 
Doob's {\em maximal inequality for supermartingales} says that if $(H_s)$ is a non-negative 
supermartingale with respect to another process $(J_s)$ and  $\expe{H_0}=\mu$,  $a>0$ then
$\proba{\sup_s H_s\geq a\mu}\leq 1/a$.

\subsection{Probability in Schelling segregation}
In this section we lay out a general way for arguing about the probability
of the various properties that a node can have in the initial configuration.

\begin{defi}[Rare and common events in the initial configuration]
A property of a node in the initial configuration is called rare 
(or a rare event) if it holds with 
probability at most $n\cdot e^{-\delta w}$,  
for some positive constant $\delta$ which may depend on $\tau,\rho$ but not on $w,n$.
A property whose negation is rare is called common.
\end{defi}

\begin{defi}[Local properties]\label{de:localevent}
A local property $P_u$ of a node $u$ in the initial configuration is one that
only depends on the nodes that are at most
$f(w)$-far from $u$, where $f$ is a fixed function.
In other words the property is local if given any two nodes $u,v$
such that for all $i \in [-f(w),f(w)]$, $u+i$ is of the same 
type as $v+i$, then $P_u$ holds iff $P_v$ holds.
In this case we say that $P_u$ is $f$-local.
\end{defi}

Note that the two probabilities mentioned in Lemma 
\ref{le:sllnschel} are on different spaces. The first one refers to the product space
where a point is an infinite series of initial states. The second one refers to
the space of points on a random initial state.

\begin{lem}[Strong law of large numbers for the Schelling process]\label{le:sllnschel}
Given a local property $P_u$ of nodes in the initial state of 
the process $(n,w,\tau,\rho)$, with probability one, as $n\to\infty$
the proportion of nodes $u$ that satisfy $P_u$ tends to the probability of $P_u$.
\end{lem}
\begin{proof}
Let $p$ be the probability of $P_u$ and let $f$ be the function
indicating the area around $u$ on which $P_u$ depends (as in
Definition \ref{de:localevent}).
We wish to use the strong law of large numbers, so we need to
manufacture a series of independent trials of properties with given
expectation.
Let $m\in\Nat$ be a parameter that depends on $n$ 
(to be specified shortly).
 We consider the ring
as a union of intervals of length $mf(w)+2f(w)$ 
(which we think of an interval of length $mf(w)$ with padding $f(w)$ 
nodes on each side).
We always assume that $mf(w)+2f(w)<n$. Starting from node $0$, 
denote the $i$th such interval by $V_i$ so that $|V_i|=mf(w)+2f(w)$.
Also, denote the subinterval of $V_i$ 
that results from deleting the 
$f(w)$-node prefix and the $f(w)$-node suffix of $V_i$ by $I_i$.
Hence $|I_i|=mf(w)$.
Let $M_n\in\Nat$ be the largest integer 
such that $M_n(mf(w)+2f(w))\leq n$, so that $M_n\to\infty$ as $n\to\infty$ and
$n-M_n(mf(w)+2f(w))<mf(w)+2f(w)$. 
Hence for each $i<M_n$, the intervals 
$V_i$ are defined and are disjoint. 
The same is true for $I_i$, $i<M_n$.
Moreover, if $S$ is the set of all nodes,
\begin{equation}\label{eq:coverwIi}
2f(w)M_n \leq |S-\cup_{i<M_n} I_i| < 2f(w)M_n+mf(w)+2f(w).
\end{equation}
For each $i<M_n$ let $Y_i$ be the number of 
nodes 
$u\in I_i$ such that $P_u$ holds, and note that these
random variables are independent.
Moreover, by linearity of expectation, 
$\mathbb{E}(Y_i)=pmf(w)$.
Recall that $M_n\to\infty$ as $n\to \infty$.
According to the strong law of large numbers, 

\begin{equation}\label{eq:prob1newis}
\frac{\sum_{i<M_n} Y_i}{M_n}\to 
pmf(w) \hspace{0.4cm}
\textrm{as $n\to\infty$, with probability 1.}
\end{equation}
By \eqref{eq:coverwIi}, the required proportion is
\[
\frac{\sum_{i<M_n} Y_i +\zeta f(w)\cdot (2M_n+m+2)}{(M_n +\delta) f(w)(m+2)}=
\frac{\frac{\sum_{i<M_n} Y_i}{M_n} +
\zeta f(w)\cdot (2+\frac{m+2}{M_n})}{(1 +\frac{\delta}{M_n}) f(w)(m+2)}
\]
where $\delta,\zeta$ range in $[0,1)$ (depending on how close
$n$ is to being a multiple of $f(m)(m+2)$).
If we take $0\ll m\ll n$, the ratio $m/M_n$ tends to $0$, so
by \eqref{eq:prob1newis} the
required proportion tends to
\[
\frac{pm f(w)+2f(w)\zeta}{f(w)(m+2)}=
\frac{pm +2\zeta}{m+2}=
\frac{p +\frac{2\zeta}{m}}{1+\frac{2}{m}}.
\]
Since $0\ll m$, the required proportion tends to $p$.
More formally, we may let $m=\log n$. In this case, as $n\to\infty$
we have $m/M_n\to 0$ because $(\log n)^2/n$ tends to $0$.
Moreover $M_n\to\infty$ and $m\to\infty$ when $n\to\infty$ so
the previous argument applies as indicated.
\end{proof}

The following fact concerns pairs of properties
$P$ and $Q$ that a node can have, which may 
both be rare but one (say $P$) occurs with much higher probability than
the other. It asserts that in this case, a random node
$u$ is much more likely to be  nearer to a node $v$ satisfying $P$ than
a node $t$ satisfying $Q$ (although it may be far 
from any node satisfying $P$ or $Q$). In the statement and proof of this
result we use $P_u$ as a Boolean random variable which asserts that 
`$u$ satisfies $P$' (and similar with $Q_u$).

\begin{lem}[Rare properties in the Schelling ring]\label{le:rareprop}  
Let $P_u$, $Q_u$ be $\ell$-local properties of nodes in the initial state
(where $\ell=\ell_w$ is a function of $w$) and for
each node $u$ let
$x_u$ be the first node $v$ to the 
right of $u$ such that either $P_v$ or 
$Q_v$ holds. If $\rho,\lambda$ are the probabilities
of $P_u, Q_u$ respectively,
the probability that $P_{x_u}$ and
there is no node $v$ with $Q_v$ to the left of and at distance at most $\ell$
from $x_u$ tends to a number $\geq\rho/(\rho+\lambda(2\ell+1))$ as $n\to\infty$. 
An analogous result holds when `right' is replaced by `left'. 
\end{lem} 
\begin{proof}
Consider a partition of the ring into 
disjoint neighbourhoods, starting from a node $u_0$ 
as follows. Recall that addition of nodes is always modulo $n$.
Given $u=u_0$, suppose inductively that $u_t$ has been defined.
Then define $u_{t+1}=x_{u_t}+2\ell +1$. This iteration continues 
as long as $u_{t+1}< n$. Let $k_n$ be 
the number
of iterations in this recursive definition 
(i.e.\ the number of terms of the sequence $(u_t)$).
Consider the property
\[
\textrm{$T_u$:\ \ \ \ $P_{x_u}$ holds and no node $v$ to the left of
$x_u$ and at distance at most  $\ell$ satisfies $Q_v$}.
\]
The sequence $(u_i)$ can be seen as independent trials
for this property.
Let $\pi_n$ be the proportion of the terms of $(u_i)$ that satisfy of $T_{u_i}$
in a random initial state. 
Note that $k_n\to\infty$ as $n\to\infty$ with probability 1.
If $\pi$ is the probability of $T_u$, by the strong law of large numbers
we have that $\pi_n\to\pi$ as $n\to\infty$ with probability 1.
Let $\rho_n$ be the 
proportion of nodes that satisfy $P_u$ and let
$\lambda_n$ be the proportion of nodes that satisfy $Q_u$.
Note that we view $\pi_n, \rho_n, \lambda_n$ as random variables
that depend on the initial state. Then
\[
\frac{\rho_n}{\lambda_n}\leq 
\frac{(2\ell+1) \pi_n k_n}{(1-\pi_n)k_n}=
\frac{(2\ell+1) \pi_n}{1-\pi_n}\Rightarrow
\pi_n\geq \frac{\rho_n}{\rho_n+\lambda_n (2\ell+1)}
\]
By Lemma \ref{le:sllnschel} we have $\rho_n\to\rho$ and $\lambda_n\to\lambda$
as $n\to\infty$, which gives the required assymptotic bound.
\end{proof}

\subsection{Initial expectations}\label{sub:initexpect} 
An important part of our analysis relies on the values of
the welfare metrics at the initial state. With high probability, these will be near to
their expected values, which we may compute.
We start with the mixing index.
 \begin{lem}\label{eq:expecmixind}
The expectation of the mixing index in the initial state of $(n,w,\tau,\rho)$ is
$2nw\rho(1-\rho)$.
\end{lem}
\begin{proof}
Consider the random variables $\beta_i^{\alpha}$
and note that $\expe{\beta_i^{\alpha}}=2w(1-\rho)$ for each $i$. 
If $n_{\beta}$ is the number of $\beta$-nodes, the expectation of the mixing index
in the initial state is $n_{\beta}\cdot 2w(1-\rho)$ by the linearity of expectation. 
If we see $n_{\beta}$ as a random variable, its expected value
is $n\rho$. By the rule of iterated expectation, 
the expected value of the mixing index is $2nw\rho(1-\rho)$.
\end{proof}

Note that the expected value of the mixing index in the initial
state is only slightly smaller than the maximum possible
mixing index $n\cdot (2w+1)\cdot \rho_{\ast}(1-\rho_{\ast})$. This is hardly
surprising, as a random state will be almost perfectly mixed, with the occasional
non-uniformities that are implied by randomness (e.g.\ the existence of contiguous blocks
of certain sizes).
 
Next, we are interested in the expected number of unhappy nodes of each type.
It is not hard to see that this depends on whether
$\tau+\rho<1$ or $\tau+\rho>1$ (we will not consider the special case where
$\tau+\rho=1$). 

\begin{lem}[Unhappy $\alpha$-nodes]
\label{prop:boureduninit}
Given $\rho, \tau$ such that $\rho+\tau<1$,
with high probability 
 the number of initially unhappy $\alpha$-nodes in the process  
$(n,w,\tau,\rho)$ is 
$n\cdot  e^{-\Theta(w)}$.
\end{lem}
\begin{proof}
Let $X_j$ be 1 if the $j$th node $u_j$ 
in the initial state is of type $\alpha$ and unhappy, 
and 0 otherwise.
By Lemma \ref{le:sllnschel}, it suffices to show that
$\expe{X_j}$ is $e^{-\Theta(w)}$.
 Recall that the nodes are labelled independently, following 
 a Bernoulli distribution, with the probability 
 of a $\beta$-label being 
 $\rho$. Let $\epsilon=1-\rho-\tau$ 
 which is positive, according to our hypothesis.
 If $u_j$ is an unhappy $\alpha$-node, 
 then the proportion of $\beta$-nodes in its 
neighbourhood $N(u_j)$ 
is larger than $1-\tau$. 
Hence the proportion of $\beta$-nodes in 
$N(u_i(j))-\{u_i(j)\}$ is larger than $1-\tau$,
so it is at least $\rho+\epsilon$. 

Let $A$ be the event that $u_j$ is an $\alpha$-node 
and $B$ the event that $u_j$ is unhappy, so that
$\proba{A}=1-\rho$ and $\proba{A\cap B}= \probac{B}{A}\cdot \proba{A}$. 
If we see the labels of the nodes in $N(u_j)-\{u_j\}$
as a series of $2w$ independent  Bernoulli trials, 
by Hoeffding's inequality for Bernoulli trials 
the probability that the proportion of $\beta$-nodes 
is at least $\rho+\epsilon$ 
is bounded by $e^{-4w\epsilon^2}$. 
Hence by the above discussion,
$\probac{B}{A}<e^{-4w\epsilon^2}$. 
We may conclude that  
$\proba{X_j=1}< (1-\rho)\cdot e^{-4w\epsilon^2}$. Hence
$\expe{X_j}\leq (1-\rho)\cdot e^{-4w(1-\tau-\rho)^2}$. 
Similarly, by Lemma \ref{le:tighthoeffd} we have
$\expe{X_j} \geq (1-\rho)\cdot e^{-4w(1-\tau-\rho)^2/\rho}/4$.
Hence $\expe{X_j}$ is $n\cdot  e^{-\Theta(w)}$,
which concludes the proof.
\end{proof}

A similar argument gives an analogous result for the unhappy $\beta$-nodes.

\begin{lem}[Unhappy $\beta$-nodes]\label{prop:bougreenuninit}\label{prop:bforbetainit}
Given $\tau,\rho$ such that $\rho <\tau$,
with high probability 
 the number of initially happy $\beta$-nodes in the process  
$(n,w,\tau,\rho)$ is $n\cdot  e^{-\bigo{w}}$.
If in addition $\tau+\rho<1$,
with high probability 
this number is $n\cdot  e^{-\Theta(w)}$.
\end{lem}
\begin{proof}
Let $Y_j$ be 1 if $u_j$ is of type $\beta$ and happy, and 0 otherwise.
Then provided that $\rho <\tau$, by Hoeffding's inequality for Bernoulli variables
we have that $\expe{X_j}\leq \rho\cdot  e^{-4w(\tau-\rho)^2}$.
Then Lemma \ref{le:sllnschel} gives the first clause of the claim.
Now lets assume that we also have $\tau+\rho<1$.
Then by the second clause of Lemma \ref{le:tighthoeffd} 
we get $\expe{X_j}\geq  \rho \cdot  e^{-4w(\tau-\rho)^2/\rho}$.
This application is possible with $p=\rho$ and $\epsilon =\tau-\rho$ because
$\rho<0.5$ and $\tau+\rho<1$, which means that $\epsilon <1-2p$.
Then by Lemma \ref{le:sllnschel} 
we get the second clause of the claim.
\end{proof}

By a similar argument we get a bound on the total size of the incubators.

\begin{lem}[Number of incubators]\label{prop:blacksitesbou}
If $\tau+\rho<1$, the probability that a node belongs to an incubator is
$e^{-\Theta(w)}$.
Hence with high probability  
 the number of incubators as well as the number of nodes belonging to incubators
 of the process $(n,w,\tau,\rho)$ is 
$n e^{-\Theta(w)}$.
\end{lem}
\begin{proof}
Let $\epsilon_{\ast}=(1-\tau-\rho)/2$, and let
$X_j$ be the index variable of the event that
the left semi-neighbourhood of the 
$j$th node has less than $(\tau+\epsilon_{\ast})w$ many
$\alpha$-nodes. 
Given that $\tau+\rho<1$, by Hoeffding's inequality
for Bernoulli variables (Lemma \ref{le:hoeffdingber}) and the
tightness of it (Lemma \ref{le:tighthoeffd}),
the probability that $X_j=1$
is $e^{-\Theta(w)}$.
Let $Y_j$ be the index variable of the 
event that the $j$th node belongs to
an incubator, so that the probability that 
$Y_j=1$ is $e^{-\Theta(w)}$ (since $(2w+1)e^{-\Theta(w)}$ is
$e^{-\Theta(w)}$).
Hence $\mathbb{E}(Y_j)$ is  $e^{-\Theta(w)}$.
Then by Lemma \ref{le:sllnschel}
with high probability 
 the number of nodes belonging to incubators
 of the process $(n,w,\tau,\rho)$ is 
$n e^{-\Theta(w)}$.
\end{proof}

\subsection{Accessibility of dormant states}\label{sub:accesdormst} 
It is crucial to
understand the dormant states and assess their accessibility from an initial state.
IWe demonstrate that this issue ultimately depends on the given parameters
$\tau, \rho$. We show  that if $\tau+\rho>1$ then with high probability
we may assert that no dormant state is accessible from the initial state. On the other hand, if 
$\tau+\rho<1$ then with high probability there are permutations of the initial state which are dormant.

\begin{prop}[Existence of dormant states]\label{prop:existdormst}
If $\tau +\rho_{\ast}<1$ then there are permutations of the initial state which are
dormant  (provided that $n>2w+1/(1-\tau-\rho_{\ast})$).
\end{prop}
\begin{proof}
Consider the state where the $\beta$-nodes occur in blocks of length
$\floor{(2w+1)\rho_{\ast}}$, which are divided by blocks of $\alpha$-nodes 
of length at least $\ceil{(2w+1)(1-\rho_{\ast})}$.
Since $\ceil{(2w+1)(1-\rho_{\ast})}=(2w+1)-\floor{(2w+1)\rho_{\ast}}$
and $n>2w+1/(1-\tau-\rho_{\ast})$ we can consider an arrangement such that
all blocks of $\alpha$-nodes have length exactly $\ceil{(2w+1)(1-\rho_{\ast})}$,
except perhaps one which may have longer length. In this state
all $\alpha$-nodes are happy and all $\beta$-nodes are unhappy.
In particular, it is a dormant state.
\end{proof}

\begin{lem}[Existence of unhappy nodes]\label{taurhostarmt1}
Suppose $\gamma\in\{\alpha, \beta\}$
and
let $\theta_{\ast}$ be the proportion 
of $\gamma$-nodes in a state of the  process $(n, w, \tau, \rho)$.
If  $\tau>0.5$ and  $\theta_{\ast}<\tau$, then for  $0\ll w \ll n$ 
there exist unhappy $\gamma$-nodes 
in the state.
\end{lem}
\begin{proof}
Given the parameters $\theta_{\ast}, \tau$,  $w$ which is large,  
and any state 
of  the process $(n, w, \tau, \rho)$ with no unhappy $\gamma$-nodes,
it suffices to produce an upper bound on $n$ (which does not depend on the particular state but only $\theta_{\ast},\tau, w$ and the fact that no $\gamma$-nodes are unhappy).

Let $\delta\in \{\alpha, \beta\} -\{\gamma\}$.
Since $\tau>0.5$  and all
$\gamma$-nodes are happy, there
are no $\delta$-blocks of length $\geq w$.
We may assume that $n> 3w+1$. 
Define the {\em bias} $\bias{I}$ of an interval $I$ of nodes 
to be the difference between the number of
$\gamma$-nodes in the interval and the number of $\delta$-nodes in the interval.
Without loss of generality suppose that the node occupying site $w$ is a
$\gamma$-node (otherwise consider a rotation).
We define a sequence $(u_i)$ of $\gamma$-nodes in the state, 
starting with $u_0=w$. Let $N_i$ denote the neighbourhood of $u_i$.
Given $u_i$, define $u_{i+1}$ to be the rightmost 
$\gamma$-node in $N_i$. 
Since there are no $\delta$-blocks of length
$\geq w$, the sequence $(u_i)$ 
is well defined and it never happens that $u_i=u_{i+1}$. Let $m$ be the largest number such that none of the neighbourhoods $N_i$ for $0<i\leq m$ contain the node at site 0. 
Since $n> 3w+1$ we have $m>0$.
Let $I_m= \cup_{i=0}^m N_i$ and $V_m=\sum_{i=0}^m \bias{N_i}$.
Note that $I_m$ contains all of the nodes except at most $w$. 
Moreover since $u_{i+1}-u_i\leq w$ we have
\begin{equation}\label{eq:bouleIm}
|I_m|\leq 2w+1+mw.
\end{equation}
Let $L_i$,  and $R_i$ be the leftmost and rightmost 
$w$-many nodes in $N_i$ respectively.
Since  $N_i$ contains at least $\tau(2w+1)$ nodes of type $\gamma$:
 
\begin{equation}\label{eq:vlowbou}
\bias{N_i}\geq (2w+1) (2\tau-1)\hspace{0.5cm}\textrm{and}
\hspace{0.5cm} V_m\geq (m+1)(2w+1) (2\tau-1).
\end{equation}
Note, however,  that some nodes have been 
counted multiple times in the sum that defines $V_m$, since
the intervals $N_i$ are not disjoint. For each 
$k\in\Nat$ let $J^m_k$ consist of the nodes in $I_m$ which belong to 
exactly $k$ distinct intervals $N_i$.

By the definition of $(u_i)$, the node $u_{i+2}$ is always outside $N_i$.
Similarly, $u_{i+4}$ is always outside $N_{i+2}$. 
This means that it is not possible for the neighbourhoods
of 5 consecutive terms of $(u_i)$ to have a nonempty intersection. 
This, in turn, implies that 
$J^m_k=\emptyset$ for each $k>4$. 
A similar consideration shows 
that $J^m_4$ consists entirely of $\delta$-nodes
(hence $\bias{J^m_4}\leq 0$).
Next, note that $J_1^m\subseteq L_0\cup R_m$, so $|J_1^m|\leq 2w$.
Hence by counting the multiplicities of the nodes in the sum which defines 
$V_m$, we have 
\begin{equation}\label{eq:breakvm}
V_m = 2\bias{I_m}-\bias{J^m_1}+\bias{J^m_3} +2\bias{J^m_4}\hspace{0.4cm}\textrm{and}\hspace{0.4cm}
V_m \leq 2\bias{I_m}+2w+\bias{J^m_3}.
\end{equation}
Let $N_i'=N_{i-1}\cap N_{i+1}$ and note that $N_i'=R_{i-1}\cap L_{i+1}$.
Moreover let $L_i'=N_i'\cap L_i$ and $R_i'=N_i'\cap R_i$.
By the definition of $(u_i)$ it follows that if 
$R_i'$ is nonempty, then it consists entirely
of $\delta$-nodes.
Since $u_i\in J^m_3$ for each 
$i\in [1, m-1]$, $N_i'=L_{i}'\cup R_i'\cup\{u_i\}$ and 
$J^m_3\subseteq\bigcup_{i\in [1, m-1]} N_i'$, 
we have:
\begin{equation}\label{eq:jm3unravel}
\bias{J^m_3}< m + \sum_{i=1}^{m-1} (|L_i'|-|R_i'|).
\end{equation}
Let $d_i=u_i-u_{i-1}$. 
Then $|R_i'|=w-d_i$ and $|L_k'|=w-d_{i+1}$.
Hence  $|L_i'|=|R_{i+1}'|$ and
\[
\sum_{i=1}^{m-1} (|L_i'|-|R_i'|)\leq |L_{m-1}'|-|R_{1}'|\leq w.
\]
Then from \eqref{eq:jm3unravel} we get
$\bias{J_3^m}< m+w$.
From the second clause of \eqref{eq:vlowbou} and 
\eqref{eq:breakvm} we have
\begin{equation}\label{eq:lowbouthetaimeve}
2\bias{I_m} > (m+1)(2w+1)(2\tau-1)-3w-m.
\end{equation}
If $x_m,y_m$ are the numbers of  $\gamma$ and $\delta$ nodes 
in $I_m$ respectively, then $x_m+y_m=|I_m|$ 
and $x_m-y_m=\bias{I_m}$.
Hence $2x_m=|I_m|+\bias{I_m}$. 
By hypothesis we have $x_m\leq n\theta_{\ast}$.
Moreover, since $n\leq |I_m| +w$ 
we have $x_m\leq (|I_m| +w)\theta_{\ast}$.
Hence
$\bias{I_m}\leq (2\theta_{\ast}-1)|I_m| + 2w\theta_{\ast}$, so
by \eqref{eq:bouleIm}, 
\[
\bias{I_m}\leq  mw(2\theta_{\ast}-1) + 2w(3\theta_{\ast}-1) +2\theta_{\ast}-1.
\]
By \eqref{eq:lowbouthetaimeve}  we may deduce that 
\begin{equation}\label{finalgthasd}
2m\cdot [2w(\tau-\theta_{\ast})-(1-\tau)]< 
w(12\theta_{\ast}-4\tau+1)  +4\theta_{\ast}-2\tau-1.
\end{equation}
We may assume that 
$w$ is larger than $(1-\tau)/[2(\tau-\theta_{\ast})]$.
By this condition and the fact that 
$\tau-\theta_{\ast}> 0$,  
the left side of \eqref{finalgthasd} is positive. 
Also, $n\leq |I_m|+w$, so by \eqref{eq:bouleIm} 
we have $n\leq 3w+1+mw$.
If we combine the latter inequality with \eqref{finalgthasd} we get
\[
n< 3w+1 + w\cdot \frac{w(12\theta_{\ast}-4\tau+1)  +4\theta_{\ast}-2\tau-1}{4w(\tau-\theta_{\ast})-2(1-\tau)}
\]
which is the required bound on $n$.
\end{proof}

Note that in the above result, the lower bound that is required on $w$ depends only on
$\tau, \rho_{\ast}$, while the lower bound that is required on $n$ 
depends on $\tau, \rho_{\ast}$ and $w$.
We may now apply Lemma  \ref{taurhostarmt1} in order to establish the conditional
existence of unhappy nodes of both types.

\begin{coro}[Existence of unhappy nodes]\label{coro:exiunhangen1}
Suppose that $\tau>0.5$, $\rho_{\ast}<\tau$  and $w$ is sufficiently large.
Then for
all sufficiently large $n$, every state of the process
$(n, w, \tau, \rho)$ has  unhappy $\beta$-nodes, and if $\tau+\rho_{\ast}>1$ 
then every state also has unhappy  $\alpha$-nodes.
\end{coro}

Given $\rho$, by the law of large numbers with high probability 
(tending to 1, as $n$ tends to infinity)
$\rho_{\ast}$ will be arbitrarily close to $\rho$. 
Hence we may deduce the absence of dormant states
(with high probability) in the case that $\tau+\rho>1$. 

\begin{coro}[Absence of dormant states]
If $\rho\leq 0.5<\tau$ and $\tau+\rho>1$  then
with high probability none of the accessible states of
the process $(n, w, \tau, \rho)$ is dormant.
\end{coro}

This corollary along with Proposition \ref{prop:existdormst}
establishes the main dichotomy in the analysis of the process.

\subsection{Accessibility of complete segregation or dormant state}\label{se:accesscomplseg}
A central part of our analysis is the fact that from any state there is a transition to either a dormant state or
complete segregation. This is what we prove in this section.
This also means that the only absorbing states of the process are the dormant states.
If $\tau\leq 0.5$ then it is clear that the only absorbing states of the process are the dormant states, since unhappy pairs of nodes of different type can always swap. It is also not difficult to find an appropriate Lyapunov function, establishing that a dormant state must eventually be reached. Consider the  \emph{mixing index} which is non-negative and strictly decreasing in stages for $\tau\leq 0.5$. 
For the case where $\tau>0.5$ more effort is required. We argue in four steps. 
The numbers in what follows are fairly arbitrary. First we show that
from a state with few unhappy nodes of one type (here $5w^4$ is a convenient upper bound of
what we mean by `few', which is by no means optimal) there is a series of transitions which lead to either
a state with a contiguous block of length $2w$ or a dormant state.
Second, 
(assuming that $\tau>0.5$) from a state with a contiguous block of length $\geq 2w$
there is a series of transitions to complete segregation or to a dormant state.
Third,
any state which has at least $2w^4$ unhappy nodes of each type,
there is a series of transitions to a state with a contiguous 
block of length at least $w$,  and at least $w^4$ unhappy nodes of each type.
Finally 
(if $\tau> 0.5$) from a state that has a contiguous block  of length $\geq w$
and at least $4w$ unhappy nodes of opposite type from the block,  there is a series of
transitions to a state with a contiguous block  of length $\geq 2w$.
The combination of these four statements constitutes a strategy for arriving at a dormant state or
a state of complete segregation, from any given state.

In the following arguments we will often make use of the following two rather simple facts
that hold when  $\tau>0.5$.
One is that (if $w>(1-\tau)/(2\tau-1)$), 
any $\beta$-node that is adjacent to a happy $\alpha$-node is unhappy. 
The second concerns the situation where
next to a happy $\alpha$-node there is a $\beta$-node, and we swap
the $\beta$-node for another $\alpha$-node. Then, provided that before the swap the 
the second $\alpha$ node is outside the neighbourhood of
the $\beta$-node, both $\alpha$-nodes will be happy after the swap.

\begin{lem}[Shortage of unhappy nodes]\label{le:shortunhapno}
Suppose that $\tau>0.5$ and that $0\ll w \ll n$. From a state with less 
than $5w^4$ unhappy nodes of one of the types, there is a series of transitions
to either a dormant state or to a state containing a contiguous block of length
at least $2w$.
\end{lem}
\begin{proof}
Without loss of
generality suppose that the state has 
less than $5w^4$ unhappy $\alpha$-nodes. Since $\rho_{\ast}\in (0,1)$, and $0\ll w \ll n$, 
if there does not already exist a contiguous block of length $2w$ then 
there exists an interval $[u,v]$ of $2w$ nodes which contains at least one $\alpha$-node 
and such that
any unhappy $\alpha$-node is distance at least
$2w^2$ from any node in
$[u,v]$. 
Any unhappy $\alpha$ node which cannot see any node in  
$[u,v]$ can move to any position in $[u,v]$ that is adjacent to an $\alpha$-node  (because by doing so, it becomes happy and because if a swap is legal for one member of a potential swapping pair then it is legal for both). 
Hence we can start successively replacing the $\beta$-nodes in
$[u,v]$ which are adjacent to $\alpha$-nodes, with unhappy $\alpha$-nodes, each time choosing unhappy $\alpha$-nodes that have maximal distance from $u,v$.
Note that
this recursive procedure is valid because all $\alpha$-nodes in $[u,v]$ are happy after each swap.
Ultimately we either 
run out of unhappy $\alpha$-nodes,
or else $[u,v]$ becomes an $\alpha$-block. 
\end{proof}

\begin{lem}[Toward a block of length $w$]\label{le:toblolw}
Suppose $\tau>0.5$. If $0\ll w \ll n$ then from any state which has at least $2w^4$ unhappy nodes of each type,
there is a series of transitions to a state with an
$\alpha$-block or $\beta$-block of length at least $w$, and at least $w^4$ many unhappy nodes of each type.
\end{lem}
\begin{proof}
Suppose that we are given a certain state of the process.
Define a sequence $u_i,$ $ i\leq w^2$ of $\alpha$-nodes with neighbourhoods
$N_{u_i}$  respectively, by induction as follows. Let $u_0$ be the least $\alpha$-node
whose neighbourhood contains the minimum number of $\alpha$-nodes amongst all
neighbourhoods of $\alpha$-nodes. If $u_i$ is defined and $i<w^2$, define $u_{i+1}$
to be the least $\alpha$-node whose neighbourhood is disjoint from 
$\cup_{j\leq i} N_{u_i}$ and whose neighbourhood contains the 
minimum number of $\alpha$-nodes amongst all $\alpha$-nodes with the same
property (i.e.\ with neighbourhoods that are disjoint from $\cup_{j\leq i} N_{u_i}$).
This completes the definition of $(u_i)$, which is sound provided that $n$ 
is sufficiently large.
We define a sequence $v_i,$ $ i\leq w^2$ of $\beta$-nodes 
with neighbourhoods $N_{v_i}$ respectively, in a way entirely analogous to the above
definition, ensuring also that all neighbourhoods $N_{u_i}$ and $N_{v_j}$ are disjoint.  

The sequences $(u_i)$ and $(v_i)$ provide a pool of nodes which will be used
for legitimate swaps in a series of transitions which will lead to the desired state
of the process. We start by considering an interval $J$ of nodes of length $3w$
which is disjoint from $\cup_{j\leq w^2} N_{u_i}$ and disjoint from 
$\cup_{j\leq w^2} N_{v_i}$.
Such an interval exists, provided that $n$ is sufficiently large. Let $I$ 
consist of the $w$-many nodes
in $J$ that are at distance at least $w+1$ from any node outside the interval.
Clearly any swap that occurs between a node in $I$ and
one of the nodes $u_i$, does not affect the composition of the neighbourhoods $N_{u_j}$ for $j\neq i$,
or $N_{v_j}$ for  $j\leq w^2$ (and similarly for a swap between a node in $I$ and one of the $v_i$). 

Let $t_i,$ $ i<w$ be the nodes of $I$ enumerated
from left to right. We shall describe a swapping process,  involving less than $w^2$ swaps.  At the end of this process of legal swaps, all nodes in $I$ will be of the same type, (but which type that is will not be determined until the end of the process). 
This process has $w$-many \emph{steps}, with each step $s$ involving up to $s$ swaps.
Let $\gamma_s$ be the type of $t_s$ at the end of stage $s$. 
Also, let  $V_{s}$ contain the nodes $u_i, v_i, i\leq w^2$ which are of type
$\gamma_{s}$ and have not been involved in a swap by the end of stage $s$.
The construction
is designed so that 
 $\gamma_s$ is the type of all $t_i, i\leq s$ st the end of stage $s$. This feature
 guarantees that at the end of the process, all nodes in $I$ have the same type.
Stage 0 is null (i.e.\ we carry out no instructions at stage 0). 

At stage $s+1$ we check if $t_{s+1}$ has type $\gamma_s$. If so, then we go to the next
stage. If not, then suppose first that $t_{s+1}$ is unhappy. In the case that $t_s$ is happy,
any unhappy $\gamma_s$-node outside $J$  can swap with
$t_{s+1}$ (because an unhappy $\gamma_s$-node moving next to a happy
$\gamma_s$-node cannot decrease its utility). 
In the case that $t_s$ is unhappy, we claim that any node $x$ from
$V_s$ can legitimately swap with $t_{s+1}$. In order to see this,
note that the number of $\gamma_s$-nodes in the neighbourhood of
$t_s$ is at least as large as this number at the beginning of the process.
By the definition of $V_s$, this number is at least as large as the number
of $\gamma_s$ nodes in the neighbourhood of $x$. This means that
if $x$ moves to the place that $t_{s+1}$ occupies, its utility will not decrease.

The last case in the procedure is if $t_{s+1}$ is happy and of type different than
$\gamma_s$. In this case we define $\gamma_{s+1}\in \{\alpha,\beta\}-\{\gamma_s\}$
and swap all $t_i, i\leq s$ with distinct nodes in $V_{s+1}$, starting with $t_s$ and moving
to the left. These are legitimate swaps, as nodes of type $\gamma_{s+1}$
move next to happy nodes of the same type (so their utility is not decreased after the swap).
This concludes the description of the process.

By the end of stage $w-1$, all nodes in $I$ are of the same type. Since we perform less than $w^2$ many swaps, there are less than $2(2w+1)w^2$ many nodes whose neighbourhoods are affected by these swaps. Since $w$ is large, there are therefore at least $w^4$ many unhappy nodes remaining of each type remaining. 
\end{proof}

\begin{lem}[Toward a contiguous block of length $2w$]\label{le:tomonbo2w}
Suppose that $\tau> 0.5$ and $0 \ll w \ll n$.  From a state that has an $\alpha$-block  of length $\geq w$
and at least $w^4$ unhappy nodes of each type, there is a series of
transitions to a state with an $\alpha$-block  of length $\geq 2w$.
The same holds for $\beta$-blocks. 
\end{lem}
\begin{proof}
Consider the given state and assume that there is no $\alpha$-block of
length $\geq 2w$ (otherwise 0 transitions suffice).
Let $[x,y]$ be the longest $\alpha$-block in the given state, and let $J$ consist
of all the nodes that are at distance at least $w$ from the interval $[y-2w,y]$. 

Note that $x-1$ is a $\beta$-node and since $\tau>0.5$ it is unhappy.
Let $z$ be the rightmost $\alpha$-node to the left of $x$. If $z$ is unhappy,
then we may swap it with $x-1$ since its utility will
not decrease. Otherwise, if $z$ is happy, then it is at a distance at most $w$ from $x$ and we may successively swap the $\beta$-nodes in
$(z, x)$, starting from $z+1$ and moving to the right, for an equal number of
unhappy $\alpha$-nodes in $J$. This is possible because each time that we
move an $\alpha$ node next to a happy $\alpha$-node, the new $\alpha$ node becomes happy.  

We repeat this process until an $\alpha$-block
of length $2w$ has been formed. Each step of the process
increases the length of the $\alpha$-block that is adjacent and to the left of $y$.
Therefore the process will terminate. We also perform at most $w$ many swaps, meaning that we shall not run out of unhappy nodes to perform the swaps with. 
\end{proof}

\begin{lem}[Complete segregation or dormant state from long block]\label{le:complsegdromslb}
Suppose  $\tau>0.5$ and that $0\ll w \ll n$. From a state with a contiguous block of length $\geq 2w$
there is a series of transitions to complete segregation or to a dormant state.
\end{lem}
\begin{proof}
Consider any state which is not completely segregated, but which has a contiguous block of length at least $2w$. Without loss of generality, suppose that this is a block  of  $\alpha$ nodes  occupying the interval $[u,v]$, where this interval is chosen to be of maximum possible length. Our aim is to show that from this state, one may legally reach another with a contiguous block of greater length (or else a dormant state).   Now if the nodes $u$ and $v$ are both happy then the length of the interval ensures that all nodes in the block are happy -- this follows by induction on the distance from the edge of the interval by considering the difference between successive neighbourhoods. In this case, if there exists  an unhappy $\alpha$ node $u'$, then  let $t\in \{ u,v \}$ be distance at least $w+1$ from $u'$.  Then $u'$ and the $\beta$ neighbour of $t$ may legally be swapped, increasing the length of the run by at least 1.

 So suppose instead that at least one of the nodes  $u$ and $v$ is not happy, and without loss of generality  suppose that $u$ has \emph{bias} less than or equal to  $v$, where the bias of a node is the number of $\alpha$-nodes minus the number of $\beta$-nodes in its neighbourhood. Then $u$ and  $v+1$ may legally be swapped. Performing this swap causes position $v+1$ to have at least the same bias as $v$ did before the swap, and causes  $u+1$ to have at most the same bias as $u$ did before the swap. Thus, the swap has the effect of shifting the run one position to the right and may be repeated until the length of the run is increased by at least 1, i.e.\  for successive $i\geq 0$ we can swap the nodes $u+i$ and $v+i+1$, so long as the latter is of type $\beta$. The first stage at which the latter is of type $\alpha$ the length of the contiguous block has been increased.  Putting these observations together,  we conclude that from any state which  has a contiguous block of length at least $2w$ it is possible to reach full segregation. 
\end{proof}

Finally, we piece together the above processes in order to show the following
comprehensive statement.
\begin{coro}[Complete segregation or dormant state]\label{coro:inevcompsegd1}
From any state of the process $(n, w, \tau, \rho)$ with $0\ll w\ll n$
there exists a series of transitions to complete segregation or to a dormant state.
\end{coro}
\begin{proof}
The case  $\tau\leq 0.5$, we considered earlier. 
Suppose that $\tau>0.5$. 
We may assume that
$\rho_{\ast}\in (0,1)$, because otherwise
every state is a dormant state. 
If there exist at most $5w^4$ unhappy nodes of each type in the state, 
Lemma \ref{le:shortunhapno} shows how to reach a dormant state or a state
with a contiguous block
of length $\geq 2w$. In the latter case, 
Lemma \ref{le:complsegdromslb} shows that there is a series of transitions
to complete segregation or to a dormant state.

So we may assume that  the given state has more than 
$5w^4$ unhappy nodes of each type.
Then Lemma \ref{le:toblolw} shows how to reach a state with a contiguous block
of length $\geq w$ and at least $w^4$ many unhappy nodes of each type.
Furthermore, from such a state Lemma \ref{le:tomonbo2w} shows how to reach a 
dormant state or a state
with a contiguous block
of length $\geq 2w$. In the latter case,  
Lemma \ref{le:complsegdromslb} shows that there is a series of transitions
to complete segregation or to a dormant state.
This is an exhaustive analysis that establishes a path to a dormant state or complete segregation, from
every state.
\end{proof}

This completes our proof of Theorem \ref{th:complsegm} for the  case that $\tau +\rho>1$.

\subsection{Persistent blocks and unhappy nodes in intervals}
Now we focus on the case where $\tau>0.5$ and $\tau +\rho<1$.
Having established that a low number of unhappy nodes suffices to ensure dormant states are inaccessible, we now wish 
to show that such state is reached, before any dormant state is reached. Since in this case there are always
unhappy $\beta$-nodes, we are only concerned about the existence of unhappy $\alpha$-nodes.
One way to ensure this is to establish the existence of blocks of $\beta$-nodes of length $>w$.

\begin{lem}[Persistent \texorpdfstring{$\beta$}{beta}-block]\label{le:persistb1}
Consider the process
$(n,w,\tau,\rho)$ with $\tau>0.5$  
and let $s_{\ast}$ be the least stage where
the ratio between the very unhappy $\beta$-nodes and the unhappy $\alpha$-nodes
becomes less than $4w^2$ (putting $s_{\ast}=\infty$ if no such stage exists). Then with 
high probability there is a $\beta$-block of length $\geq 2w$ at all
stages $<s_{\ast}$ of the process. 
\end{lem}
\begin{proof}
Let $\epsilon>0$,
let $\delta= 2w/(2w+1)-w/(w+1)$,
and let $y$ be a sufficiently large integer so that
\begin{equation}\label{eq:eqforproai}
e^{-2(y-2w)\delta^2/w}/(1-e^{-2\delta^2})<\epsilon/2.
\end{equation}

Since the initial state is random, as $n\to\infty$ 
the probability that
there is a $\beta$-block of length at least $y$ 
in the initial state tends to 1.
Hence (for sufficiently large $n$) we may assume that
there is a  $\beta$-block of length $\geq y$ 
in the initial state, with probability at least $1-\epsilon/2$.
Fix such a block and note that during the stages it may expand or
retract. It suffices to show that, conditionally on the existence of such a block in the
initial state, the probability
that it shrinks to a block of length less than $2w$ before stage $s_{\ast}$
is bounded by $\epsilon/2$.
Let $\ell_s$ be the length of the block at stage $s$, so that
$\ell_0= y$.
Also let $X_0=0$ and for each $s>0$ let 
$X_s=\ell_s-\ell_{s-1}$. Then $X_s\geq -w$ for all $s$,
and  $\ell_s=\ell_0+\sum_{t\leq s} X_s$.
 Let $Z_s$ be $-w$ if $X_s<0$, and let
 $Z_s$ be $1$ if $X_s>0$ (and $Z_s=0$ if $X_s=0$).
 Also let $Y_s=\sum_{t\leq s} Z_s +\ell_0-2w$, so
if at stage $s$ the length of the $\beta$-block becomes less than $2w$,
the random walk $(Y_i)$ is ruined (by stage $s_{\ast}$). Let 
$(\hat{Y}_i)$ be identical to $(Y_i)$, except for stages after $s_{\ast}$, at which
it remains identical to $Y_{s_{\ast}-1}$.  
Hence it suffices to show that the probability that 
$(\hat{Y}_i)$ is ruined is bounded above by $\epsilon/2$. 
 Let $p_s=\probac{X_s>0}{X_s\neq 0}$ and 
 let $q_s=\probac{X_s<0}{X_s\neq 0}$, so that $p_s+q_s=1$.

Since $\tau>1/2$, as long as the length of the block is at least $w$, the nearest
$\alpha$-node on each side is unhappy.
Moreover these $\alpha$-nodes can swap with any very unhappy $\beta$-node. 
Any such swap at stage $s$ would make $Z_s= 1$.
On the other hand,
the only way that $Z_s=-w$ (i.e.\ the length of the block is reduced)
is that a $\beta$-node from one of the $2w$ outer nodes in the block ($w$ on each side)
is part of a swap at stage $s$ (with an unhappy $\alpha$-node).
Hence according to our hypothesis we have $p_s/q_s>2w$ for all $s<s_{\ast}$. So
\[
\probac{X_s>0}{X_s\neq 0}>\frac{2w}{2w+1}> \frac{w}{w+1}+\delta
\]
which means that the walk $(\hat{Y}_i)$ meets the requirements of Lemma
\ref{le:ranwalkbiasked}. Hence the probability that
$(\hat{Y}_i)$ is ruined is bounded by the expression
on the left-hand-side of \eqref{eq:eqforproai}.
We sum up our argument. Given $\epsilon>0$, 
we start with a block of $\beta$-nodes of length $y$, with probability at least $1-\epsilon/2$.
Conditionally on this starting assumption, 
our argument says that with probability at least $1-\epsilon/2$ this block will continue to have
length more than $2w$ at all stages up to stage $s_{\ast}$. Hence the probability that
there is no $\beta$-block of length $\geq 2w$ at some stage $<s_{\ast}$ is less than $\epsilon$.
\end{proof}

Another tool that was used in our analysis is a bound on the number of unhappy $\beta$-nodes
in the infected area, in terms of the number of $\alpha$-nodes in the infected area.
This is based on the fact that, when the number of $\alpha$-nodes 
in an interval is limited, then the number of unhappy $\beta$-nodes in the
same interval is also limited. 

\begin{lem}[Proportions in a block of nodes]\label{le:blockofpro}
Consider a block of adjacent nodes 
which contains exactly $x$ nodes of type $\alpha$.
Then for each $y\in (0,1)$ there are at most $x/y+2w$ many $\beta$-nodes in the block
for which the proportion of $\alpha$-nodes in their neighbourhood is at least $y$.
\end{lem}
\begin{proof}
We are given a block of 
adjacent nodes $A$.
Let us call a node {\em weak} if it is a 
$\beta$-node for which the proportion of $\alpha$-nodes
in its neighbourhood is less than $y$.
It suffices to show that the number of $\beta$-nodes in $A$ 
which are not weak is at most $x/y +2w$.
If we remove all of the weak nodes from $A$,  thus obtaining a possibly 
different (and shorter) block $B$, then in the resulting configuration there are no weak nodes. 
It then suffices to show that the number of $\beta$-nodes
in $B$ is at most $x/y+2w$.
Note that the number of $\alpha$-nodes in $B$ remains $x$, since we did not remove any $\alpha$-nodes.
Let $b_0<b_1$ be the endpoints of block $B$ and define a finite sequence $(\beta_i)$ of
$\beta$-nodes as illustrated in Figure \ref{fig:partpabBinprl} 
and formally defined as follows. Let $\beta_0$ 
be the leftmost $\beta$-node in $B$ such that the left endpoint of its neighbourhood  is $\geq b_0$ and such that the neighbourhood of $\beta_0$ is entirely contained in $B$ (if there exists such). 
Assuming that $\beta_i$ is defined and there are $\beta$-nodes between the right endpoint of $\beta_i$ and $b_1$,
define $\beta_{i+1}$ to be the leftmost $\beta$-node in $B$ which is to the right of $\beta_i$, whose neighbourhood is disjoint from that of $\beta_i$ and entirely contained in $B$. Let $\beta_i, i<k$ be the sequence defined in this way.
Then $\beta_i, i<k$ have disjoint neighbourhoods, each of them containing at least $y(2w+1)$ nodes of type $\alpha$.
Hence 
$ky(2w+1)\leq x$ so $k(2w+1)\leq x/y$ which means that the number
of nodes that are contained in the union of these
neighbourhoods is bounded by $x/y$. Since these are neighbourhoods of
$\beta$-nodes that are not weak, the number of $\beta$-nodes that are contained in the
union of these neighbourhoods is at most $x/y(1-y)=x/y-x$.

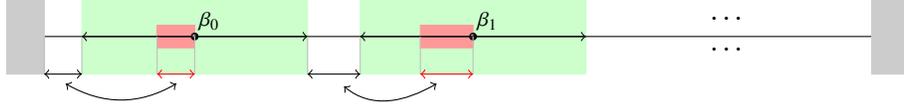
\begin{figure}
 \centering
\begin{tikzpicture}
\draw[green!20!white, fill] (-5,-0.5) rectangle (-2,0.5);
\draw[red!40!white, fill] (-3.5,-0.15) rectangle (-4,0.15);
\draw[green!20!white, fill] (-1.3,-0.5) rectangle (1.7,0.5);
\draw[red!40!white, fill] (0.2,-0.15) rectangle (-0.5,0.15);
\draw (-6,0) -- (6,0);
\draw[gray!40!white, fill] (-6,-0.5) rectangle (-5.5,0.5);
\draw[gray!40!white, fill] (5.5,-0.5) rectangle (6,0.5);
 \draw[gray!50!white] (-5,0) -- (-5,-0.5);
 \draw[gray!50!white] (-5.5,0) -- (-5.5,-0.5);
 \draw[<->] (-5,0) -- (-2,0);
 \draw[<->] (-5.5,-0.5) -- (-5,-0.5);
\shade[ball color=black] (-3.5,0) circle (0.3ex);
\node  at (-3.3,0.2) {{\scriptsize $\beta_0$}};
 \draw[gray!50!white] (-3.5,0) -- (-3.5,-0.5);
 \draw[gray!50!white] (-4,0) -- (-4,-0.5);
\draw[<->, red] (-3.5,-0.5) -- (-4,-0.5);
\node(b2)  at (-5.35,-0.55) {};
\node(r2)  at (-3.61,-0.55) {};
\draw [<->] (b2) to [in=-150, out = -30] (r2); 
\shade[ball color=black] (0.2,0) circle (0.3ex);
 \draw[gray!50!white] (-1.3,0) -- (-1.3,-0.5);
 \draw[gray!50!white] (-2,0) -- (-2,-0.5);
 \draw[<->] (-1.3,0) -- (1.7,0);
 \draw[<->] (-2,-0.5) -- (-1.3,-0.5);
\node  at (0.4,0.2) {{\scriptsize $\beta_1$}};
 \draw[gray!50!white] (0.2,0) -- (0.2,-0.5);
 \draw[gray!50!white] (-0.5,0) -- (-0.5,-0.5);
\draw[<->, red] (0.2,-0.5) -- (-0.5,-0.5);
\node(b1)  at (-1.65,-0.55) {};
\node(r1)  at (-0.15,-0.55) {};
\draw [<->] (b1) to [in=-150, out = -40] (r1); 
\node  at (3.6,0.2) {$\cdots$};
\node  at (3.6,-0.2) {$\cdots$};
\end{tikzpicture}
 \caption{Partition of patched block $B$ in the proof of Lemma \ref{le:blockofpro}}
 \label{fig:partpabBinprl}
\end{figure}

Let $x_i$ be the distance between the right endpoint 
of the neighbourhood of $\beta_i$ and 
the left endpoint of the neighbourhood of $\beta_{i+1}$. 
Note that for each $i<k$ there is a block of at least 
$x_i$ nodes of type $\alpha$ 
 in the the left semi-neighbourhood of $\beta_{i+1}$.
 Indeed, according to the definition of $(\beta_i)$, the only reason why
 there is some distance $d$ between the two endpoints is that a block of
 $\alpha$-nodes of length $d$ immediately to the left of $\beta_{i+1}$. 
We may conclude that there are at least $\sum_i x_i$ 
nodes of type $\alpha$. By the hypothesis the $\alpha$-nodes 
are exactly $x$, so
$\sum_i x_i\leq x$. 
Hence the number of $\beta$-nodes in $B$ that do not belong to the
neighbourhood of some $\beta_i, i<k$ is at most $x+2w$
(where $2w$ is an upper bound for the number of $\beta$-nodes in the final segment of $B$ to the  right of the neighbourhood of $\beta_{k-1}$, or the whole of $B$ if $k=0$).
Hence, overall, there are at most $(x/y-x)+ (x+2w)=x/y+2w$ nodes of type $\beta$ in $B$, which concludes the proof.
\end{proof}

A $\beta$-node is unhappy if and only if the proportion of $\alpha$-nodes in its 
neighbourhood is more than $1-\tau$. Hence we may apply Lemma \ref{le:blockofpro}
with $y$ equal to a value that is slightly larger than $1-\tau$ (taking the limit $y\to 1-\tau$ from above
and taking into account that the number of nodes are integers) gives the following bound on the number of
unhappy $\beta$-nodes in a block.

\begin{coro}[Unhappy $\beta$-nodes versus $\alpha$-nodes]
\label{le:redvsunhapgr}
Consider a block of adjacent nodes 
of type $\alpha$ or $\beta$ such that exactly $x$ of these 
nodes are of type $\alpha$.
Then there are at most $x/(1-\tau)+2w$ unhappy $\beta$-nodes in this block.
\end{coro}

By applying this fact to each of the infected segments of the process, and adding up the numbers
unhappy nodes in each of the segments we see that $\YY_s\leq\ZZ_s/(1-\tau)+2w\CC$, which is the fact used
in the main part of our analysis. 

\subsection{Infected area and random variables}
In this case of unbalanced happiness (i.e.\ when $\tau>0.5$ and $\tau+\rho<1$, see Table \ref{ta:validptwocasesnadcorr})
the unhappy $\alpha$-nodes are initially very rare, so the interesting activity (namely $\alpha$-to-$\beta$ swaps) occurs
in small intervals of the entire population (at least in the early stages). 
These intervals contain the unhappy $\alpha$-nodes, and gradually expand, while outside these intervals
all $\beta$-nodes are very unhappy. 
Figure \ref{fig:biasplot} (produced from a simulation) 
shows the development of this process, where the height of the nodes (perpendicular lines) is proportional to
the number of $\alpha$-nodes in their neighborhood and the horizontal black line denotes the threshold where an $\alpha$-node  becomes unhappy.
These cascades that spread the unhappy $\alpha$-nodes are due to the following domino effect. An unhappy 
$\alpha$-node moves out of a neighbourhood, thus reducing the number of $\alpha$-nodes in that interval. This
in turn often makes another $\alpha$-node in the interval unhappy, which can move out at a latter stage, thus
causing another $\alpha$-node nearby to be unhappy, and so on.
The expanding intervals are the {\em infected segments} which start their life as {\em incubators}.
\begin{defi}[Incubators]\label{de:infectincuba}
Consider the set $I$ of nodes in the initial state which 
belong to an interval of nodes of length $w$ with less than
$\epsilon_{\ast}=w(1-\rho+\tau)/2$ many $\alpha$-nodes.
Let $I^{\ast}$ be the set of nodes whose neighborhood  contains a node in $I$.
 An incubator is a maximal interval of nodes
that is entirely contained in $I^{\ast}$.
\end{defi}

An interval of nodes is called {\em active} at a certain state if it contains 
an unhappy $\alpha$-node.
The {\em infected area} is the area that incubators generate by 
making additional $\alpha$-nodes unhappy. It is always expanding, and is defined formally as follows.

\begin{defi}[Infected segments]
Let $I$ be an incubator. 
At stage $0$ the infected segment $I_0$ corresponding to $I$ is
$I$ itself.
At the end of stage $s+1$, we first consider those $I_s$ which were active at the end of stage $s$ and which are still active. 
We consider these active infected segments in turn, starting at position 0 and moving clockwise.  For each such $I_s$   
we let $I_{s+1}=I_{s}\cup J$,
where $J$
consists of the nodes 
which do not already belong
to another active infected segment (by the time we consider $I_s$)  and
whose neighborhood contains an unhappy $\alpha$-node in 
$I_{s}$ at stage $s+1$. We then consider the remaining $I_s$ (i.e.\ those which are no longer active): 
for each such we  define $I_{s+1}=I_s-Q$ where $Q$ consists of the nodes in $I_s$
which now belong to an active infected segment.
\end{defi}

The {\em infected area} is the union of the infected segments.
The {\em fresh infected segment} corresponding to infected segment $I$
is $I-I_0$, i.e.\ consists of the nodes $I$ except the nodes in its incubator.
Hence a {\em fresh infected segment} consists of two growing intervals of nodes. 
The {\em fresh infected area} is the infected area except the nodes in the incubators.
The {\em interior} of a set of nodes $J$ consists of those nodes whose
neighbourhood is entirely contained in $J$.
The {\em boundary} of $J$ consists of the nodes in $J$ which are not in the interior. 
It is not hard to show that if $\tau+\rho<1$, the probability that a node belongs to an incubator is
$e^{-\Theta(w)}$. Hence with high probability  
 the number of incubators as well as the number of nodes belonging to incubators
 of the process $(n,w,\tau,\rho)$ is $n e^{-\Theta(w)}$.

Our goal is now to show that the number of unhappy 
$\alpha$-nodes remains suitably bounded throughout a significant part of the process. 
Formally, the main idea is to bound this number with a 
martingale. Intuitively though, why should the number of unhappy $\alpha$-nodes remain fairly small? 
At the start of the process the infected area is a very small proportion of the
entire ring. The vast majority of unhappy $\beta$-nodes occur outside
the infected area, while all unhappy $\alpha$-nodes are inside the infected area.
It follows that with high probability a 
swap will involve an $\alpha$-node in the infected
area and a $\beta$-node outside the infected area. A {\em bogus swap} is a swap is one that is not of this kind.

\begin{defi}[Bogus swaps]
A swap
which involves a $\beta$-node 
currently inside the infected area is called bogus. 
Given an infected segment $I$, a bogus swap in $I$ is a
swap that moves an $\alpha$-node into $I$.
\end{defi}

In the absence of bogus swaps, it is not hard to show that the $\alpha$-nodes 
in the infected area (except those in the incubators) are unhappy. This in turn can be used in order to show that 
the $\alpha$-nodes 
in the infected area (and so, the unhappy $\alpha$-nodes too)
are likely remain $\smo{n}$.
However there will be bogus swaps, and these can make certain $\alpha$-nodes in the infected area happy.

\begin{defi}[Anomalous nodes]
A node is called actively anomalous at some stage of the process if it is 
a happy $\alpha$-node in the interior of
the fresh infected area; it is called  anomalous if it has been actively anomalous in this or a previous stage.
Finally a node is called generally anomalous at some stage, if it is in the current infected area and has been
or will be actively anomalous at some later stage of the process.
\end{defi}

Clearly {\em actively anomalous} implies {\em anomalous}, which in turn implies {\em generally anomalous} 
(but not the other way around).
Let $\DD_s$ denote the number of anomalous nodes at stage $s$, and let $\bar{\DD}_s$ denote the number of
generally anomalous nodes at stage $s$.
A martingale argument will be used in order to show that as long as $\DD_s=\smo{n}$, the
$\alpha$-nodes in the infected area are likely to remain $\smo{n}$.
The definition of anomalous nodes and $\bar{\DD}_s$ may seem strange at this point, not least because $\bar{\DD}_s$ 
is not predictable at stage $s$. The reason that we introduce $\bar{\DD}_s$ is that $\DD_s$ is very hard to analyze,
and very hard to bound directly via a martingale (adapted to the stages of the process).  However 
it is possible to bound $\bar{\DD}_s$  via a martingale argument of a more general type (i.e.\ which is
not adapted to the stages of the process). Since $\DD_s\leq \bar{\DD}_s$, this suffices for our purposes.

We define additional global variables in Table \ref{ta:pranvarinfarfa2}.
By the definitions we have  $\DD_s\leq \DD_{s+1}$ and

\[
\textrm{
(a)\hspace{0.2cm}$\unhap_{\alpha}(s)\leq \mathbf{Z}_s$
\hspace{1cm}
(b)\hspace{0.2cm}$\unhap_{\beta}(s)\leq \GG_s+\YY_s$
\hspace{1cm}
(c)\hspace{0.2cm}$\GG_s\leq \unhap_{\beta}^{\ast}(s)$
\hspace{1cm}
(d)\hspace{0.2cm}$\expe{\CC} = n e^{-\Theta(w)}$
}
\]

Here (d) holds because of the likely total size of the incubators
and (c) holds because $\beta$-nodes outside 
the infected area are {\em very unhappy}.


\subsection{Probabilities in the infected area and anomalous nodes}
Recall that our current goal is  to show that the number of unhappy 
$\alpha$-nodes remains suitably bounded for a significant part of the process. 
The basic idea is that if the number of unhappy $\alpha$-nodes  increases sufficiently,
then the infected area must become quite large, and it becomes very likely that the next swap will involve an 
unhappy $\alpha$-node in the \emph{interior} of the infected
area.  We shall be able to  argue that there are good 
chances that the swap is not bogus. This means
that this $\alpha$-node will move outside the infected area and will become happy.
The anomalous nodes, however, present a difficulty with this line of argument. 
The eviction of the $\alpha$-node from the infected area
(and its replacement by a $\beta$-node) may produce 
more unhappy $\alpha$-nodes in its 
neighbourhood.
So it is not absolutely true
that the total number of unhappy $\alpha$-nodes will decrease.
In fact, as the simulations of Figure \ref{fig:inf_area} suggest, 
at the early stages of the process this
number is likely to increase slightly.

If we assume the absence of bogus swaps, 
then it is not hard to show that the nodes in the interior of the infected
area and outside the incubators have neighborhoods with proportion of $\alpha$-nodes 
well below $(2w+1)\tau$. In this case
it is straightforward
to employ a  martingale argument which shows that the number of  
$\alpha$-nodes in the infected area (hence also the total number of unhappy
$\alpha$-nodes) remains bounded with high probability throughout the process.
Indeed, in this case there will be no happy $\alpha$-nodes in the interior of the 
fresh infected area,
so (according to the argument we outlined above) 
the likely swap absolutely reduces the total number of unhappy
$\alpha$-nodes.

In the presence of bogus swaps, we will use a more sophisticated martingale argument to bound
the anomalous nodes. This can be used by another simpler martingale argument, in order to bound the
number of unhappy $\alpha$-nodes, at least up to some stopping time of the process
and with high probability.
This plan requires the calculation of certain probabilities.

\begin{lem}[Probability of a bogus swap]\label{le:probboguss}
At each stage $s+1$, the probability that the current swap will
be bogus is bounded above by $\YY_s/\GG_s$. 
\end{lem}
\begin{proof}
The number of pairs which can cause a bogus swap is bounded by $\unhap_{\alpha}(s)\cdot \YY_s$.
On the other hand, any unhappy $\alpha$-node can swap with a $\beta$-node outside the
infected area. Indeed, this is because the number of $\alpha$-nodes in
the neighbourhood of any $\beta$-node outside the infected area 
is at least $(2w+1)\tau$.
Hence there are at least $\unhap_{\alpha}(s)\cdot \GG_s$ pairs of nodes that can swap
at stage $s+1$. We can conclude that the probability of a bogus swap is bounded by
$\unhap_{\alpha}(s)\YY_s/\unhap_{\alpha}(s)\GG_s=\YY_s/\GG_s$.
\end{proof}

The calculation of the following probabilities is a first step towards our martingale argument.

\begin{lem}[Probabilities for $Z_s$]\label{le:zs1lzsprob}
The numbers
\[
\textrm{$\frac{\GG_s}{\unhap_{\alpha}(s)}\cdot 
\frac{\ZZ_s-\DD_s-2w\cdot \CC}{\GG_s+\YY_s}$ 
\hspace{1cm}and\hspace{1cm}
$2w\cdot \CC\cdot\frac{\GG_s+\YY_s}{\GG_s\cdot \unhap_{\alpha}(s)}$}
\]
are a lower bound for the probability that 
$\ZZ_{s+1}<\ZZ_s$
and an upper bound for the probability that 
$\ZZ_{s+1}>\ZZ_s$, respectively.
\end{lem}
\begin{proof}
The probability that
$\ZZ_{s+1}<\ZZ_s$ is at least as much as the probability that the swap
is not bogus and it
involves a node in
the interior of the infected area 
at stage $s+1$.
Indeed, in this case
the swap moves an $\alpha$-node from the interior of the infected area 
to outside the infected area, so $\ZZ_{s+1}=\ZZ_s-1$, 
because the length of the infected area remains the same. 
The unhappy $\alpha$-nodes of the infected area that 
cannot be part of such a swap are the ones that belong to the
boundary of the infected area,
so they are at most $2w\CC$ many. 
This means that there are at least
$\ZZ_s-\DD_s-2w\CC$ nodes of type $\alpha$  
which can be picked as part of a
swapping pair at stage $s+1$ such that $\ZZ_{s+1}-\ZZ_s$ is negative.
Note that each of these $\alpha$-nodes forms a swapping pair with any
$\beta$-node outside the infected area, since all such $\beta$-nodes are very unhappy.  
Therefore there are
at least $(\ZZ_s-\DD_s-2w\CC)\cdot \GG_s$ many  
swapping pairs which make $\ZZ_{s+1}-\ZZ_s$ negative. 
On the other hand, the total number of swapping pairs
are at most 
$(\GG_s+\YY_s)\cdot \unhap_{\alpha}(s)$ many. 
Hence 
\[
\GG_s\cdot 
\frac{\ZZ_s-\DD_s-2w\cdot \CC}{(\GG_s+\YY_s)\cdot \unhap_{\alpha}(s)}
\]
is a lower bound for the probability that $\ZZ_{s+1}<\ZZ_s$. 

For the second clause,
note that
$\ZZ_{s+1}>\ZZ_s$ can only happen in the 
case that the infected area expands at stage $s+1$.
This  can only occur
if the swapping pair involves an $\alpha$-node 
that belongs to the boundary of the infected area of stage $s$. 
There are at most 
$2w\CC$ such nodes so there are at most 
$2w\CC\cdot (\GG_s+\YY_s)$ 
swapping pairs that
can cause $\ZZ_{s+1}<\ZZ_s$.
Moreover there are at least $\GG_s\cdot \unhap_{\alpha}(s)$ possible swapping pairs
for stage $s+1$.
Hence 
\[
2w\cdot \frac{\CC\cdot (\GG_s+\YY_s)}{\GG_s\cdot \unhap_{\alpha}(s)}
\]
is an upper bound for the probability that $\ZZ_{s+1}>\ZZ_s$.
\end{proof}

We may now identify our first supermartingale. Note that the following fact is the reason why we defined
the anomalous nodes the way we did. The fact that $\DD_s$ is nondecreasing is a necessary part of the
following proof. 

\begin{lem}[Non-anomalous nodes in an infected segment]\label{le:zhsisma}\label{le:zhssecgisma}
The process $\ZZ^{\ast}_s:=\max\{\ZZ_s-\DD_s, 11w^2\cdot\CC\}$ is a supermartingale, 
for all $s<T_y$.
\end{lem}
\begin{proof}  At the end of stage $s$ (and given all information as to how the process has unfolded so far)
denote the probability that 
$\ZZ_{s+1}<\ZZ_s$ by $q$ and the probability that $\ZZ_{s+1}>\ZZ_s$
by $p$.  Let $E$ be the expected value of $\ZZ_{s+1}$. Now at  stage $s+1$ the infected area can expand by at most $w$ nodes.
Moreover, it is not possible that at stage $s+1$,
an $\alpha$-node which is not in the infected area of stage $s$  
is moved to a position in the infected area of stage $s+1$.
This is because all $\alpha$-nodes outside  
the infected area of stage $s$
are happy at stage $s$.
It follows that $\ZZ_{s+1}-\ZZ_s\leq w$ at each stage $s$. Therefore

\begin{equation}\label{eq:expecpqzsfs}
E\leq p\cdot (\ZZ_s+w) +
q\cdot (\ZZ_s-1) + (1-p-q)\cdot \ZZ_s = \ZZ_s + wp-q.
\end{equation}
By Lemma
\ref{le:zs1lzsprob}, in order to ensure that  $wp-q\leq 0$, it suffices that
\[
2w^2\cdot\CC\cdot\frac{\GG_s+\YY_s}{\GG_s\cdot \unhap_{\alpha}(s)}
\leq \frac{\GG_s}{\unhap_{\alpha}(s)}\cdot \frac{\ZZ_s-\DD_s-2w^2\cdot\CC}{\GG_s+\YY_s}
\hspace{0.5cm}
\textrm{so}
\hspace{0.5cm}
\ZZ_s\geq \DD_s + 2w^2\cdot\CC\cdot\left[1+ \left(1+\frac{\YY_s}{\GG_s}\right)^2\right]
.\]
Since $s<T_y$
the expression inside the parentheses in the latter inequality
is bounded above by 2. Hence for the condition $wp-q\leq 0$
it is sufficient that
$\ZZ_s\geq \DD_s +10w^2\cdot\CC$
for all $s<T_y$. So now we divide into two cases. If 
$\ZZ_s< \DD_s +10w^2\cdot\CC$ then $\ZZ_{s+1}^{\ast}=\ZZ_{s}^{\ast}= 11w^2\cdot \CC$. 
Otherwise, $E\leq \ZZ_s$ and the result follows from the fact that 
$\DD_s$ is non-decreasing.  \end{proof}

Now to get from $\ZZ_s^{\ast}$ to $\ZZ_s$, we need to bound $\DD_s$. 
Intuitively, we expect the proportion of the $\alpha$-nodes in
neighborhoods of nodes in the interior of the 
infected area to be rather low, e.g.\ considerably lower than the threshold $(2w+1)\tau$.
The following lemma gives a justification for such an expectation 
and is also the reason why we chose  
$\epsilon_{\ast}=(1-\tau-\rho)/2$ in the definition of incubators, Definition \ref{de:infectincuba}.
Here is an intuitive explanation of this fact.
Let us say that a node in the infected area which is not in the interior of the infected
area is in the boundary of the infected area.
A node in the boundary of the infected area can
see a node outside the infected area. The nodes in the complement of
the infected area have never seen unhappy $\alpha$-nodes, hence the
proportion of $\alpha$-nodes in their
semi-neighbourhoods can only increase. This means that one of the
semi-neighbourhoods of each node in the boundary of the infected area has not
been affected by $\alpha$-to-$\beta$ swaps.
The following lemma says that such a node can only be included in the interior
of the infected area if the semi-neighbourhood of it which has been affected by
$\alpha$-to-$\beta$ swaps, is affected by at least $\epsilon_{\ast}w$ many such swaps.
In other words, the expansion of the infected area requires a considerable number of
stages.
The particular statement refers to the case where the infection travels from
right to left. 
By symmetry, an analogous statement 
holds for the case where the infection travels the opposite direction.

\begin{lem}[Concentration of $\alpha$-to-$\beta$ swaps]\label{le:deephl}
Let $[a,d]$ be an interval of nodes in the initial state of the process,
and $\delta>0$ 
such that for each $u\in [a,d]$ the proportion of $\alpha$-nodes in  
each semi-neighbourhood of $u$ is at least $\tau+\delta$. 
Consider a time interval of the 
process where there have been no 
$\alpha$-to-$\beta$ swaps in
$[a-w, a)$. 
For each $u\in [a,d]$ and any stage in this interval,
if there is an unhappy $\alpha$-node in $[a, u]$  then there have been at least 
$2w\delta$ many $\alpha$-to-$\beta$ swaps 
in the right semi-neighbourhood of $u$ by that stage.
\end{lem}
\begin{proof}
Let $s$ be a stage of the process and suppose that
there have been no $\alpha$-to-$\beta$ swaps in $[a-w, a)$ by stage $s$.
Suppose that there is an unhappy $\alpha$-node in $[a-w, u]$ at stage $s$.
Then there must have been an unhappy
$\alpha$-node in $[u-w, u]$ at some stage $\leq s$. 
Consider the first such stage $t_0$ and let $v_0$ be the rightmost
$\alpha$-node in $(u-w, u]$ which became unhappy at stage $t_0$.
By our hypothesis, up to stage $t_0$ there has been no 
$\alpha$-to-$\beta$ swaps in
$(v_0,u]$. Hence all of the $\alpha$-to-$\beta$ swaps 
that occurred in the right semi-neighbourhood of $v_0$ are also in the
right semi-neighbourhood of $u$.
The proportion of the $\alpha$-nodes in the left 
semi-neighbourhood of $v_0$ is more than $\tau+\delta$.
Since $v_0$ is unhappy at $t_0$, the proportion 
of the $\alpha$-nodes in its neighbourhood is less than $\tau$. Hence
the proportion of the $\alpha$-nodes in 
its right semi-neighbourhood is at most $\tau-\delta$
at stage $t$. Hence by hypothesis, by stage $t$ at least 
$2w\delta$ many $\alpha$-to-$\beta$ swaps have occurred in the 
right semi-neighbourhood of $v$. By the above discussion,
these swaps have also occurred in the 
right semi-neighbourhood of $u$.
\end{proof}

According to the definition of incubators, this fact is relevant 
for $\delta=(1-\tau-\rho)/2$ and shows that the infected area expands reasonably slowly
in the stages of the  process $(n,w,\tau,\rho)$.
Indeed, the proportion of $\alpha$-nodes in
the neighbourhood of any node outside the infected area
at any particular stage is at least $\tau+(1-\tau-\rho)/2$.
This also shows that, in the absence of bogus swaps, all $\alpha$-nodes in the
interior of the fresh infected area are always unhappy (i.e.\ there are no anomalous nodes).
In the presence of bogus swaps this is no longer true, and this is why we have to work
in order to bound the spread of anomalous nodes. 

\subsection{Bounding the anomalous nodes}
Recall that $\DD_s$ denotes the number of anomalous nodes at stage $s$. In this section we
construct a martingale process which shows that $\DD_s$ is likely to be bounded appropriately,
throughout a significant part of the Schelling process. This argument requires us to consider
the random variables localized into the individual infected segments.
Recall the stopping times defined in
the second part of Table \ref{ta:validprglobbouatd}. 
We use $ \tau\rho n/(4w)$ rather than $ \tau\rho n/(3w)$ in the definition of $T_g$ 
so as to allow for the slight discrepancy which one might expect between $\rho$ and $\rho_{\ast}$.  

\begin{defi}[Stopping times]
Let  $T_g$ be the least stage such that
$\GG_{T_g}\leq \tau\rho n/(4w)$.
Define $T_y$ to be the first stage which is either  $ T_g$
or else such that $\YY_s>\GG_s$. 
Finally let $T_{\textup{mix}}$ be the first stage for which
$\mix <n(w+1)\tau\rho_{\ast}$. In all cases, if the stage 
described does not exist then we define the corresponding stopping time to be $\infty$. 
\end{defi}

Given an infected segment $I$, let   
$\bar{D}_s=\bar{D}_s(I)$ be the number of 
nodes in $I_s$ that will ever become anomalous, up to stage $T_g$.
This is a version of the {\em generally anomalous} nodes $\bar{\DD}_s$.
A stage is called an $I$-stage if a swap occurs involving a node from $I$.
\[
\textrm{If $(\nu(s))$ is an enumeration of the $I$-stages, let
$\bar{D}_{s}^{\ast}=\bar{D}_{\nu(sw^5)}$ and $I^{\ast}_s=I_{\nu(sw^5)}$.}
\]
We use $\ast$ as a superscript in other variables in the following, in order to indicate that they
are `jump processes' in the sense that they are not updated at every stage or even every $I$-stage 
of the Schelling process. For example, $D_{s}^{\ast}$ is only updated every $w^5$ many $I$-stages of the
Schelling process.

Recall that we may view the underlying probability space $\Omega$ as a tree, where the
nodes are states and  branchings correspond to state transitions. Let
$\Omega\wedge T_g$ denote the subspace restricted to the stages up to time $T_g$ (which may be infinite).
Normally we would say that an event
$\mathcal{A}\subseteq \Omega\wedge T_g$ is $I$-independent if
it did not impose any branching restrictions regarding the $I$-stages that occur in the reals in it.
We give a sightly more general definition which is more appropriate for the argument to follow.
An event $\mathcal{A}\subseteq \Omega\wedge T_g$ is called $I$-independent if for each 
$\beta\in \mathcal{A}$ and any $s$ such that
the transition from $\beta\restr_{s}$ to $\beta\restr_{s+1}$ occurs at an $I$-stage,
 $\beta\restr_{s}\ast S\in \mathcal{A}$ for every state that is obtained from $\beta\restr_s$ through a 
 non-bogus swap.
A filtration $\mathcal{A}_s\subseteq \mathcal{A}_{s+1} \subseteq \Omega\wedge T_g$ is 
called $I$-independent if  for each $s$ the event  $\mathcal{A}_s$ is $I$-independent.
Analogously, a process $(\JJ_s)$ on $\Omega$ is called
called $I$-independent if the natural filtration of it is $I$-independent.
Intuitively, a process $(\JJ_s)$ on the underlying probability space $\Omega$ is $I$-independent, 
if for each $s$, fixing the value of $\JJ_s$ does not impose any restriction on (i.e.\ is compatible with all) 
the transitions of the Schelling process from stage $s$ to stage $s+1$, that involve a non-bogus swap
and a node from $I$.  Here we use boldface font for $\JJ_s$ because this process 
will typically by global, in the sense that it involves information about
the process that is not restricted in the infected segment $I$.
In the following lemma we use $(\JJ^{\ast}_s)$ for the underlying $I$-independent global process
in order to indicate that it refers to the subsequence of stages $sw^5$ of the process, much like
$\bar{D}_{s}^{\ast}$.

\begin{lem}[$I$-supermartingale]\label{le:atomdma}
Given an infected interval $I$, the process
$\bar{D}_{s}^{\ast}-10ws$ is a supermartingale
relative to any $I$-independent process $\JJ^{\ast}_s$ to which  $\bar{D}_{s}^{\ast}$ is adapted.
\end{lem}
\begin{proof}
Given an $I$-independent process $\JJ^{\ast}_s$ such that 
$\bar{D}_{s}^{\ast}$ is adapted to $\JJ^{\ast}_s$
(i.e.\ $\bar{D}_{s}^{\ast}$ is a function of $\JJ^{\ast}_s$) 
it suffices to show that $\expec{\bar{D}_{s}^{\ast}}{J_{s-1}}\leq \bar{D}_{s-1}^{\ast} +10w$
for all $s$.
Let $\bar{D}_{s}^{\ast 0}$ be the number of nodes in the left semi-interval of $I_s$ 
that will ever become anomalous, up to stage $T_g$. Similarly let
$\bar{D}_{s}^{\ast 1}$ be the number of nodes in the right semi-interval of $I_s$ 
that will ever become anomalous, up to stage $T_g$. Clearly 
$\bar{D}_{s}^{\ast}=\bar{D}_{s}^{\ast 0}+\bar{D}_{s}^{\ast 1}$. So it suffices to show that
\[
\textrm{$\expec{\bar{D}_{s}^{\ast i}}{\JJ^{\ast}_{s-1}}\leq D_{s}^{\ast i} +5w$
 for each $i=0,1$.}
\]
 Similarly, let $I^{\ast 0}_s$ be the left interval of the fresh part of $I_s^{\ast}$
 and let $I^{\ast 1}_s$ be the right interval of the fresh part of $I_s^{\ast}$.
 Fix $i=0,1$ and
 set $H^{\ast i}_{s}=I^{\ast i}_{s}-I^{\ast i}_{s-1}$.
In order to bound the expectation of $D^{\ast i}_s$, we consider the following cases 
(where each case applies only if the one above it fails):
\begin{enumerate}
\item[(a)] $|H^{\ast i}_{s}|< 4w$;
\item[(b)] There are  bogus swaps in the $I$-stages $(s-1)w^5$ to $sw^5$;
\item[(c)] A happy $\alpha$-node appears in the interior of $J^{i}_{s}$ before the interior becomes all
$\beta$-nodes;
\item[(d)] The above $\beta$-firewall forms, but it shrinks by $4w$ at some later $I$-stage $tw^5$.
\item[(e)] Otherwise.
\end{enumerate}
We will show that all of these events yield small expectation (conditional on  $J_s$) 
on the number of happy $\alpha$-nodes
that will ever appear in the  interval $H^{\ast i}_{s}$ 
after $I$-stage $sw^5$ of the original process (in particular, the probabilities of
(b)-(d) are very small). We decide to accept $4w$ happy $\alpha$ nodes in $H^{i}_{s}$ 
as a desirable (i.e.\ not too high)
count. So, irrespective of likelihood, event (a) is desirable. 
Note that by Lemma \ref{le:deephl},
\begin{equation}\label{eq:prob234o}
\parbox{11cm}{in $w^5$ many $I$-stages $I$ cannot grow by more than $2w^5/(1-\tau-\rho)$.}
\end{equation}

Note that Lemma \ref{le:probboguss} also holds locally, by the same proof.
In other words, given an interval of nodes of length $\ell$, 
then the probability that at stage $s+1$ a bogus swap will occur involving a
$\beta$-node from the given interval is bounded above by $\ell/\GG_s$.
Since all stages are bounded by $T_g$, it follows that the probability (conditional on $\JJ^{\ast}_{s-1}$) 
of a bogus swap in an area of length $\ell$ is less than $4w\ell/n\tau\rho$.
Hence by \eqref{eq:prob234o}, event (b)  has probability 
(conditional on $\JJ^{\ast}_{s-1}$) upper bounded by $w^{2\cdot 5+2}/n$.
In this case we can bound the expectation trivially by $w^{3\cdot 5+2}/n=w^{17}/n$.

Now suppose that (a), (b) do not occur so that, by Lemma  \ref{le:deephl},
each subinterval of length $w$ in the interior of $H^{\ast i}_{s}$
has $\alpha$-proportion at most $\tau-\epsilon_{\ast}$ at $I$-stage $sw^5$, where 
recall that $(\epsilon_{\ast}=1-\tau-\rho)/2$.
In particular, all $\alpha$ nodes in 
the interior of
$H^{\ast i}_{s}$ are unhappy, 
and remain so unless $w\delta$ bogus swaps happen in $H^{i}_{s}$.
We wish to show that in this case
\begin{equation*}
\parbox{12cm}{event (c)  has probability  (conditional on $\JJ^{\ast}_{s-1}$) 
upper bounded by $(w^{3\cdot 5}/n)^{w\delta}$.} 
\end{equation*}
Indeed couple this process (conditional on $\JJ^{\ast}_{s-1}$, 
 where each stage is either a non-bogus swap in $H^{\ast i}_{s}$, or something else) 
 with a gambler's ruin process, where the gambler has $w^5/\delta$ chips and the house has $w\delta$ chips,
 and the ratio of the winning probabilities is less than $q=4w^{5+2}/n$ in favor of the house. Then we can estimate an upper
 bound the probability
 that $w\delta$ bogus swaps occur in $J^{i}_{s}$ before all the interior turns into a $\beta$-firewall.
According to the standard gambler's ruin result, this is
 \[
 \frac{1-q^{w^5/\delta}}{q^{-w\delta}-q^{w^5/\delta}}<
 \frac{1}{q^{-w\delta}}<(w^{3\cdot 5}/n)^{w\delta}
 \]
which is also a bound on the (conditional) probability of event (c). 
Now assume that (a)-(c) do not occur, and lets estimate an upper bound for the probability of (d).
Again, couple this process (conditionally on $\JJ^{\ast}_{s-1}$) 
with a biased random walk where a negative move corresponds to a bogus swap moving
something from the $w$ border (one or the other) of the firewall, and a positive move is swapping the 
$\alpha$-node at the edge with a $\beta$-node (other events are ignored). The ratio of the probabilities
is bounded above by $2w/(n\tau\rho/4w)$ which is bounded by $w^3/n$. Also note that a negative move
chips (at most) $w$ away from the firewall, while a positive move only contributes (at least) one node to 
the firewall. Then the probability that it will eat up $tw$ at any future time is bounded by 
$w\cdot (w^3/n)^{t-1}$. For $t=4$ we get
\begin{equation*}\label{eq:probbluefis}
\parbox{11cm}{event (d)  has probability (conditional on $\JJ^{\ast}_{s-1}$) upper bounded by $w^{10}/n^3$.} 
\end{equation*}
Then the expectation of the number of anomalous nodes that will ever appear in $H^{\ast i}_s$
is bounded by
\[
4w+ 2\frac{w^{2\cdot 5+2}}{n}\cdot w^5\leq 4w+\frac{w^{3\cdot 5+3}}{n}<5w.
\]
Finally under case (e) it is clear that the conditional expectation of
$D_{s}^{\ast i}$ is also bounded by $\bar{D}_{s-1}^{\ast i} +5w$. Considering all the different cases,
by the law of alternatives for conditional expectation we have that
$\expec{\bar{D}_{s}^{\ast i}}{\JJ^{\ast}_{s-1}}\leq \bar{D}_{s-1}^{\ast i} +5w$, which concludes the proof.
\end{proof}
Let $I_j, j<t$ be the infected segments (and $I_j[s]$ their state at stage $s$). 
Recall that $\bar{\DD}_s$ is the sum of all $\bar{D}_{s}(I_j)$, $j<t$.
In order to bound $\bar{\DD}_s$  we need to prove a global version of Lemma \ref{le:atomdma}.
An immediate obstacle is the asynchrony of the $I$-stages with respect to the various infected segments $I$.
We need to find a process $\LL_s$ relative to which $\bar{\DD}_s$ 
(or some `asynchronous' version $\hat{\DD}_s$ of it) is a supermartingale.

For each $j<t$ let $\tau_s(j)$ be the
stage where exactly $s\cdot w^5$ many $I_j$-stages have occurred. 
Also let $(\tau_i)$ be a monotone enumeration of the times $\{\tau_s(j)\ |\ j<t, s\in\Nat\}$.
Let $\lambda_s(j)$ 
be $\tau_m(j)$ for the maximum $m$ such that
$\tau_m(j)\leq s$.
Let $\hat{\DD}_s$ be the sum of all $\bar{D}_{\lambda_s(j)}(I_j)$, $j<t$.
The point of this definition is that $\hat{\DD}_s$ considers values of $\bar{D}(I_j), j<t$ at the last stage 
$\leq s$ where 
they completed a cycle (which happens at every $w^5$ many $I_j$-stages) and outputs their sum.
Define $\LL_s$ to be the vector containing the tuples 
$(\bar{D}_{\lambda_s(j)}, \lambda_s(j))$ for each $j<t$.
In this way,  the process $(\hat{\DD}_s)$ is adapted to $(\LL_s)$ (in other words, for each $s$, the value of 
$\hat{\DD}_s$ is a function of $\LL_s$).
Note that $\hat{\DD}_s$ remains constant in the intervals $[\tau_s, \tau_{s+1})$, just as
$\bar{D}_{\lambda_j(s)}$ remains constant in the interval $[\tau_s(j), \tau_{s+1}(j))$.

\begin{lem}\label{le:dhasmgas}
The process $\hat{\DD}_{\tau_{s}}-20ws$ is a supermartingale relative to the process
$\LL_{\tau_{s}}$.
\end{lem}
\begin{proof}
Using the law of alternatives for conditional expectation, it suffices to show that for each $s$ there is a
(finite) partition
$\mathcal{A}$ of events relative to $\LL_{\tau_{s}}$ such that for each $A\in\mathcal{A}$ we have
 \begin{equation}\label{eq:ddhola}
 \expeco{\hat{\DD}_{\tau_{s+1}}}{\LL_{\tau_{s}}}{A} \leq  \hat{\DD}_{\tau_{s}} +20w\ \ \ \ \textrm{for all $s$.}
\end{equation}

Each event $A\in\mathcal{A}$ describes which pair of infected intervals $I_j$ completes a cycle
at stage $\tau_{s+1}$, and the sequence of $I_j$-stages (for each of the two $j$) from 
$\lambda_j(\tau_s)$ to $\tau_{s+1}$.
Formally, event $A$ is a tuple one tuple $(m_0,m_1)$ where $m_i<t$, 
and for each $i=0,1$ an increasing sequence
of stages starting from $\lambda_{m_i}(\tau_s)$ 
and ending on the same number $a$. If $m_0=m_1$ then the two sequences should be the same.
The meaning of $A$ is that $\tau_{s+1}=a$ and infected intervals with indices $m_i$ are hit at stage $a$, with
the sequence of stages representing the 
exact stages from $\lambda_{m_i}(\tau_s)$ to $a$ where a swap occurs in $I_{m_i}$.
By the definition of $A$, this event is
$I_{m_i}$-independent for $i=0,1$. 
At stage $\tau_{s+1}$ of the process there must be exactly one tuple $(m_0,m_1)$ where 
$i<t$, such that the swap occurred in
$I_{m_0}$ and $I_{m_1}$. For each such event $A$ on $\LL_{\tau_s}$ we have
 \[
 \expeco{\hat{\DD}_{\tau_{s+1}}}{\LL_{\tau_{s}}}{A} =
 \sum_{j<t} \expeco{\hat{D}_{\tau_{s+1}}(I_j)}{\LL_{\tau_{s}}}{A}
 \]
 But for $j\neq m_0, m_1$ we have
$\expeco{\bar{D}_{\tau_{s+1}}(I_j)}{\LL_{\tau_{s}}}{A}=\bar{D}_{\tau_{s}}(I_j)$ and  
by Lemma \ref{le:atomdma} we have
\[
\textrm{$\expeco{\bar{D}_{\tau_{s+1}}(I_{m_i})}{\LL_{\tau_{s}}}{A}\leq 
\bar{D}_{\tau_{s}}(I_{m_i})+10w$  for $i=0,1$}
\]
since $A$ is 
$I_{m_i}$-independent for $i=0,1$. Therefore
 \eqref{eq:ddhola} holds for each of the events $A$. By the law of alternatives,
 and since there can be at most two infected segments that complete a cycle at stage $\tau_{s+1}$, 
 we get
 \[
 \expec{\hat{\DD}_{\tau_{s+1}}}{\LL_{\tau_{s}}} \leq  \hat{\DD}_{\tau_{s}} +20w\ \ \ \ \textrm{for all $s$.}
\]
Therefore $\hat{\DD}_{\tau_s}-10ws$ is a supermartingale adapted to $\LL_{\tau_{s}}$.
\end{proof}

\begin{coro}\label{le:evenbouds}
Let $a\in\Nat$. With probability $>1-1/a$, for all $s<T_g$ we have
$\DD_{s}< a+ \frac{20s}{w^{k-1}} +n e^{-O(w)}$.
\end{coro}
\begin{proof}
By Lemma \ref{le:dhasmgas} and the maximal inequality for
supermartingales, given any $a>1$,
with probability at least $1-1/a$ we have 
$\hat{\DD}_{\tau_s}< a+ 20ws$ for all $s<T_g$.
Since each stage can be an $I_j$-stage for at most two distinct $j<t$, we have
$|\{i\ |\ \tau_i\leq s\}|\leq 2s/w^5$.
Hence for each $a>1$ we have
\begin{equation}\label{eq:ladispe}
\textrm{with probability $>1-1/a$,}\ \ \ \ 
\hat{\DD}_{s}< a+ \frac{20s}{w^{4}} \ \ \ \textrm{for all $s<T_g$.}
\end{equation}
Also, note that at each stage $s$ 
we have $\bar{D}_s(I_j)\leq \bar{D}_{\lambda_j(s)} (I_j)+w^5$ for each $j<t$. 
Hence $\bar{\DD}_s\leq \hat{\DD}_s +nw^5 e^{-O(w)}$. Since we also have
$\DD_s\leq\bar{\DD}_s$ for all $s$,
the corollary follows from \eqref{eq:ladispe}.
\end{proof}

\subsection{Bounding the arrival time to a safe state}
By Lemma \ref{le:zhsisma} and
Corollary \ref{le:evenbouds} we have the desired bound on $\ZZ_s$.

\begin{coro}\label{coro:boundzhf}
Let $a\in\Nat$. With probability $>1-1/a$, for all $s<T_y$ we have
$\ZZ_{s}< a+ \frac{20s}{w^{4}} +n e^{-O(w)}$.
\end{coro}

By Lemma \ref{le:probboguss} and Corollaries 
\ref{le:redvsunhapgr} and
\ref{coro:boundzhf} we have the following

\begin{coro}\label{coro:boundzadditif}
Let $a\in\Nat$. With probability $>1-1/a$, for all $s<\min\{T_y,n\}$ we have
$w^3 \cdot \ZZ_{s}=\smo{n}$ and $\pp_s=\smo{1}$.
\end{coro}

The following result is the technical basis for the result
 that with high probability a safe state will be reached
(at some finite stage). It says that, with high probability
the stopping times $T_y, T_g$ are equal and
 are bounded by $n$. 
 
\begin{lem}[Stopping times]\label{le:stoptimefin}
With probability $1-\smo{1}$ we have $T_y=T_g<n$.
\end{lem}
\begin{proof}
Let $\epsilon>0$ such that $1-\epsilon> \rho+1/8$. 
By Hoeffding's inequality for Bernoulli trials 
we may consider $n$ large enough such that
the probability that $\GG_0>(\rho+1/8) n$ is less than $\epsilon/4$.
Recall that $\pp_s$ is the
probability of a bogus swap at stage $s+1$.
Suppose that $w$ is large enough such that with probability at least $1-\epsilon/4$
\begin{itemize}
\item[(a)] $w^2\CC<n\tau\rho/32$;
\item[(b)] $w\ZZ_{s}< n\tau\rho\cdot  (1-\tau)/32$ for each $s\leq \min\{T_y, n\}$;
\item[(c)] $\pp_s<\epsilon^3/16$ for each $s\leq \min\{T_y, n\}$.
\end{itemize}
Clause (a) can be ensured 
by Lemma \ref{prop:blacksitesbou}. 
Clause (b) 
can be ensured 
by Corollary \ref{coro:boundzhf}.
Clause (c) can be ensured 
by Corollary \ref{coro:boundzadditif}.
First, for a contradiction, assume that 
$T_y<T_g$. Then  $\GG_{T_y}< \YY_{T_y}$.
By Corollary \ref{coro:boundzadditif} and since (by definition) 
$T_y\leq T_g$ we have
\[
n\cdot \tau\rho/4 < w\cdot \GG_{T_y}< w\cdot 
\frac{\ZZ_{T_y}}{1-\tau}+w^2\cdot \CC
< \frac{n\tau\rho(1-\tau)}{32} + \frac{n\tau\rho}{32}< \frac{n\tau\rho}{16}
\]
which is the required contradiction.
Hence with probability $>1-\epsilon/4$ we have
$T_y=T_g$.
Second, we show that with probability at least $1-\epsilon/2$
we have $T_y< n$. By clause (c) above, 
\begin{equation}\label{eq:bogusprobast}
\textrm{with probability at least $1-\epsilon/4$,
at all stages $s<\min\{T_y,n\}$  
we have $\pp_s <\epsilon^2/4$.}
\end{equation}
By \eqref{eq:bogusprobast}, with probability at least $1-\epsilon/4$,
the expectation of the number of bogus swaps
that have occurred by stage $T_y$ is $<\epsilon^2\cdot T_y/4$. 
Hence, conditionally on the event that $\pp_s <\epsilon^2/4$ for all stages
$s\leq \min\{T_y, n\}$,
 the  
probability that by stage $\min\{T_y, n\}$
more than $\epsilon n$ bogus swaps have occurred
is less than $\epsilon/4$. 
Hence the unconditional probability that 
by stage $\min\{T_y, n\}$ at most 
$\epsilon T_y$ bogus swaps have occurred
is at least $(1-\epsilon/4)^2>1-\epsilon/2$.

We conclude the argument. We have established that the probability 
of the event $T_y<T_g$ or $\GG_0>n(\rho+1/8)$ is bounded by $\epsilon/2$.
It remains to show that outside this rare event, $T_y<n$.
Since every non-bogus swap reduces $\GG_s$ by (at least) 1, and 
$\GG_0\leq n(\rho+1/8)$, $\rho<0.5$,
with probability at least
$1-\epsilon/2$ we have
\[
\GG_{T_y}\leq \GG_0-(1-\epsilon) T_y
\Rightarrow
T_y\leq (\GG_0-\GG_{T_y})/(1-\epsilon)\leq n(\rho+1/8)/(1-\epsilon)<n
\]
which shows that
$T_z=T_g<n$ with probability at least $1-\epsilon$.
\end{proof}
\begin{coro} [Safe state arrival]
Suppose that $\tau+\rho<1$. Then with high probability
the process $(n,w,\tau,\rho)$  reaches a safe state, and then complete segregation.
\end{coro}
\begin{proof}
Let $\epsilon>0$. By the law of large numbers,
with probability at least $1-\epsilon/4$ and sufficiently large $n$ 
we have $3\rho<4\rho_{\ast}$. 
Pick $w,n$ large enough such that
\begin{itemize}
\item[(a)] $T_g=T_y<n$ with probability $>1-\epsilon/4$;
\item[(b)] $2w\ZZ_{T_g}/(1-\tau)<n\tau\rho/4$ with probability $>1-\epsilon/4$;
\item[(c)]  $\CC(w+1)<n\tau\rho/(4w)$ with probability $>1-\epsilon/4$.
\end{itemize}
Clause (a) can be ensured by 
Lemma \ref{le:stoptimefin} and clause (b)
can be ensured by
Corollary \ref{coro:boundzhf}.
Clause (c) can be ensured by
Lemma 
\ref{prop:blacksitesbou}.
By the definition of $T_g$,
$\GG_{T_g}\leq \tau\rho n/(4w)$.
Hence by
Corollary \ref{coro:boundzadditif}
we have
\[
\mathtt{U}_{T_g}\leq 
\GG_{T_g} + \YY_{T_g} + \ZZ_{T_g}\leq
\CC(w+1) + \GG_{T_g} + 2\ZZ_{T_g} /(1-\tau)\leq
\frac{3\rho}{4}\cdot\frac{n\tau}{w}<
\frac{n\tau\rho_{\ast}}{w}
\]
with probability $>1-\epsilon$.
But  $\mix\leq \mathtt{U}\cdot w(w+1)$
so the mixing index at stage $T_g$ is less than
$n\tau\rho_{\ast}\cdot(w+1)$. In other words, $T_{\textrm{mix}}\leq T_g$, so
by Proposition \ref{prop:maxmixindor} the process at stage $T_g$ is
in a safe state, with probability more than $1-\epsilon$.
Hence by Corollary \ref{coro:inevcompsegd1}, the process will arrive
to complete segregation 
with probability at least $1-\epsilon$.
\end{proof}

\end{document}